\documentclass[a4paper,UKenglish]{lipics}
\usepackage{mathtools}
\usepackage{color}
\usepackage{subfig}
\usepackage{xspace}
\usepackage{todonotes}
\usepackage{wrapfig}
\usepackage{dsfont}
\graphicspath{{img/}}
\newtheorem{conjecture}{Conjecture}

\newcommand{\temptitle}{On Planar Greedy Drawings of 3-Connected Planar Graphs}
\title{\temptitle\footnote{This article reports on work supported by the U.S.~Defense Advanced
Research Projects Agency (DARPA) under agreement no.~AFRL FA8750-15-2-0092. The views expressed are those of the authors and do not reflect the official policy or position of the Department of Defense or the U.S.~Government. This research was also partially supported by the Natural Sciences and Engineering Research Council of Canada (NSERC), by MIUR Project ``MODE'' under PRIN 20157EFM5C, and by H2020-MSCA-RISE project 73499 -- ``CONNECT''.}}
\titlerunning{\temptitle} 
\author[1]{Giordano {Da Lozzo}}
\author[2]{Anthony D'Angelo}
\author[3]{Fabrizio Frati}
\affil[1]{Department of Computer Science, University of California, Irvine, CA, USA
	\\ \texttt{gdalozzo@uci.edu}}
\affil[2]{School of Computer Science, Carleton University, Ottawa, Canada \\\texttt{anthonydangelo@cmail.carleton.ca}}
\affil[3]{Department of Engineering, Roma Tre University, Rome, Italy\\
  \texttt{frati@dia.uniroma3.it}}
\authorrunning{G. {Da Lozzo}, A. {D'Angelo}, and F. Frati} 

\Copyright{Giordano {Da Lozzo}, Anthony {D'Angelo}, and Fabrizio Frati}


\begin{document}

\maketitle

\begin{abstract}
A graph drawing is {\em greedy} if, for every ordered pair of vertices $(x,y)$, there is a path from $x$ to $y$ such that the Euclidean distance to $y$ decreases monotonically at every vertex of the path. Greedy drawings support a simple geometric routing scheme, in which any node that has to send a packet to a destination ``greedily'' forwards the packet to any neighbor that is closer to the destination than itself, according to the Euclidean distance in the drawing. In a greedy drawing such a neighbor always exists and hence this routing scheme is guaranteed to succeed.  

In 2004 Papadimitriou and Ratajczak stated two conjectures related to greedy drawings. The {\em greedy embedding conjecture} states that every $3$-connected planar graph admits a greedy drawing. The {\em convex greedy embedding conjecture} asserts that every $3$-connected planar graph admits a planar greedy drawing in which the faces are delimited by convex polygons. In 2008 the greedy embedding conjecture was settled in the positive by Leighton and Moitra. 

In this paper we prove that every $3$-connected planar graph admits a {\em planar} greedy drawing. Apart from being a strengthening of Leighton and Moitra's result, this theorem constitutes a natural intermediate step towards a proof of the convex greedy embedding conjecture. 
\end{abstract}

\section{Introduction}
{\em Geographic routing} is a family of routing protocols for {\em ad-hoc networks}, which are networks with no fixed infrastructure -- such as routers or access points -- and with dynamic topology~\cite{frs-rwsn-09,sm-ahmn-04,t-ahmn-02}. In a geographic routing scheme each node of the network actively sends, forwards, and receives packets; further, it does so by only relying on the knowledge of its own geographic coordinates, of those of its neighbors, and of those of the packet destination.  {\em Greedy routing} -- originally called {\em Cartesian routing}~\cite{f-rap-87} -- is the simplest and most renowned geographic routing scheme. In this protocol, a node that has to send a packet simply forwards it to any neighbor that is closer -- according to the Euclidean distance -- to the destination than itself. The greedy routing scheme might fail to deliver packets because of the presence of a {\em void} in the network; this is a node with no neighbor closer to the destination than itself. For this reason, several variations of the greedy routing scheme have been proposed; see, e.g.,~\cite{bmsu-rgdhwn-01,ksu-crgn-99,kwz-aagrhsn-08}.

Apart from its failure in the presence of voids, the greedy routing protocol has two disadvantages which limit its applicability. First, in order for the protocol to work, each node of the network has to be equipped with a GPS, which might be expensive and might consume excessive energy. Second, two nodes that are close geographically might be unable to communicate with each other because of the presence of topological obstructions. Rao {\em et al.}~\cite{rpss-grwli-03} introduced the following brilliant idea for extending the applicability of geographic routing in order to overcome the above issues. Suppose that a network topology is known; then one can assign {\em virtual coordinates} to the nodes and use these coordinates instead of the geographic locations of the nodes in the greedy routing protocol. The virtual coordinates can then be chosen so that the greedy routing protocol is guaranteed to succeed. 

Computing the virtual coordinate assignment for the nodes of a network corresponds to the following graph drawing problem: Given a graph $G$, construct a {\em greedy drawing} of $G$, that is a drawing in the plane such that, for any ordered pair of vertices $(x,y)$, there exists a neighbor of $x$ in $G$ that is closer -- in terms of Euclidean distance -- to $y$ than $x$. Equivalently, a greedy drawing of $G$ is such that, for any ordered pair of vertices $(x,y)$, there exists a {\em distance-decreasing} path from $x$ to $y$, that is, a path $(u_1,u_2,\dots,u_m)$ in $G$ such that $x=u_1$, $y=u_m$, and the Euclidean distance between $u_{i+1}$ and $u_m$ is smaller than the one between $u_{i}$ and $u_m$, for any $i=1,2,\dots,m-2$.

Greedy drawings experienced a dramatical surge of popularity in the theory community in 2004, when Papadimitriou and Ratajczak~\cite{pr-crgr-04} proposed the following two conjectures about greedy drawings of $3$-connected planar graphs.\footnote{The convex greedy embedding conjecture has not been stated in the journal version~\cite{pr-crgr-05} of Papadimitriou and Ratajczak paper~\cite{pr-crgr-04}.}

\begin{conjecture} ({\em Greedy embedding conjecture}) 
Every $3$-connected planar graph admits a greedy drawing.
\end{conjecture} 	

\begin{conjecture} ({\em Convex greedy embedding conjecture}) 
	Every $3$-connected planar graph admits a convex greedy drawing.
\end{conjecture} 	

Papadimitriou and Ratajczak~\cite{pr-crgr-04,pr-crgr-05} provided several reasons why $3$-connected planar graphs are central to the study of greedy drawings. First, there exist non-$3$-connected planar graphs and $3$-connected non-planar graphs that do not admit any greedy drawing. Thus, the $3$-connected planar graphs form the largest class of graphs that might admit a greedy drawing, in a sense. Second, all the graphs with no $K_{3,3}$-minor admit a $3$-connected planar spanning graph, hence they admit a greedy drawing, provided the truth of the greedy embedding conjecture. Third, the preliminary study of Papadimitriou and Ratajczak~\cite{pr-crgr-04,pr-crgr-05} provided evidence for the mathematical depth of their conjectures.

In 2008 Leighton and Moitra~\cite{lm-srgems-10,ml-srgems-08} settled the greedy embedding conjecture in the affirmative; the same result was established (independently and slightly later) by Angelini et al.~\cite{afg-acgdt-09,afg-acgdt-10}. In this paper we show the following result. 

\begin{theorem} \label{th:main}
	Every $3$-connected planar graph admits a planar greedy drawing.
\end{theorem}

Given a $3$-connected planar graph $G$, both the algorithm by Leighton and Moitra~\cite{lm-srgems-10,ml-srgems-08} and the one by Angelini et al.~\cite{afg-acgdt-09,afg-acgdt-10} find a certain spanning subgraph $S$ of $G$ and construct a (planar) greedy drawing of $S$; then they embed the edges of $G$ not in $S$ as straight-line segments obtaining a, in general, {\em non-planar} greedy drawing of $G$. Thus, Theorem~\ref{th:main} strengthens Leighton and Moitra's and Angelini et al.'s results. Furthermore, {\em convex} drawings, in which all the faces are delimited by convex polygons, are planar, hence Theorem~\ref{th:main} provides a natural step towards a proof of the convex greedy embedding conjecture.  

Our proof employs a structural decomposition for $3$-connected planar graphs which finds its origins in a paper by Chen and Yu~\cite{cy-lc3g-02}. This decomposition actually works for a super-class of the $3$-connected planar graphs known as {\em strong circuit graphs}. We construct a planar greedy drawing of a given strong circuit graph $G$ recursively: We apply the structural decomposition to $G$ in order to obtain some smaller strong circuit graphs, we recursively construct planar greedy drawings for them, and then we suitably arrange these drawings together to get a planar greedy drawing of $G$. For this arrangement to be feasible, we need to ensure that the drawings we construct satisfy some coercive geometric requirements; these are described in the main technical theorem of the paper -- Theorem~\ref{th:main-aux}.

{\bf Related results.} Planar greedy drawings always exist for {\em maximal planar graphs}~\cite{d-gdt-10}. Further, every planar graph $G$ with a {\em Hamiltonian path} $P=(u_1,u_2,\dots,u_n)$ has a planar greedy drawing. Namely, construct a planar straight-line drawing $\Gamma$ of $G$ such that $y(u_1)<y(u_2)<\dots<y(u_n)$; such a drawing always exists~\cite{dt-apr-88}; scale $\Gamma$ down horizontally, so that $P$ is ``almost vertical''. Then, for any $1\leq i<j\leq n$, the paths $(u_i,u_{i+1}\dots,u_j)$ and $(u_j,u_{j-1}\dots,u_i)$ are distance-decreasing. A characterization of the {\em trees} that admit a (planar) greedy drawing is known~\cite{np-egdt-13}; indeed, a greedy drawing of a tree is always planar~\cite{adf-sgd-12}. 


Algorithms have been designed to construct {\em succinct} greedy drawings, in which the vertex coordinates are represented with a polylogarithmic number of bits~\cite{eg-sggruhg-11,gs-sggrep-09,hz-osgdptt-14}; this has been achieved by allowing the embedding space to be different from the Euclidean plane or the metric to be different from the Euclidean distance. 

Planar graph drawings have been studied in which paths between pairs of vertices are required to exist satisfying properties different from being distance-decreasing. Consider a path $P=(u_1,u_2,\dots,u_m)$ in a graph drawing. We say that $P$ is {\em self-approaching}~\cite{acglp-sag-12,npr-osaicd-16} if, for any three points $a,b,c$ in this order along $P$ from $u_1$ to $u_m$, the Euclidean distance between $a$ and $c$ is larger than the one between $b$ and $c$ -- then a self-approaching path is also distance-decreasing. We say that $P$ is {\em increasing-chord}~\cite{acglp-sag-12,dfg-icgps-15,npr-osaicd-16} if it is self-approaching in both directions. We say that $P$ is {\em strongly monotone}~\cite{acdfp-mdg-12,fikkms-smdpg-16,kssw-omdt-14} if the orthogonal projections of the vertices of $P$ on the line $\ell$ through $u_1$ and $u_m$ appear in the order $u_1,u_2,\dots,u_m$. We explicitly mention~\cite{fikkms-smdpg-16} the recent proof that every $3$-connected planar graph admits a planar drawing in which every pair of vertices is connected by a strongly monotone path.

\section{Preliminaries}

In this section we introduce some preliminaries. For a graph $G$, we denote by $V(G)$ and $E(G)$ its vertex and edge sets, respectively.

{\bf Subgraphs and connectivity.} Let $G$ be a graph and $U\subseteq V(G)$; we denote by $G-U$ the graph obtained from $G$ by removing the vertices in $U$ and their incident edges. Further, if $e\in E(G)$, we denote by $G-e$ the graph obtained from $G$ by removing the edge $e$. Let $H$ be a subgraph of $G$. An {\em $H$-bridge} $B$ of $G$ is either an edge of $G$ not in $H$ with both the end-vertices in $H$ (then we say that $B$ is a {\em trivial $H$-bridge}), or a connected component of $G- V(H)$ together with the edges from that component to the vertices in $V(H)$ (then we say that $B$ is a {\em non-trivial $H$-bridge}); the vertices in $V(H)\cap V(B)$ are the {\em attachments} of $B$ in $H$. Further, for a vertex $v\in V(G)-V(H)$, we denote by $H\cup \{v\}$ the subgraph of $G$ composed of $H$ and of the isolated vertex $v$. Further, 

A {\em vertex $k$-cut} (in the following simply called {\em $k$-cut}) in a connected graph $G$ is a set of $k$ vertices whose removal disconnects $G$. For $k\geq 2$, a connected graph is {\em $k$-connected} if it has no $(k-1)$-cut. A {\em $k$-connected component} of a graph $G$ is a maximal (with respect to both vertices and edges) $k$-connected subgraph of $G$. Given a $2$-cut $\{a,b\}$ in a $2$-connected graph $G$, an {\em $\{a,b\}$-component} is either the edge $ab$ (then we say that the $\{a,b\}$-component is {\em trivial}) or a subgraph of $G$ induced by $a$, $b$, and the vertices of a connected component of $G-\{a,b\}$ (then we say that the $\{a,b\}$-component is {\em non-trivial}).

{\bf Plane graphs and embeddings.} A drawing of a graph is {\em planar} if no two edges intersect except at common end-vertices. A {\em plane graph} is a planar graph together with a plane embedding; a {\em plane embedding} of a connected planar graph $G$ is an equivalence class of planar drawings of $G$, where two drawings $\Gamma_1$ and $\Gamma_2$ are {\em equivalent} if: (i) for each $v\in V(G)$, the clockwise order of the edges incident to $v$ coincides in $\Gamma_1$ and in $\Gamma_2$; and (ii) the clockwise order of the edges composing the walks delimiting the outer faces of $\Gamma_1$ and $\Gamma_2$ is the same. When we talk about a planar drawing of a plane graph $G$, we always mean that it respects the plane embedding of $G$. We assume that any subgraph $H$ of $G$ is associated with the plane embedding obtained from the one of $G$ by deleting the vertices and edges not in $H$. In a plane graph $G$ a vertex is {\em external} or {\em internal} depending on whether it is or it is not incident to the outer face of $G$, respectively. 

Refer to Fig.~\ref{fi:shape}. For two external vertices $u$ and $v$ of a $2$-connected plane graph $G$, let $\tau_{uv}(G)$ and $\beta_{uv}(G)$ be the paths composed of the vertices and edges encountered when walking along the boundary of the outer face of $G$ in clockwise and counter-clockwise direction from $u$ to $v$, respectively. Note that $\tau_{uv}(G)$ and $\beta_{vu}(G)$ have the same vertices and edges, however in reverse linear orders. 

\begin{figure}[htb]
	\centering
	\includegraphics[scale=0.6]{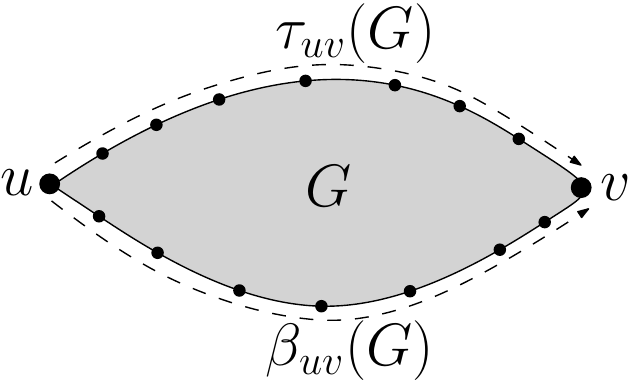}
	\caption{The paths $\tau_{uv}(G)$ and $\beta_{uv}(G)$ in a $2$-connected plane graph $G$.}
	\label{fi:shape}
\end{figure}

{\bf Geometry.} In this paper every angle is measured in radians, even when not explicitly stated. The {\em slope of a half-line} $\ell$ is defined as follows. Denote by $p$ the starting point of $\ell$ and let $\ell'$ be the vertical half-line starting at $p$ and directed towards decreasing $y$-coordinates. Then the slope of $\ell$ is the angle spanned by a counter-clockwise rotation around $p$ bringing $\ell'$ to coincide with $\ell$, minus $\frac{\pi}{2}$.  Note that, because of this definition, the slope of any half-line is assumed to be between -$\frac{\pi}{2}$ (included) and $\frac{3\pi}{2}$ (excluded); in the following, there will be very few exceptions to this assumption, which will be however evident from the text. Every angle expressed as $\arctan(\cdot)$ is assumed to be between -$\frac{\pi}{2}$ and $\frac{\pi}{2}$. We define the {\em slope of an edge} $uv$ in a graph drawing as the slope of the half-line from $u$ through $v$. Note that the slope of an edge $uv$ is equal to the slope of the edge $vu$ plus or minus $\pi$. For a directed line $\ell$, we let its slope be equal to the slope of any half-line starting at a point of $\ell$ and directed as $\ell$. We denote by $\Delta pqr$ a triangle with vertices $p,q,r$, and we denote by $\measuredangle pqr$ the angle of $\Delta pqr$ incident to $q$; note that $\measuredangle pqr$ is between $0$ and $\pi$. 

Let $\Gamma$ be a drawing of a graph $G$ and let $u,v$ be vertices in $V(G)$. We denote by $d(\Gamma,uv)$ the Euclidean distance between $u$ and $v$ in $\Gamma$. We also denote by $d_H(\Gamma,uv)$ the horizontal distance between $u$ and $v$ in $\Gamma$, that is, the absolute value of the difference between the $x$-coordinates of $u$ and $v$ in $\Gamma$; the vertical distance $d_V(\Gamma,uv)$  between $u$ and $v$ in $\Gamma$ is defined analogously. With a slight abuse of notation, we will use $d(\Gamma,pq)$, $d_H(\Gamma,pq)$, and $d_V(\Gamma,pq)$ even if $p$ and $q$ are points in the plane (and not necessarily vertices of $G$). A drawing of a graph is a {\em straight-line} drawing if each edge is represented by a straight-line segment.

The following lemma argues that the planarity and the greediness of a drawing are not lost as a consequence of any sufficiently small perturbation of the vertex positions.

\begin{lemma} \label{le:perturbation-preserves-planarity} \label{le:perturbation-preserves-greedy}
Let $\Gamma$ be a planar straight-line drawing of a graph $G$. There exists a value $\varepsilon^*_{\Gamma}>0$ such that the following holds true. Let $\Gamma'$ be any straight-line drawing in which, for every vertex $z\in V(G)$, the Euclidean distance between the positions of $z$ in $\Gamma$ and $\Gamma'$ is at most $\varepsilon^*_{\Gamma}$; then $\Gamma'$ is planar and any path which is distance-decreasing in $\Gamma$ is also distance-decreasing in $\Gamma'$.
\end{lemma} 

\begin{proof}
Let $\delta$ be the minimum Euclidean distance in $\Gamma$ between any two vertices, or between any vertex and any non-incident edge, or between any two non-adjacent edges, where the Euclidean distance between a point $p$ and a straight-line segment $s$ is the minimum Euclidean distance between $p$ and any point of $s$, and the Euclidean distance between two straight-line segments $s_1$ and $s_2$ is the minimum Euclidean distance between any point of $s_1$ and any point of $s_2$. Note that $\delta>0$, since $\Gamma$ is planar. Further, let $\gamma=\min\{d(\Gamma,uz)-d(\Gamma,vz)\}$, where the minimum is taken over all the ordered triples $(u,v,z)$ of distinct vertices of $G$ such that $d(\Gamma,uz)>d(\Gamma,vz)$. Note that $\gamma>0$. Set $\varepsilon^*_{\Gamma}=\min\{\frac{\delta}{3},\frac{\gamma}{5}\}$. Note that $\varepsilon^*_{\Gamma}>0$.

Consider any straight-line drawing $\Gamma'$ of $G$ in which, for each vertex $z\in V(G)$, the Euclidean distance between the positions of $z$ in $\Gamma$ and $\Gamma'$ is at most $\varepsilon^*_{\Gamma}$. 

We prove that $\Gamma'$ is planar. In order to do that, we exploit the following observation. For any point $p'$ that belongs to the straight-line segment $s'$ representing an edge $e$ in $\Gamma'$, there exists a point $p$ whose distance from $p'$ is at most $\varepsilon^*_{\Gamma}$ and that belongs to the straight-line segment $s$ representing $e$ in $\Gamma$. This is because $s'$ is contained in the convex hull of the two disks with radius $\varepsilon^*_{\Gamma}$ centered at the end-points of $s$ or, equivalently, in the region which is the Minkowski sum of $s$ with a disk with radius $\varepsilon^*_{\Gamma}$. Now suppose, for a contradiction, that in $\Gamma'$ two distinct vertices $v_1$ and $v_2$ coincide at a point $p'$, or an edge $e$ overlaps a non-incident vertex $v$  at a point $p'$, or two non-adjacent edges $e_1$ and $e_2$ cross at a point $p'$. Then there exist two points $p_1$ and $p_2$ in $\Gamma$ that are at distance at most $\varepsilon^*_{\Gamma}$ from $p'$ and hence at most $2\varepsilon^*_{\Gamma}$ from each other and such that $v_1$ and $v_2$ are placed at $p_1$ and $p_2$ in $\Gamma$, or such that $v$ is placed at $p_1$ and $p_2$ belongs to the straight-line segment representing $e$ in $\Gamma$, or such that $p_1$ and $p_2$ belong to the straight-line segments representing $e_1$ and $e_2$ in $\Gamma$, respectively. However, $2\varepsilon^*_{\Gamma}\leq \frac{2\delta}{3}<\delta$, which contradicts the definition of $\delta$. 

We prove that any path $P=(u_1,u_2,\dots,u_m)$ which is distance-decreasing in $\Gamma$ is also distance-decreasing in $\Gamma'$. Since $P$ is distance-decreasing, we have that $d(\Gamma,u_iu_m)>d(\Gamma,u_{i+1}u_m)$, for every $i=1,2\dots,m-2$. Since the Euclidean distance between the positions of any vertex $z\in V(G)$ in $\Gamma$ and $\Gamma'$ is at most $\varepsilon^*_{\Gamma}$, for any $i=1,\dots,m-2$, we have that $d(\Gamma',u_iu_m)\geq d(\Gamma,u_iu_m)-2\varepsilon^*_{\Gamma}$ and that $d(\Gamma',u_{i+1}u_m)\leq d(\Gamma,u_{i+1}u_m)+2\varepsilon^*_{\Gamma}$. It follows that $d(\Gamma',u_iu_m)-d(\Gamma',u_{i+1}u_m)\geq d(\Gamma,u_iu_m)-d(\Gamma,u_{i+1}u_m)-4\varepsilon^*_{\Gamma}\geq \frac{d(\Gamma,u_iu_m)-d(\Gamma,u_{i+1}u_m)}{5}>0$. Hence, $d(\Gamma',u_iu_m)>d(\Gamma',u_{i+1}u_m)$ for $i=1,\dots,m-2$. It follows that $P$ is distance-decreasing in $\Gamma'$. 
\end{proof}

We conclude this section with a technical lemma we are going to exploit heavily in the next section. Refer to Fig.~\ref{fi:three-lines1}. 

\begin{figure}[htb]
	\centering
	\subfloat[]{
		\includegraphics[scale=0.7]{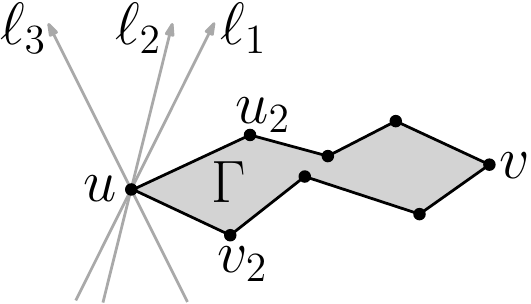}
		\label{fi:three-lines1}
	}\hfil
	\subfloat[]{
		\includegraphics[scale=0.6]{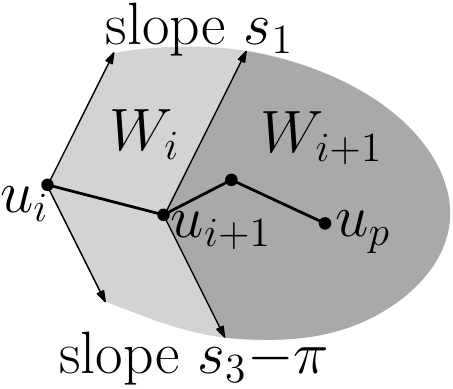}
		\label{fi:three-lines2}
	}\hfil
	\caption{(a) Illustration for the statement of Lemma~\ref{le:same-side}. (b) Illustration for the proof of Lemma~\ref{le:same-side}.}
	\label{fi:three-lines}
\end{figure}

\begin{lemma} \label{le:same-side}
Let $G$ be a $2$-connected plane graph whose outer face consists of two paths $(u=u_1,u_2,\dots,u_p=v)$ and $(u=v_1,v_2,\dots,v_q=v)$. Let $\ell_1$, $\ell_2$, and $\ell_3$ be three directed lines that pass through a point $p_u$ and that have slopes $s_1$, $s_2$, and $s_3$, respectively, where $0<s_1\leq s_2\leq s_3<\pi$. Let $\Gamma$ be a planar drawing of $G$ such that $u$ lies at $p_u$; let $s_m$ ($s_M$) be the minimum (maximum, respectively) slope of an edge $u_iu_{i+1}$ or $v_jv_{j+1}$. If $s_3-\pi <s_m\leq s_M <s_1$ (if $s_3<s_m\leq s_M <s_1+\pi$), then $\Gamma$ lies entirely to the right (to the left, respectively) of $\ell_2$, except for the vertex $u$. 
\end{lemma}

\begin{proof}
We only prove that, if $s_3-\pi <s_m\leq s_M <s_1$, then $\Gamma$ lies entirely to the right of $\ell_2$, except for the vertex $u$; the proof that, if $s_3<s_m\leq s_M <s_1+\pi$, then $\Gamma$ lies entirely to the left of $\ell_2$, except for the vertex $u$, is symmetric.

Further, it suffices to prove that the paths $(u=u_1,u_2,\dots,u_p)$ and $(u=v_1,v_2,\dots,v_q)$ lie to the right of $\ell_2$, except for the vertex $u$; indeed, if that is the case, then the planarity of $\Gamma$ implies that the entire drawing $\Gamma$, except for the vertex $u$, lies to the right of $\ell_2$. 

We now prove that the path $(u=u_1,u_2,\dots,u_p)$ lies to the right of $\ell_2$, except for the vertex $u$; the proof for the path $(u=v_1,v_2,\dots,v_q)$ is analogous.

For $i=1,\dots,p$, let $W_i$ be the open wedge delimited by the half-lines starting at $u_i$ with slopes $s_3-\pi$ and $s_1$; that is, $W_i$ is the region of the plane that is spanned by a half-line starting at $u_i$ with slope $s_3-\pi$ while rotating counter-clockwise around $u_i$ until it has slope $s_1$. We claim that $W_i$ contains the path $(u_i,u_{i+1},\dots,u_p)$ in its interior, except for the vertex $u_i$ which is on the boundary of $W_i$. Observe that the claim (with $i=1$) implies the lemma, since $W_i$ lies to the right of $\ell_2$, given that $s_2-\pi \leq s_3-\pi$ and $s_1\leq s_2$. 

We now prove the claim by reverse induction on $i$. The case $i=p$ is trivial. Hence, assume that $W_{i+1}$ contains the path $(u_{i+1},u_{i+2},\dots,u_p)$ in its interior, except for the vertex $u_{i+1}$ which is on the boundary of $W_{i+1}$. See Fig.~\ref{fi:three-lines2}. Since $s_3-\pi <s_m\leq s_M <s_1$, the edge $u_iu_{i+1}$ lies in the interior of the wedge $W_{i}$, except for the vertex $u_{i}$ which is on the boundary of $W_{i}$. Further, since $u_{i+1}$ lies in the interior of $W_{i}$, the entire wedge $W_{i+1}$, and hence the path $(u_{i+1},u_{i+2},\dots,u_p)$, lies in the interior of $W_{i}$. This completes the induction and hence concludes the proof of the lemma.
\end{proof}

\section{Proof of Theorem~\ref{th:main}}

In this section we prove Theorem~\ref{th:main}. Throughout the section, we will work with {\em plane graphs}. Further, we will deal with a class of graphs that is wider than the one of $3$-connected planar graphs. The graphs in this class have been introduced by Chen and Yu~\cite{cy-lc3g-02} with the name of {\em strong circuit graphs}, as they constitute a subclass of the well-known {\em circuit graphs}, whose definition is due to Barnette and dates back to 1966~\cite{b-tpg-65}. Here we rephrase the definition of strong circuit graphs as follows. 

\begin{definition} \label{def:strong-circuit}
	A {\em strong circuit graph} is a triple $(G,u,v)$ such that either: (i) $G$ is an edge $uv$ or (ii) $|V(G)|\geq 3$ and the following properties are satisfied. 
	
	\begin{itemize}
		\item[(a)] $G$ is a $2$-connected plane graph;
		\item[(b)] $u$ and $v$ are two distinct external vertices of $G$;
		\item[(c)] if edge $uv$ exists, then it coincides with the path $\tau_{uv}(G)$; and
		\item[(d)] for every $2$-cut $\{a,b\}$ of $G$ we have that $a$ and $b$ are external vertices of $G$ and at least one of them is an internal vertex of the path $\beta_{uv}(G)$; further, every non-trivial $\{a,b\}$-component of $G$ contains an external vertex of $G$ different from $a$ and $b$.
	\end{itemize}
\end{definition}

Several problems are more easily solved on (strong) circuit graphs than on $3$-connected planar graphs. This is because the (strong) circuit graphs can be easily decomposed into smaller (strong) circuit graphs, and hence are suitable for inductive proofs. We now present a structural decomposition for strong circuit graphs whose main ideas can be found in a paper by Chen and Yu~\cite{cy-lc3g-02} (see also a recent paper by Da Lozzo et al.~\cite{ddfmr-dpgmcv-16} for an application of this decomposition to {\em cubic} strong circuit graphs). 

Consider a strong circuit graph $(G,u,v)$ such that $G$ is neither a single edge nor a simple cycle. The decomposition distinguishes the case in which the path $\tau_{uv}(G)$ coincides with the edge $uv$ (Case A) from the case in which it does not (Case B). 

\begin{figure}[htb]
	\centering
	\includegraphics[scale=0.6]{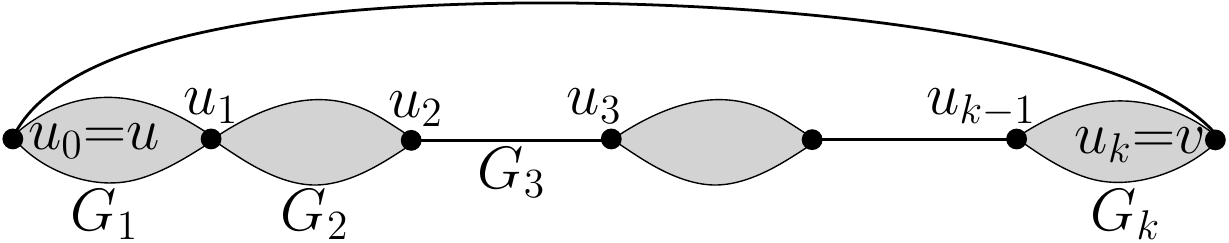}
	\caption{Structure of $(G,u,v)$ in Case~A.}
	\label{fi:structure1}
\end{figure}

\begin{lemma} \label{le:decomposition-A}
Suppose that we are in Case A (refer to Fig.~\ref{fi:structure1}). Then the graph $G'=G-uv$ consists of a sequence of graphs $G_1,\dots,G_k$, with $k\geq 1$, such that:
\begin{itemize}
	\item~\ref{le:decomposition-A}a: for $i=1,\dots,k-1$, the graphs $G_i$ and $G_{i+1}$ share a single vertex $u_i$; further, $G_i$ is in the outer face of $G_{i+1}$ and vice versa in the plane embedding of $G$;
	\item~\ref{le:decomposition-A}b: for $1\leq i,j \leq k$ with $j\geq i+2$, the graphs $G_i$ and $G_j$ do not share any vertex; and
	\item~\ref{le:decomposition-A}c: for $i=1,\dots,k$ with $u_0=u$ and $u_k=v$, $(G_i,u_{i-1},u_i)$ is a strong circuit graph.
\end{itemize}
\end{lemma}

\begin{proof}
Consider the {\em BC-tree} $T'$ of $G'$, which is the tree that is defined as follows. The tree $T'$ contains a {\em B-node} for each $2$-connected component of $G'$ and a {\em C-node} for each $1$-cut of $G'$; further, $T'$ contains an edge between a B-node $b$ and a C-node $c$ if the $1$-cut corresponding to $c$ is a vertex of the $2$-connected component corresponding to $b$. 

\begin{figure}[htb]
	\centering
	\subfloat[]{
		\includegraphics[scale=0.6]{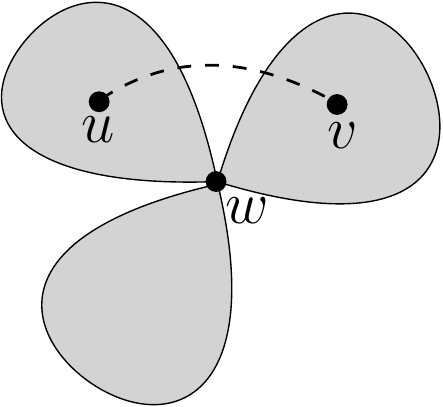}
		\label{fi:BC-tree1}
	}\hfil
	\subfloat[]{
		\includegraphics[scale=0.6]{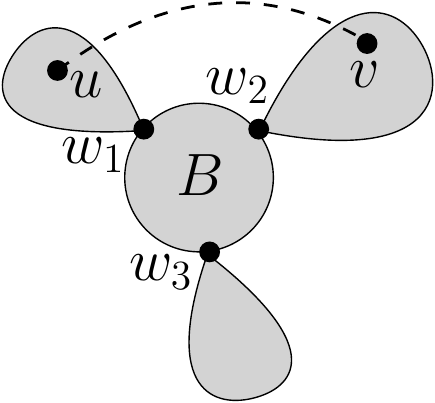}
		\label{fi:BC-tree2}
	}\hfil
	\subfloat[]{
		\includegraphics[scale=0.6]{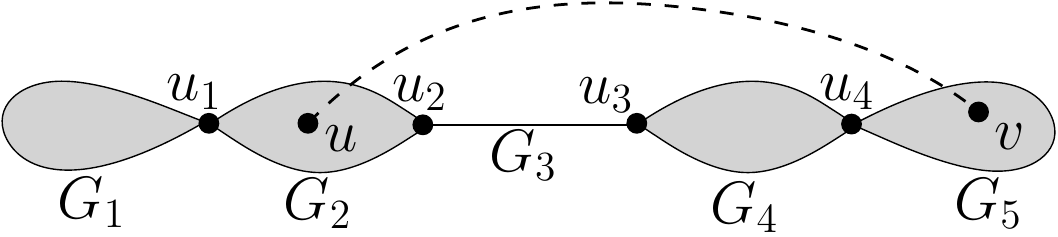}
		\label{fi:BC-tree3}
	}
	\caption{(a) The BC-tree $T'$ of $G'$ contains a node with degree at least $3$ corresponding to a $1$-cut $\{w\}$ of $G'$. (b) The BC-tree $T'$ of $G'$ contains a node with degree at least $3$ corresponding to a $2$-connected component $B$ of $G'$. (c) The vertices $u$ and $v$ are not in $G_1$.}
	\label{fi:BCtree}
\end{figure}

First, we have that $T'$ is a path. Namely suppose, for a contradiction, that $T'$ has a node $t$ with degree at least $3$. If $t$ corresponds to a $1$-cut $\{w\}$ of $G'$, as in Fig.~\ref{fi:BC-tree1}, then $w$ belongs to at least three $2$-connected components of $G'$ and the graph $G''=G'-\{w\}$ consists of at least three connected components. Hence, the graph $G''$ plus edge $uv$ is disconnected, which implies that $\{w\}$ is a $1$-cut of $G$; this contradicts Property~(a) of $(G,u,v)$. Analogously, if $t$ corresponds to a $2$-connected component $B$ of $G'$, as in Fig.~\ref{fi:BC-tree2}, then $B$ contains three distinct $1$-cuts $\{w_1\}$, $\{w_2\}$, and $\{w_3\}$ of $G'$; for each $i\in \{1,2,3\}$, the removal of $w_i$ from $G'$ disconnects $G'$ into at least two connected components, at least one of which, denoted by $G_i$, does not contain vertices of $B$. Since $G_1$, $G_2$, and $G_3$ share no vertex, the edge $uv$ connects at most two components $G_i$ and $G_j$ with $i,j\in \{1,2,3\}$, which implies that $\{w_h\}$ is a $1$-cut of $G$, where $h\neq i,j$ and $h\in \{1,2,3\}$; this contradicts Property~(a) of $(G,u,v)$. Hence $T'$ is a path $(b_1,c_1,b_2,c_2,\dots,b_{k-1},c_{k-1},b_k)$. 

Let $G_i$ be the $2$-connected component of $G'$ corresponding to the B-node $b_i$ and let $\{u_i\}$ be the $1$-cut of $G'$ corresponding to the C-node $c_i$. Then, for $i=1,\dots,k-1$, the graphs $G_i$ and $G_{i+1}$ share a single vertex $u_i$, while for $1\leq i,j \leq k$ with $j\geq i+2$ the graphs $G_i$ and $G_j$ do not share any vertex. The vertices $u$ and $v$ are one in $G_1$ and one in $G_k$; indeed, if say $G_1$ did not contain any of $u$ and $v$, as in Fig.~\ref{fi:BC-tree3}, then $\{u_1\}$ would be a $1$-cut of $G$; this would contradict Property~(a) of $(G,u,v)$. Assume, w.l.o.g.\ up to renaming, that $u$ belongs to $G_1$ and $v$ belongs to $G_k$. We also have that $u\neq u_1$, as if $u=u_1$ then $\{u_1\}$ would be a $1$-cut of $G$, again contradicting Property~(a) of $(G,u,v)$; analogously, $v\neq u_{k-1}$. 

We prove that $G_{i+1}$ lies in the outer face of $G_{i}$ in the plane embedding of $G$, for every $i=1,\dots,k-1$. Suppose for a contradiction that, for some $i\in\{1,\dots,k-1\}$, the graph $G_{i+1}$ lies inside an internal face $f$ of $G_i$ (except for the vertex $u_i$, which is on the boundary of $f$) in the plane embedding of $G$. Since the graphs $G_{i+2}, \dots, G_k$ do not share any vertex with $G_i$, by planarity they all lie inside $f$. It follows that the vertex $v$ lies inside $f$ (note that $v\neq u_i$ even if $k=i+1$) and hence it is not incident to the outer face of $G$, which contradicts Property~(b) of $(G,u,v)$. An analogous proof shows that $G_{i}$ lies in the outer face of $G_{i+1}$ in the plane embedding of $G$, for every $i=1,\dots,k-1$.

It remains to prove that, for $i=1,\dots,k$, the triple $(G_i,u_{i-1},u_{i})$ is a strong circuit graph, where $u_0=u$ and $u_k=v$. We are going to use the fact that $\beta_{uv}(G)$ is composed of the paths $\beta_{uu_1}(G_1),\beta_{u_1u_2}(G_2),\dots,\beta_{u_{k-1}v}(G_k)$. This is because $uv$ coincides with $\tau_{uv}(G)$ by Property~(c) of $(G,u,v)$ and because $G_{i+1}$ lies in the outer face of $G_{i}$ and vice versa in the plane embedding of $G$.  

(a) Graph $G_i$ is $2$-connected by assumption and it is associated with a plane embedding, given that it is a subgraph of the plane graph $G$.   

(b) For $i=1,\dots,k-1$, the vertex $u_{i}$ is external in the plane embedding of $G_i$, since $G_i$ is in the outer face of $G_{i+1}$ and vice versa; analogously, for $i=2,\dots,k$, the vertex $u_{i-1}$ is external in the plane embedding of $G_i$. Further, $u_0=u$ and $u_k=v$ are external in the plane embeddings of $G_1$ and $G_k$, respectively, since they are external in the plane embedding of $G$. Finally, for $i=1,\dots,k$, the vertices $u_{i-1}$ and $u_i$ are distinct, as otherwise $\{u_{i-1}=u_i\}$ would be a $1$-cut of $G$, which would contradict Property~(a) of $(G,u,v)$.

(c) Suppose, for a contradiction, that the edge $u_{i-1}u_i$ exists and that it does not coincide with $\tau_{u_{i-1}u_i}(G_i)$. This implies that $G_i$ contains vertices different from $u_{i-1}$ and $u_i$, and hence that $\{u_{i-1},u_i\}$ is a $2$-cut of $G$. The $\{u_{i-1},u_i\}$-component $H_i$ of $G$ that contains $\tau_{u_{i-1}u_i}(G_i)$ is non-trivial, given that $\tau_{u_{i-1}u_i}(G_i)$ does not coincide with $u_{i-1}u_i$. Further, no vertex of $H_i$ other than $u_{i-1}$ and $u_i$ is incident to the outer face of $G$, given that all the vertices of $H_i$ other than $u_{i-1}$ and $u_i$ lie inside the region delimited by the cycle $\beta_{uu_1}(G_1) \cup \dots \cup \beta_{u_{i-2}u_{i-1}}(G_{i-1})\cup u_{i-1}u_i \cup \beta_{u_iu_{i+1}}(G_{i+1})\cup \dots \cup\beta_{u_{k-1}v}(G_k) \cup vu$. However, this contradicts Property~(d) of $(G,u,v)$. 

(d) Consider any $2$-cut $\{a,b\}$ of $G_i$; then $G_i$ has at least two non-trivial $\{a,b\}$-components. 
	
We prove that $a$ and $b$ are external vertices of $G_i$. Suppose, for a contradiction, that $a$ is an internal vertex of $G_i$ (the argument if $b$ is an internal vertex of $G_i$ is analogous), as in Fig.~\ref{fi:PropertyD1}. Then the cycle delimiting the outer face of $G_i$ belongs to a single $\{a,b\}$-component $H$ of $G_i$, and there is a non-trivial $\{a,b\}$-component $H'\neq H$ of $G_i$ that does not contain any external vertices of $G_i$ other than $b$. Since all the edges in $E(G)-E(G_i)$ lie in the outer face of $G_i$ in the plane embedding of $G$, it follows that $\{a,b\}$ is a $2$-cut of $G$, and that $H'$ is a non-trivial $\{a,b\}$-component of $G$ that does not contain any external vertices of $G$ other than $b$. However, this contradicts Property~(d) of $(G,u,v)$.  

\begin{figure}[htb]
	\centering
	\subfloat[]{
		\includegraphics[scale=0.6]{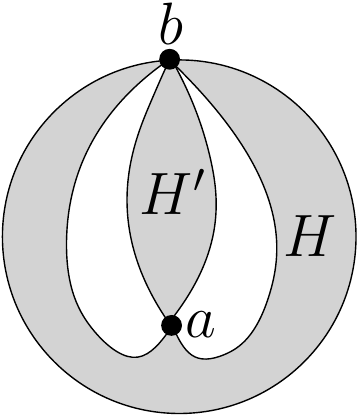}
		\label{fi:PropertyD1}
	}\hfil
	\subfloat[]{
		\includegraphics[scale=0.6]{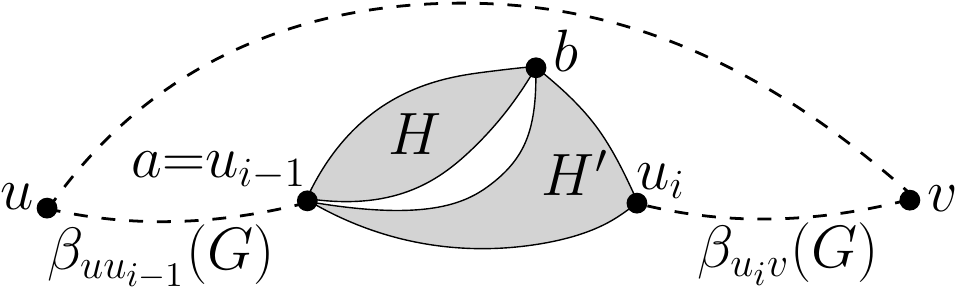}
		\label{fi:PropertyD2}
	}
	\caption{Illustration for the proof that $(G_i,u_{i-1},u_{i})$ satisfies Property~(d) of a strong circuit graph. (a) A vertex $a$ of a $2$-cut $\{a,b\}$ of $G_i$ is an internal vertex of $G_i$. (b) The $\{a,b\}$-component $H$ of $G_i$ containing $\tau_{ab}(G_i)$ contains no internal vertex of $\beta_{ab}(G_i)$.}
	\label{fi:PropertyD}
\end{figure}

Next, we prove that at least one of $a$ and $b$ is an internal vertex of $\beta_{u_{i-1}u_i}(G_i)$. Suppose, for a contradiction, that $a$ and $b$ are both in $\tau_{u_{i-1}u_i}(G_i)$. Assume, w.l.o.g.\ up to renaming of $a$ and $b$, that $u_{i-1}$, $a$, $b$, and $u_{i}$ appear in this order in $\tau_{u_{i-1}u_i}(G_i)$, where possibly $u_{i-1}=a$ and/or $b=u_i$. Let $H$ be the $\{a,b\}$-component of $G_i$ containing $\tau_{ab}(G_i)$; let $H'$ be any non-trivial $\{a,b\}$-component of $G_i$ different from $H$. 
	
\begin{itemize}
	\item If $H$ contains an internal vertex of $\beta_{ab}(G_i)$, then it contains the entire cycle delimiting the outer face of $G_i$. The planarity of $G_i$ implies that $H'$ lies inside an internal face of $H$, except at vertices $a$ and $b$. This has two consequences. First, since all the edges in $E(G)-E(G_i)$ lie in the outer face of $G_i$ (and of $H$) in the plane embedding of $G$, the set $\{a,b\}$ is a $2$-cut of $G$ and hence $H'$ is a non-trivial $\{a,b\}$-component of $G$. Second, no vertex of $H'$ other than $a$ and $b$ is incident to the outer face of $G_i$ or to the outer face of $G$. These two statements contradict Property~(d) for $(G,u,v)$. 
		
	\item If $H$ contains no internal vertex of $\beta_{ab}(G_i)$, as in Fig.~\ref{fi:PropertyD2}, then $u_{i-1},u_i \notin V(H)-\{a,b\}$, hence no edge in $E(G)-E(G_i)$ is incident to a vertex of $H$ different from $a$ and $b$. Since the vertices of $\beta_{u_{i-1}u_i}(G_i)$ are the only external vertices of $G$ in $V(G_i)$, it follows that $H$ is a non-trivial $\{a,b\}$-component of $G$ that contains no external vertex of $G$ other than, possibly, $a$ and $b$. This contradicts Property~(d) for $(G,u,v)$.
\end{itemize}   
	
Finally, we prove that every non-trivial $\{a,b\}$-component $H$ of $G_i$ contains an external vertex of $G_i$ different from $a$ and $b$. Namely, if that is not the case for a non-trivial $\{a,b\}$-component $H$ of $G_i$, then no edge in $E(G)-E(G_i)$ is incident to a vertex of $H$ different from $a$ and $b$. This implies that the set $\{a,b\}$ is a $2$-cut of $G$ and $H$ is a non-trivial $\{a,b\}$-component of $G$. However, no vertex of $H$ other than, possibly, $a$ and $b$ is incident to the outer face of $G$. This contradicts Property~(d) for $(G,u,v)$. 
\end{proof}

Given a strong circuit graph $(G,u,v)$ that is not a single edge, the vertex $u$ belongs to one $2$-connected component of the graph $G-\{v\}$. Indeed, if it belonged to more than one $2$-connected component of $G-\{v\}$, then $\{u\}$ would be a $1$-cut of $G-\{v\}$, hence $\{u,v\}$ would be a $2$-cut of $G$, which contradicts Property~(d) for $(G,u,v)$. We now present the following.

\begin{figure}[htb]
	\centering
	\includegraphics[scale=0.6]{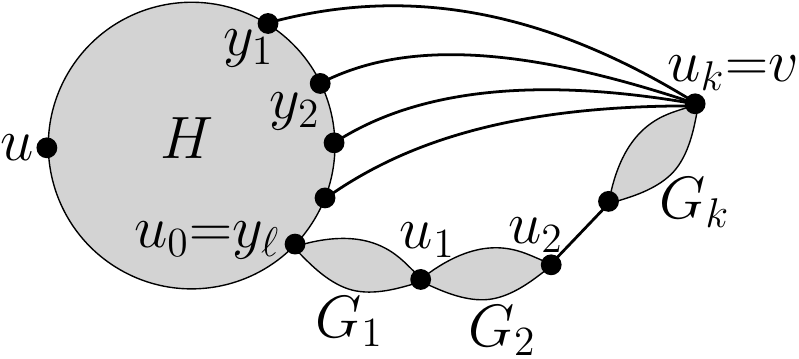}
	\caption{Structure of $(G,u,v)$ in Case~B.}
	\label{fi:structure2}
\end{figure}

\begin{lemma} \label{le:decomposition-B}
Suppose that we are in Case B  (refer to Fig.~\ref{fi:structure2}). Let $H$ be the $2$-connected component of the graph $G-\{v\}$ that contains $u$; then we have $|V(H)|\geq 3$. Further, let $H'$ denote the graph $H\cup\{v\}$. Then $G$ contains $\ell$ distinct $H'$-bridges $B_1,\dots,B_\ell$, for some $\ell\geq 2$, such that: 
\begin{itemize}
	\item\ref{le:decomposition-B}a: each $H'$-bridge $B_i$ has two attachments, namely $v$ and a vertex $y_i\in V(H)$;
	\item\ref{le:decomposition-B}b: the $H'$-bridges $B_1,\dots,B_{\ell-1}$ are trivial, while $B_\ell$ might be trivial or not; 
	\item\ref{le:decomposition-B}c: any two among $y_1,\dots,y_{\ell}$ are distinct except, possibly, for $y_{\ell-1}$ and $y_{\ell}$; also if $\ell=2$, then $y_1$ and $y_2$ are distinct;
	\item\ref{le:decomposition-B}d: $y_1$ is an internal vertex of $\tau_{uv}(G)$; further, $B_1$ is an edge that coincides with $\tau_{y_1v}(G)$;
	\item\ref{le:decomposition-B}e: $y_{\ell}$ is an internal vertex of $\beta_{uv}(G)$ and $\beta_{uy_1}(H)$;
	further, $B_\ell$ contains the path $\beta_{y_\ell v}(G)$;
	\item\ref{le:decomposition-B}f: $B_1,\dots,B_{\ell-1}$ appear in this counter-clockwise order around $v$ and lie in the outer face of $B_{\ell}$ in the plane embedding of $G$;
	\item\ref{le:decomposition-B}g: the triple $(H,u,y_1)$ is a strong circuit graph; and 
	\item\ref{le:decomposition-B}h: $B_{\ell}$ consists of a sequence of graphs $G_1,\dots,G_k$, with $k\geq 1$, such that:
	\begin{itemize}
		\item for $i=1,\dots,k-1$, the graphs $G_i$ and $G_{i+1}$ share a single vertex $u_i$; further, $G_i$ is in the outer face of $G_{i+1}$ and vice versa in the plane embedding of $G$;
		\item for $1\leq i,j \leq k$ with $j\geq i+2$, the graphs $G_i$ and $G_j$ do not share any vertex; and
		\item for $i=1,\dots,k$ with $u_0$$=$$y_{\ell}$ and $u_k$$=$$v$, the triple $(G_i,u_{i-1},u_i)$ is a strong circuit graph.
	\end{itemize}	
\end{itemize}

\end{lemma}

\begin{proof}
We first prove that $|V(H)|\geq 3$. Suppose, for a contradiction, that $H$ is a single edge $uy_1$. If the degree of $u$ in $G$ is one, then $\{y_1\}$ is a $1$-cut of $G$; this contradicts Property~(a) for $(G,u,v)$. Otherwise, there is a $(H\cup \{v\})$-bridge $B_i$ of $G$ whose attachment in $H$ is $u$. If $B_i$ is trivial, then it coincides with the edge $uv$; however, this contradicts the hypothesis of Case~B. Otherwise, $B_i$ is non-trivial; however, this implies that $\{u,v\}$ is a $2$-cut of $G$, as the removal of $u$ and $v$ from $G$ disconnects $y_1$ from the vertices in $V(B_i)-\{u,v\}$; since neither $u$ nor $v$ is an internal vertex of $\beta_{uv}(G)$, this contradicts Property~(d) for $(G,u,v)$. 

We now prove the properties of the lemma. First, if $G$ had no $H'$-bridge, then it would not be connected, while it is $2$-connected. Hence, $G$ contains $\ell$ distinct $H'$-bridges $B_1,\dots,B_\ell$ with $\ell\geq 1$. Each $H'$-bridge $B_i$ has at most one attachment $y_i\in V(H)$, as if $B_i$ had at least two attachments in $V(H)$ then it would contain a path (not passing through $v$) between two vertices of $H$; however, such a path would be in $H$, and not in $B_i$, given that $H$ is a maximal $2$-connected subgraph of $G-\{v\}$. It follows that $\ell\geq 2$, as if $\ell=1$ then $y_1$ would be a $1$-cut of $G$, whereas $G$ is $2$-connected. Further, for $i=1,2,\dots,\ell$, the vertex $v$ is an attachment of $B_i$, as otherwise $y_i$ would be a $1$-cut of $G$, whereas $G$ is $2$-connected. Analogously, for $i=1,2,\dots,\ell$, the vertex $y_i$ is an attachment of $B_i$, as otherwise $v$ would be a $1$-cut of $G$, whereas $G$ is $2$-connected. This proves Property~\ref{le:decomposition-B}a. 

Suppose, for a contradiction, that $y_i=u$, for some $i\in \{1,2,\dots,\ell\}$. If $B_i$ is a trivial $H'$-bridge, then it coincides with the edge $uv$; however, this contradicts the fact that we are in Case~B. If $B_i$ is a non-trivial $H'$-bridge, then $\{u,v\}$ is a $2$-cut of $G$; namely, the removal of $u$ and $v$ from $G$ disconnects the vertices in $V(H)-\{u\}$ from the vertices in $V(B_i)-\{u,v\}$ -- the latter set is non-empty given that $B_i$ is non-trivial. However, this contradicts Property~(d) for $(G,u,v)$, given that neither $u$ nor $v$ is an internal vertex of $\beta_{uv}(G)$. It follows that $y_i\neq u$, for $i=1,2,\dots,\ell$.

We now prove Properties~\ref{le:decomposition-B}b--\ref{le:decomposition-B}f. Since $v$ is incident to the outer face of $G$, it lies in the outer face of $H$. It follows that all the $H'$-bridges  $B_1,\dots,B_\ell$ lie in the outer face of $H$, except at the vertices $y_1,\dots,y_\ell$, respectively. By the planarity of $G$, there are at most two $H'$-bridges among $B_1,\dots,B_\ell$ that contain edges incident to the outer face of $G$. If there were only one $H'$-bridge $B_i$ containing edges incident to the outer face of $G$, as in Fig.~\ref{fi:CaseB-Proof1}, then $\{y_i\}$ would be a $1$-cut of $G$, whereas $G$ is $2$-connected. Hence, there are exactly two $H'$-bridges among $B_1,\dots,B_\ell$ containing edges incident to the outer face of $G$. Denote them by $B_1$ and $B_\ell$, as in Fig.~\ref{fi:CaseB-Proof2}, so that $u$, $y_1$, and $y_\ell$ appear in this clockwise order along the outer face of $H$. Then $y_1\neq y_\ell$, as otherwise $\{y_1\}$ would be a $1$-cut of $G$, whereas $G$ is $2$-connected; in particular, $y_1\neq y_2$ if $\ell=2$. It also follows that $y_1$ is an internal vertex of $\tau_{uv}(G)$, that $y_{\ell}$ is an internal vertex of $\beta_{uv}(G)$ and $\beta_{uy_1}(H)$, that $B_1$ contains $\tau_{y_1 v}(G)$, that $B_\ell$ contains $\beta_{y_\ell v}(G)$, and that every vertex $y_i\neq y_1,y_{\ell}$ is not incident to the outer face of $G$. Now consider any $H'$-bridge $B_i$ of $G$ with $y_i\neq y_\ell$. The graph $B_i$ is a $\{y_i,v\}$-component of $G$, however $\{y_i,v\}$ is a pair of vertices none of which is internal to $\beta_{uv}(G)$, hence if $B_i$ were non-trivial, then Property~(d) of $(G,u,v)$ would be violated.  It follows that the only $H'$-bridges which might be non-trivial are those whose attachment in $H$ is $y_\ell$; in particular, since $y_1\neq y_\ell$, we have that $B_1$ is trivial and coincides with the edge $y_1v=\tau_{y_1v}(G)$. If there were at least two non-trivial $H'$-bridges whose attachment in $H$ is $y_\ell$, then at least one of them (in fact all the ones different from $B_{\ell}$) would not contain any vertex incident to the outer face of $G$ other than $y_\ell$ and $v$; however, this would violate Property~(d) of $(G,u,v)$. It follows that $B_\ell$ is the only $H'$-bridge of $G$ that is possibly non-trivial. By the planarity of $G$ and the connectivity of $B_\ell-\{y_\ell,v\}$, all the trivial $H'$-bridges of $G$ lie in the outer face of $B_\ell$; denote them by $B_1,\dots,B_{\ell-1}$ in their counter-clockwise order around $v$. Since $B_1,\dots,B_{\ell-1}$ are trivial and incident to $v$, then  $y_1,\dots,y_{\ell-1}$ are all distinct. By planarity and since $y_1\neq y_{\ell}$, it follows that $B_{\ell-1}$ and $B_{\ell}$ are the only $H'$-bridges which might share their attachment in $H$. This concludes the proof of Properties~\ref{le:decomposition-B}b--\ref{le:decomposition-B}f.

\begin{figure}[htb]
	\centering
	\subfloat[]{
		\includegraphics[scale=0.6]{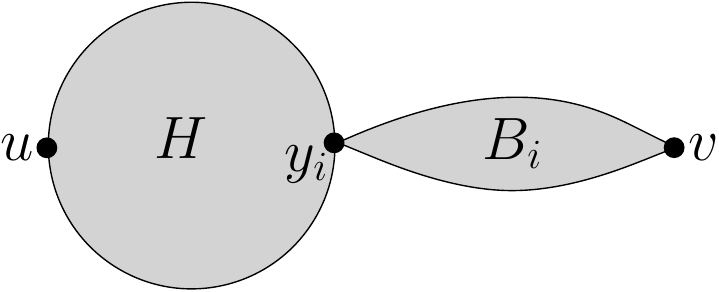}
		\label{fi:CaseB-Proof1}
	}\hfil
	\subfloat[]{
		\includegraphics[scale=0.6]{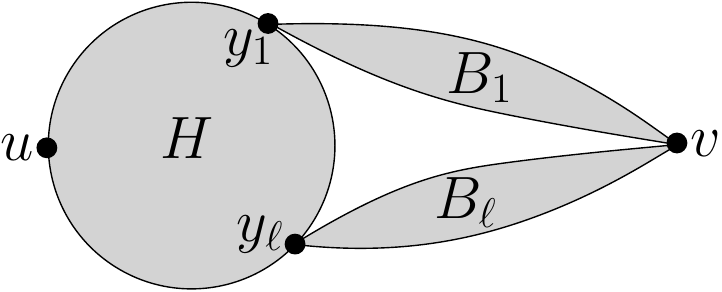}
		\label{fi:CaseB-Proof2}
	}
	\caption{(a) If there were exactly one $H'$-bridge $B_i$ containing edges incident to the outer face of $G$, then $y_i$ would be a $1$-cut of $G$. (b) The $H'$-bridges $B_1$ and $B_\ell$ of $G$.}
	\label{fi:CaseB-Proof}
\end{figure}

We now prove that the triple $(H,u,y_1)$ is a strong circuit graph. 

(a) Graph $H$ is $2$-connected by assumption and it is associated with a plane embedding, given that it is a subgraph of the plane graph $G$.

(b) The vertex $u$ is incident to the outer face of $H$ since $(G,u,v)$ satisfies Property~(b). The vertex $y_1$ is a vertex of $\tau_{uv}(G)$, as argued above, and hence it is incident to the outer face of $G$ and to the one of $H$. Finally, $u$ and $y_1$ are distinct, as otherwise $\tau_{uv}(G)$ would coincide with the edge $uv$, contradicting the fact that we are in Case~B. 

(c) Suppose, for a contradiction, that the edge $u y_1$ exists and does not coincide with $\tau_{uy_1}(H)$. Then $\{u,y_1\}$ is a $2$-cut of $G$, since the removal of $u$ and $y_1$ disconnects the internal vertices of $\tau_{uy_1}(H)$ (which exist since $\tau_{uy_1}(H)$ is not the edge $uy_1$) from $v$. However, none of $u$ and $y_1$ is an internal vertex of $\beta_{uv}(G)$; this contradicts Property~(d) for $(G,u,v)$.    

(d) The proof that $(H,u,y_1)$ satisfies Property~(d) is very similar to the proof that $(G_i,u_{i-1},u_{i})$ satisfies Property~(d) in Lemma~\ref{le:decomposition-A}, hence it is only sketched here. 

Consider any $2$-cut $\{a,b\}$ of $H$. If one of $a$ and $b$ is an internal vertex of $H$, then one non-trivial $\{a,b\}$-component $L'$ of $H$ lies inside an internal face of another non-trivial $\{a,b\}$-component $L$ of $H$. This implies that $\{a,b\}$ is a $2$-cut of $G$ and that $L'$ is an $\{a,b\}$-component of $G$ that does not contain any external vertices of $G$ other than $a$ or $b$; this contradicts Property~(d) for $(G,u,v)$. It follows that $a$ and $b$ are external vertices of $H$.   

If $a$ and $b$ are both in $\tau_{uy_1}(H)$, then assume that $u$, $a$, $b$, and $y_1$ appear in this order in $\tau_{uy_1}(H)$, where possibly $u=a$ and/or $b=y_1$. Let $L$ be the $\{a,b\}$-component of $H$ containing $\tau_{ab}(H)$ and let $L'$ be any non-trivial $\{a,b\}$-component of $H$ different from $L$. If $L$ contains an internal vertex of $\beta_{ab}(H)$, then it contains the entire cycle delimiting the outer face of $H$. It follows that $L'$ lies inside an internal face of $L$, except at $a$ and $b$, and hence that $L'$ is an $\{a,b\}$-component of $G$ that does not contain any external vertices of $G$ other than $a$ or $b$. This contradicts Property~(d) for $(G,u,v)$. If $L$ contains no internal vertex of $\beta_{ab}(H)$, then no edge in $E(G)-E(H)$ is incident to a vertex of $L$ different from $a$ and $b$. It follows that $L$ is a non-trivial $\{a,b\}$-component of $G$; further, neither $a$ nor $b$ is an internal vertex of $\beta_{uv}(G)$, given that $V(\tau_{ab}(H))\subseteq V(\tau_{uv}(G))$. This contradicts Property~(d) for $(G,u,v)$. It follows that at least one of $a$ and $b$ is an internal vertex of $\beta_{uy_1}(H)$. 

Finally, assume that a non-trivial $\{a,b\}$-component $L$ of $H$ contains no external vertex of $H$ other than $a$ and $b$. Then $\{a,b\}$ is a $2$-cut of $G$ and $L$ is a non-trivial $\{a,b\}$-component of $G$ that contains no external vertex of $G$ other than, possibly, $a$ and $b$. This contradicts Property~(d) for $(G,u,v)$.  Hence, every non-trivial $\{a,b\}$-component of $H$ contains an external vertex of $H$ other than $a$ and $b$. This proves Property~\ref{le:decomposition-B}g for $(H,u,y_1)$.  

In order to prove Property~\ref{le:decomposition-B}h, assume that $B_{\ell}$ does not coincide with the edge $y_\ell v$, as otherwise there is nothing to prove. Let $B'_{\ell}$ be the plane graph obtained by adding the edge $y_\ell v$ to $B_{\ell}$, so that $y_\ell$ immediately precedes $v$ in the clockwise order of the vertices along the outer face of $B'_{\ell}$ (both $y_{\ell}$ and $v$ are indeed incident to the outer face of $B_{\ell}$); see Fig.~\ref{fi:CaseB-Bell}. We prove that $(B'_{\ell}, y_\ell, v)$ is a strong circuit graph; then Property~\ref{le:decomposition-B}h follows by applying Lemma~\ref{le:decomposition-A} to~$(B'_{\ell}, y_\ell, v)$.

\begin{figure}[htb]
	\centering
	\includegraphics[scale=0.6]{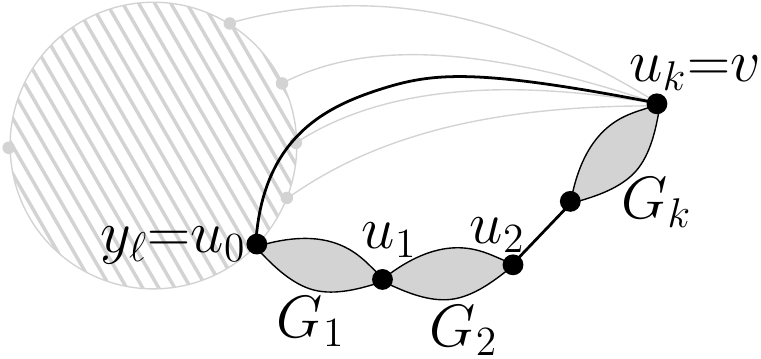}
	\caption{The graph $B'_{\ell}$.}
	\label{fi:CaseB-Bell}
\end{figure}

(a) The graph $B_{\ell}$ is associated with a plane embedding, given that it is a subgraph of the plane graph $G$. Further, $y_{\ell}$ and $v$ are both incident to the outer face of $B_{\ell}$, hence the plane graph $B'_{\ell}$ is well-defined. We prove that $B'_{\ell}$ is $2$-connected. Since $y_{\ell}$ and $v$ are adjacent in $B'_{\ell}$, they belong to the same $2$-connected component of $B'_{\ell}$. However, the only vertices of $B_{\ell}$ that are incident to edges in $E(G)-E(B_{\ell})$ are $y_{\ell}$ and $v$. It follows that any $1$-cut of $B'_{\ell}$ is also a $1$-cut of $G$. Then $B'_{\ell}$ is $2$-connected since $G$ is.

(b) The vertices $y_{\ell}$ and $v$ are distinct since the first one belongs to $H$, while the second one does not. Further, both $y_{\ell}$ and $v$ are incident to the outer face of $B_{\ell}$, as argued above, and hence are external vertices of $B'_{\ell}$.

(c) The edge $y_{\ell}v$ exists and coincides with $\tau_{y_{\ell}v}(B'_{\ell})$, by construction.    

(d) Consider any $2$-cut $\{a,b\}$ of $B'_{\ell}$ (possibly $\{a,b\}\cap \{y_\ell,v\}\neq \emptyset$). 

Since $y_\ell v\in E(B'_{\ell})$, we have that $y_\ell$ and $v$ are in the same $\{a,b\}$-component $L$ of $B'_{\ell}$. Since $y_{\ell}$ and $v$ are the only vertices of $B_{\ell}$ incident to edges in $E(G)-E(B_{\ell})$, it follows that $\{a,b\}$ is also a $2$-cut of $G$. Then $a$ and $b$ are external vertices of $B'_{\ell}$ since they are external vertices of $G$. 

Next suppose, for a contradiction, that neither $a$ nor $b$ is an internal vertex of $\beta_{y_\ell v}(B'_{\ell})$. Since $a$ and $b$ are external vertices of $B'_{\ell}$ and since $\tau_{y_\ell v}(B'_{\ell})$ coincides with the edge $y_\ell v$, it follows that $a=y_\ell$ and $b=v$ (or vice versa). However, $B_{\ell}$ is a $\{y_\ell,v\}$-component of $G$, hence the removal of $y_\ell$ and $v$ from $B_{\ell}$ (or $B'_{\ell}$) does not disconnect $B_{\ell}$ (or $B'_{\ell}$); this contradicts the assumption that $\{a,b\}$ is a $2$-cut of $B'_{\ell}$, and implies that one of $a$ and $b$ is an internal vertex of $\beta_{y_\ell v}(B'_{\ell})$. 

Finally, consider any non-trivial $\{a,b\}$-component $L$ of $B'_{\ell}$. As proved above $\{a,b\}$ is a $2$-cut of $G$ and at least one of $a$ and $b$ is not in $\{y_\ell,v\}$. If $L$ contains the edge $y_\ell v$, then it contains an external vertex of $B'_{\ell}$ other than $a$ and $b$, namely whichever vertex of $\{y_\ell,v\}$ that is not in $\{a,b\}$. Otherwise, $L$ is also an $\{a,b\}$-component of $G$ and it contains an external vertex of $B'_{\ell}$ other than $a$ and $b$ since it contains an external vertex of $G$ other than $a$ and $b$. 

This concludes the proof that $(B'_{\ell}, y_\ell, v)$ is a strong circuit graph, hence it implies Property~\ref{le:decomposition-B}h via Lemma~\ref{le:decomposition-A}. The lemma follows. 
\end{proof}

We prove that any strong circuit graph $(G,u,v)$ has a planar greedy drawing by exploiting Lemmata~\ref{le:decomposition-A} and~\ref{le:decomposition-B} in a natural way. Indeed, if we are in Case~A (in Case~B) then Lemma~\ref{le:decomposition-A} (resp.\ Lemma~\ref{le:decomposition-B}) is applied in order to construct strong circuit graphs $(G_i,u_{i-1},u_i)$ with $i=1,\dots,k$ (resp.\ strong circuit graphs $(H,u,y_1)$ and $(G_i,u_{i-1},u_i)$  with $i=1,\dots,k$) for which planar greedy drawings are inductively constructed and then combined together in order to get a planar greedy drawing of $(G,u,v)$. The base cases of the induction are the ones in which $G$ is an edge or a simple cycle. Then a planar greedy drawing of $G$ is directly constructed.


In order to be able to combine planar greedy drawings for the strong circuit graphs $(G_i,u_{i-1},u_i)$ (and $(H,u,y_1)$ if we are in Case~B) to construct a planar greedy drawing of $(G,u,v)$, we need the inductively constructed drawings to satisfy some restrictive geometric requirements, which are expressed in the following theorem, which is the core of the proof of Theorem~\ref{th:main}.

\begin{theorem} \label{th:main-aux}
Let $(G,u,v)$ be a strong circuit graph with at least three vertices and let $0<\alpha<\frac{\pi}{4}$ be an arbitrary parameter. Let $\beta_{uv}(G)=(u=b_1,b_2,\dots,b_m=v)$. There exists a straight-line drawing $\Gamma$ of $G$ in the Cartesian plane such that the following holds. For any value $\delta\geq 0$, denote by $\Gamma_{\delta}$ the straight-line drawing obtained from $\Gamma$ by moving the position of vertex $u$ by $\delta$ units to the left. Then $\Gamma_{\delta}$ satisfies the following properties (refer to Fig.~\ref{fi:theorem-statement}). 

\begin{figure}[htb]
	\centering
	\includegraphics[scale=0.6]{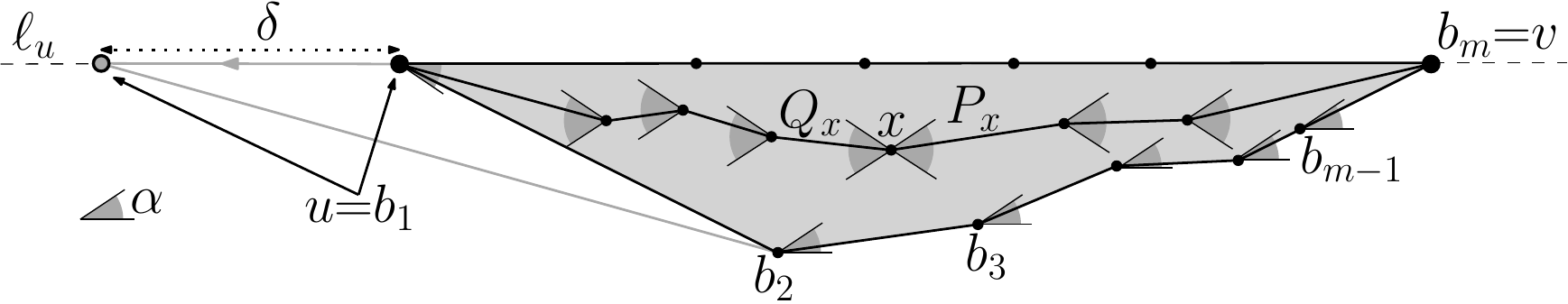}
	\caption{Illustration for the statement of Theorem~\ref{th:main-aux}.}
	\label{fi:theorem-statement}
\end{figure}

\begin{enumerate}
	\item $\Gamma_{\delta}$ is planar;
	\item $\tau_{uv}(G)$ lies entirely on a horizontal line $\ell_u$ with $u$ to the left of $v$;
	\item the edge $b_1b_2$ has slope in the interval $(-\alpha;0)$ and the edge $b_ib_{i+1}$ has slope in the interval $(0; \alpha)$, for each $i=2,3,\dots,m-1$;
	\item for every vertex $x\in V(G)$ there is a path $P_x=(x=v_1,v_2,\dots,v_p=v)$ from $x$ to $v$ in $G$ such that the edge $v_iv_{i+1}$ has slope in the interval $(-\alpha; \alpha)$ in $\Gamma_{\delta}$, for each $i=1,2,\dots,p-1$; further, if $x\neq u$, then $u\notin V(P_{x})$;
	\item for every vertex $x\in V(G)$ there is a path $Q_x=(x=w_1,w_2,\dots,w_q=u)$ from $x$ to $u$ in $G$ such that the edge $w_iw_{i+1}$ has slope in the interval $(\pi-\alpha; \pi+\alpha)$ in $\Gamma_{\delta}$, for each $i=1,2,\dots,q-1$; and 
	\item for every ordered pair of vertices $(x,y)$ in $V(G)$ there is a path $P_{xy}$ from $x$ to $y$ in $G$ such that $P_{xy}$ is distance-decreasing in $\Gamma_{\delta}$; further, if $x,y\neq u$, then $u\notin V(P_{xy})$.
\end{enumerate}
\end{theorem}

Before proceeding with the proof of Theorem~\ref{th:main-aux}, we comment on its statement. First, let us set $\delta=0$ and argue about $\Gamma_0=\Gamma$. Properties~1 and~6 are those that one would expect, as they state that $\Gamma$ is planar and greedy, respectively. Properties~2 and~3 state that all the edges incident to the outer face of $\Gamma$ are ``close'' to horizontal; indeed, the edges of $\tau_{uv}(G)$ are horizontal, the edge $b_1b_2$ has a slightly negative slope, and all the other edges of $\beta_{uv}(G)$ have a slightly positive slope. Since $\Gamma$ is planar, this implies that $\Gamma$ is contained in a wedge delimited by two half-lines with slopes $0$ and $-\alpha$ starting at $u$. Properties~4 and~5 argue about the existence of certain paths from any vertex to $u$ and $v$; these two vertices play an important role in the structural decomposition we employ, since distinct subgraphs are joined on those vertices, and the paths incident to them are inductively combined together in order to construct distance-decreasing paths. Finally, all these properties still hold true if $u$ is moved by an arbitrary non-negative amount $\delta$ to the left. This is an important feature we exploit in one of our inductive cases.


We now present an inductive proof of Theorem~\ref{th:main-aux}. In the base cases $G$ is a single edge (we call this the {\em Trivial Case}) or a simple cycle (we call this the {\em Cycle Case}). 

We start with the {\bf Trivial Case}, in which $G$ is a single edge. Although Theorem~\ref{th:main-aux} assumes that $|V(G)|\geq 3$, for its proof we need to inductively draw certain subgraphs of $G$ which might be single edges. Whenever we need to draw a strong circuit graph $(G,u,v)$ such that $G$ is a single edge $uv$, we draw it as a horizontal straight-line segment with positive length, with $u$ to the left of $v$. We remark that, since Theorem~\ref{th:main-aux} assumes that $|V(G)|\geq 3$, we do not need the constructed drawing to satisfy Properties~1--6. 

We next deal with the {\bf Cycle Case}, in which $G$ is a simple cycle with at least $3$ vertices. Refer to Fig.~\ref{fi:basecase}. By Property~(d) of $(G,u,v)$, the set $\{u,v\}$ is not a $2$-cut of $G$, hence $u$ and $v$ appear consecutively along the cycle $G$. By Property~(c) of $(G,u,v)$, the edge $uv$ coincides with the path $\tau_{uv}(G)$. Drawing $\Gamma$ is constructed as follows. Place $b_1$, $b_2$, and $b_m$ at the vertices of an isosceles triangle $\Delta b_1b_2b_m$ in which the edge $b_1b_m$ lies on a horizontal line $\ell_u$, with $b_1$ to the left of $b_m$, and in which the angles $\measuredangle b_2b_1b_m$ and $\measuredangle b_1b_mb_2$ are $\frac{\alpha}{2}$, with $b_2$ below $\ell_u$. Place the vertices $b_3,\dots,b_{m-1}$ on the straight-line segment $\overline{b_2b_m}$ in this order from $b_2$ to $b_m$. This completes the construction of $\Gamma$. We have the following.

\begin{figure}[htb]
	\centering
	\includegraphics[scale=0.6]{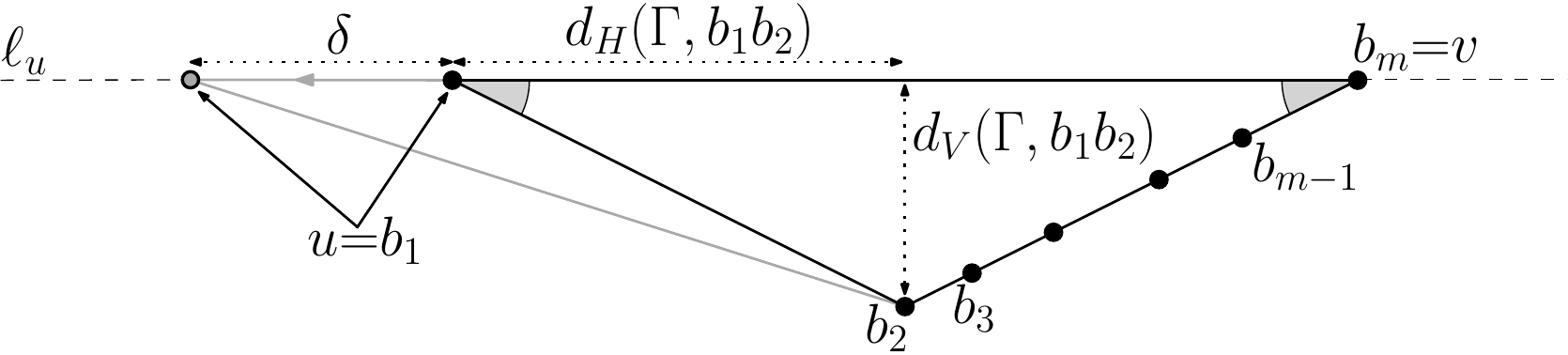}
	\caption{The Cycle Case of the algorithm for the proof of Theorem~\ref{th:main-aux}. The gray angles are $\frac{\alpha}{2}$.}
	\label{fi:basecase}
\end{figure}

\begin{lemma} \label{le:base-case}
For any $\delta\geq 0$, the drawing $\Gamma_{\delta}$ constructed in the Cycle Case satisfies Properties~1--6 of Theorem~\ref{th:main-aux}. 	
\end{lemma} 

\begin{proof}
Properties~1 and~2 are trivially satisfied. 

Concerning Property~3, by construction the edge $b_ib_{i+1}$ has slope $\frac{\alpha}{2}\in (0;\alpha)$ for $i=2,\dots,m-1$. Further, the edge $b_1b_2$ has slope $-\arctan\frac{d_V(\Gamma,b_1b_2)}{\delta+d_H(\Gamma,b_1b_2)}$, which is smaller than $0$, given that $d_V(\Gamma,b_1b_2), d_H(\Gamma,b_1b_2)>0$ and $\delta\geq 0$, and larger than or equal to $-\arctan\frac{d_V(\Gamma,b_1b_2)}{d_H(\Gamma,b_1b_2)}=-\frac{\alpha}{2}$, hence it is in $(-\alpha;0)$. This implies that $\Gamma_{\delta}$ satisfies Property~3.

Concerning Property~4, let $x=b_i$ with $i<m$. Then a path $P_x$ satisfying the requirements can be defined as $P_x=(b_1,b_m)$ if $i=1$ or as $P_x=(b_i,b_{i+1},\dots,b_m)$ if $i>1$. In the former case, the only edge of $P_x$ has slope $0\in (-\alpha;\alpha)$; in the latter case, all the edges of $P_x$ have slope $\frac{\alpha}{2}\in (-\alpha;\alpha)$ and $P_x$ does not pass through $u$. Hence $\Gamma_{\delta}$ satisfies Property~4. 

Concerning Property~5, let $x=b_i$ with $i>1$. Then a path $Q_x$ satisfying the requirements can be defined as $Q_x=(b_i,b_{i-1},\dots,b_1)$. Any edge $b_jb_{j-1}$ with $j\geq 3$ has slope $\pi+\frac{\alpha}{2}\in (\pi-\alpha;\pi+\alpha)$ and edge $b_2b_1$ has slope $\pi-\arctan\frac{d_V(\Gamma,b_1b_2)}{\delta+d_H(\Gamma,b_1b_2)}$, which is smaller than $\pi$, given that $d_V(\Gamma,b_1b_2), d_H(\Gamma,b_1b_2)>0$ and $\delta\geq 0$, and larger than or equal to $\pi-\arctan\frac{d_V(\Gamma,b_1b_2)}{d_H(\Gamma,b_1b_2)}=\pi-\frac{\alpha}{2}$, hence it is in $(\pi-\alpha;\pi+\alpha)$. This implies that $\Gamma_{\delta}$ satisfies Property~5. 

Finally we deal with Property~6. Let $x=b_i$ and $y=b_j$, for some $1\leq i,j \leq m$. 

\begin{itemize}
	\item If $2\leq i,j \leq m$ and $i<j$ (and $j<i$), then the path $P_{xy}=(b_i,b_{i+1},\dots,b_j)$ (resp.\ $P_{xy}=(b_i,b_{i-1},\dots,b_j)$) is distance-decreasing in $\Gamma_{\delta}$. Namely, it suffices to observe that the vertex $b_{h+1}$ (resp.\ $b_{h-1}$) lies on the open straight-line segment $\overline{b_{h}b_j}$ for $h=i,i+1,\dots,j-2$ (resp.\ for $h=i,i-1,\dots,j+2$). Further, $P_{xy}$ does not pass through $u$. 
	\item If $i=1$ and $j\geq 3$, then the path $P_{xy}=(b_1,b_2,\dots,b_j)$ is distance-decreasing in $\Gamma_{\delta}$. Namely, since the angles $\measuredangle b_2b_1b_m$ and $\measuredangle b_1b_mb_2$ are $\frac{\alpha}{2}$ by construction, the angle $\measuredangle b_mb_2b_1 = \measuredangle b_jb_2b_1$ is equal to $\pi-\alpha$ in $\Gamma$, and to at least $\pi-\alpha$ in $\Gamma_{\delta}$. Since by assumption $\alpha<\frac{\pi}{4}$, it follows that $\measuredangle b_jb_2b_1$ is the largest angle of the triangle $\Delta b_1 b_2 b_j$ in $\Gamma_{\delta}$, hence $d(\Gamma_{\delta},b_1b_j)>d(\Gamma_{\delta},b_2b_j)$. That $(b_2,b_3,\dots,b_j)$ is distance-decreasing can be proved as in the previous point. 
	\item If $j=1$ and $i\geq 3$, then the path $P_{xy}=(b_i,b_{i-1},\dots,b_1)$ is distance-decreasing in $\Gamma_{\delta}$. In order to prove that, it suffices to argue that $d(\Gamma_{\delta},b_1b_h)>d(\Gamma_{\delta},b_1b_{h-1})$ for any $h=3,4,\dots, i$. Since$\measuredangle b_hb_2b_1$ is at least $\pi-\alpha$ in $\Gamma_{\delta}$ (as from the previous point), the angle $\measuredangle b_hb_{h-1}b_1$ is also at least $\pi-\alpha$. Since by assumption $\alpha<\frac{\pi}{4}$, it follows that $\measuredangle b_hb_{h-1}b_1$ is the largest angle of the triangle $\Delta b_1 b_{h-1} b_h$ in $\Gamma_{\delta}$, hence $d(\Gamma_{\delta},b_1b_h)>d(\Gamma_{\delta},b_1b_{h-1})$. 
\end{itemize}

This concludes the proof of the lemma.
\end{proof}


We now discuss the inductive cases. In Case~A the path $\tau_{uv}(G)$ coincides with the edge $uv$, while in Case~B it does not. We discuss {\bf Case~A} first. Let $G'=G-uv$, where $G'$ consists of a sequence of graphs $G_1,\dots,G_k$, with $k\geq 1$, satisfying the properties described in Lemma~\ref{le:decomposition-A}. Our construction is different if $k=1$ and $k\geq 2$. 

Suppose first that {$\bf k=1$}; by Lemma~\ref{le:decomposition-A} the triple $(G'=G_1,u,v)$ is a strong circuit graph (and $G_1$ is not a single edge, as otherwise we would be in the Trivial Case). Apply induction in order to construct a straight-line drawing $\Gamma'$ of $G'$ with $\frac{\alpha}{2}$ as a parameter. Let $\tau_{uv}(G')=(u=a_1,a_2,\dots,a_t=v)$. By Property~2 the path $\tau_{uv}(G')$ lies on a horizontal line $\ell_u$ in $\Gamma'$ with $u$ to the left of $v$. Let $Y>0$ be the minimum distance in $\Gamma'$ of any vertex strictly below $\ell_u$ from $\ell_u$. Let 

$$\varepsilon=\frac{1}{2}\min\{\varepsilon^*_{\Gamma'},Y,\tan(\alpha) \cdot d(\Gamma',a_1a_2),\tan(\alpha) \cdot d(\Gamma',a_2a_t)\}.$$ 

We construct a straight-line drawing $\Gamma$ of $G$ from $\Gamma'$ as follows; refer to Fig.~\ref{fi:inductive1A-statement}. Decrease the $y$-coordinate of the vertex $a_2$ by $\varepsilon$. Further, decrease the $y$-coordinate of the vertex $a_i$, with $i=3,4,\dots,t-1$, so that it ends up on the straight-line segment $\overline{a_2a_t}$. Draw $uv$ as a straight-line segment. We have the following. 

\begin{figure}[htb]
	\centering
		\includegraphics[scale=0.6]{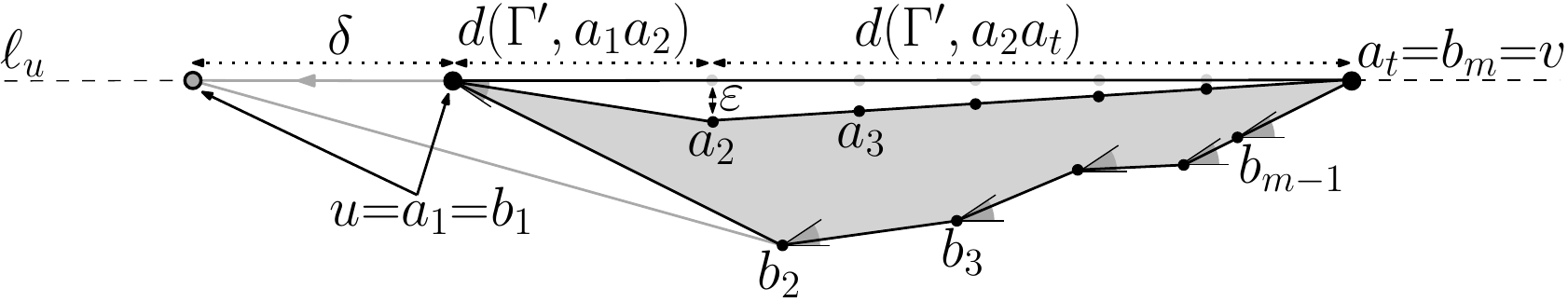}
	\caption{The straight-line drawing $\Gamma$ of $G$  in Case~A if $k=1$.}
	\label{fi:inductive1A-statement}
\end{figure}

\begin{lemma} \label{le:case1A}
For any $\delta\geq 0$, the drawing $\Gamma_{\delta}$ constructed in Case~A if $k=1$ satisfies Properties~1--6 of Theorem~\ref{th:main-aux}.	
\end{lemma}

\begin{proof}
Concerning Property~1, note first that $\Gamma$ is planar, given that $\varepsilon<\varepsilon^*_{\Gamma'}$. Since $\Gamma_{\delta}$ and $\Gamma$ coincide, except for the position of the vertex $u$, we only need to prove that no edge incident to $u$ crosses any other edge in $\Gamma_{\delta}$. Then consider any two edges $uu'$ and $ww'$ with $u',w,w'\in V(G)$ (possibly $w=u$ or $w=u'$) and suppose, for a contradiction, that they cross or overlap in $\Gamma_{\delta}$. 

\begin{figure}[htb]
	\centering
	\includegraphics[scale=0.7]{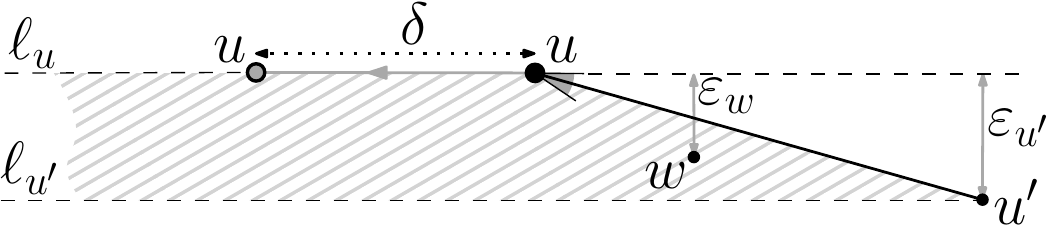}
	\caption{Illustration for the proof that the edges $uu'$ and $ww'$ do not cross in $\Gamma_{\delta}$.}
	\label{fi:inductive1A-crossings}
\end{figure}

Refer to Fig.~\ref{fi:inductive1A-crossings}. If $u'=v$, then $uu'$ and $ww'$ do not cross in $\Gamma_{\delta}$, given that no vertex other than $u$ and $v$ lies on or above $\ell_u$ in $\Gamma_{\delta}$. We can hence assume that $u'\neq v$ and that $y(u')<y(u)$. By Properties~1--3, we have that $\Gamma'$ lies in the closed wedge that is delimited by the half-lines starting at $u$ with slopes $0$ and $-\frac{\pi}{4}$. It follows that $x(u')>x(u)$ in $\Gamma'$, $\Gamma$, and $\Gamma_{\delta}$ (given that every vertex has the same $x$-coordinate in $\Gamma'$, $\Gamma$, and $\Gamma_{\delta}$, except for $u$, whose $x$-coordinate might be smaller in $\Gamma_{\delta}$ than in $\Gamma'$ and $\Gamma$). Consider the unbounded region $R$ of the plane that is delimited by $\ell_u$ from above, by the horizontal line $\ell_{u'}$ through $u'$ from below, and by the representation of the edge $uu'$ in $\Gamma$ from the right. For any value $\delta>0$, we have that $uu'$ lies in the interior of $R$ (except at points $u$ and $u'$) in $\Gamma_{\delta}$, hence if $uu'$ and $ww'$ cross in $\Gamma_{\delta}$ then at least one end-vertex of $ww'$, say $w$, lies in the interior of $R$, given that $y(w),y(w')\leq y(u)$ and that $ww'$ does not cross $uu'$ in $\Gamma$. This implies that $x(u)<x(w)<x(u')$ in $\Gamma'$, $\Gamma$, and $\Gamma_{\delta}$. We now distinguish four cases, based on whether $u'$ and/or $w$ belong to $V(\tau_{uv}(G'))$.

\begin{itemize}
\item If $u', w \in V(\tau_{uv}(G'))$, then by Property~2 we have that $u'$ and $w$ lie on $\ell_u$ in $\Gamma'$. However, since $x(u)<x(w)<x(u')$, it follows that the edge $uu'$ overlaps the vertex $w$ in $\Gamma'$, a contradiction to Property~1 of $\Gamma'$. 
\item If $u'\in V(\tau_{uv}(G'))$ and $w\notin V(\tau_{uv}(G'))$, then when transforming $\Gamma'$ into $\Gamma$ the $y$-coordinate of $u'$ has been decreased by a value $\varepsilon_{u'}\leq \varepsilon$ which is larger than the distance $\varepsilon_w\geq Y$ between $w$ and $\ell_u$. This contradicts $\varepsilon\leq \frac{Y}{2}<Y$.
\item If $u'\notin V(\tau_{uv}(G'))$ and $w\in V(\tau_{uv}(G'))$, then when transforming $\Gamma'$ into $\Gamma$ the $y$-coordinate of $w$ has been decreased by a value $\varepsilon_w\leq \varepsilon$. The point $p$ on the edge $uu'$ with $x$-coordinate equal to $x(w)$ has $y$-coordinate larger than $y(w)$, hence the distance from $p$ to $\ell_u$ is a value $\varepsilon_p<\varepsilon_w$. This implies that the drawing obtained from $\Gamma'$ by decreasing the $y$-coordinate of $w$ by $\varepsilon_p$, while every other vertex stays put, is not planar, given that the edge $uu'$ overlaps the vertex $w$. However, since $\varepsilon_p<\varepsilon_w$, this contradicts $\varepsilon_w\leq \varepsilon<\varepsilon^*_{\Gamma'}$. 
\item Finally, if $u',w\notin V(\tau_{uv}(G'))$, then $u$, $u'$ and $w$ have the same positions in $\Gamma'$ and $\Gamma$. Consider the line through $u'$ and $w$; let $q$ be its intersection point with $\ell_u$ and let $\delta_q$ be the Euclidean distance between $q$ and $u$ in $\Gamma'$. Then the drawing $\Gamma'_{\delta_q}$ is not planar as the edge $uu'$ overlaps the vertex $w$. This contradicts Property~1 of $\Gamma'$. 
\end{itemize}

Concerning Property~2, note that $u$ and $v$ lie on the same horizontal line $\ell_u$ (with $u$ to the left of $v$) in $\Gamma$ since they do in $\Gamma'$ and since they have not been moved when transforming $\Gamma'$ into $\Gamma$. Since $\tau_{uv}(G)$ coincides with the edge $uv$, it follows that $\Gamma_\delta$ satisfies Property~2.

Property~3 is satisfied by $\Gamma_\delta$ since it is satisfied by $\Gamma'_\delta$ and since no vertex of $\beta_{uv}(G')=\beta_{uv}(G)$ moves when transforming $\Gamma'$ into $\Gamma$ (indeed, $\tau_{uv}(G')$ and $\beta_{uv}(G')$ do not share any vertex other than $u$ and $v$, given that $G'$ is $2$-connected).

We now discuss Property~4. Let $x\in V(G)$. If $x=u$, let $P_x=(u,v)$; then the only edge of $P_x$ has slope $0\in (-\alpha;\alpha)$. If $x\neq u$, then let $P'_x=(x=v_1,v_2,\dots,v_p=v)$ be a path in $G'$ such that the slope of $v_iv_{i+1}$ in $\Gamma'$ is in the interval $(-\frac{\alpha}{2};\frac{\alpha}{2})$, for $i=1,\dots,p-1$, and such that $u\notin V(P'_x)$. This path exists since $\Gamma'$ satisfies Property~4, by induction. We distinguish two cases.

\begin{itemize}
\item If no vertex of $P'_x-\{v\}$ belongs to $\tau_{uv}(G')$, then $P_x=P'_x$ satisfies the required properties. Indeed, no vertex other than those internal to $\tau_{uv}(G')$ moves when transforming $\Gamma'$ into $\Gamma$ and no vertex other than $u$ moves when transforming $\Gamma$ into $\Gamma_{\delta}$; thus, $P_x$ has the same representation (and in particular each edge of $P_x$ has the same slope) in $\Gamma'$ and $\Gamma_{\delta}$.

\begin{figure}[htb]
	\centering
	\includegraphics[scale=0.6]{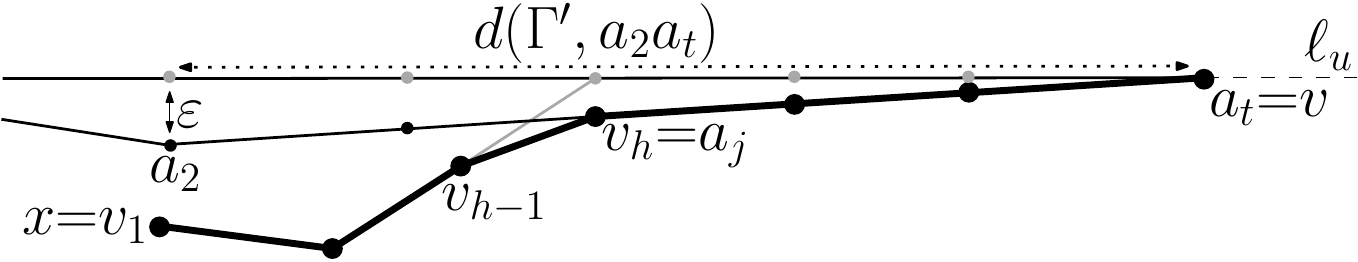}
	\caption{Illustration for the proof that the slope in $\Gamma_{\delta}$ of every edge in the path $P_x= (x=v_1,v_2,\dots,v_h=a_j,a_{j+1},\dots,a_t=v)$ is in $(-\alpha;\alpha)$. The path $P_x$ is thick.}
	\label{fi:inductive1A-pathToV}
\end{figure}

\item Otherwise, a vertex of $P'_x-\{v\}$ belongs to $\tau_{uv}(G')$; let $h$ be the smallest index such that $v_h=a_j$, for some $a_j\in V(\tau_{uv}(G'))-\{v\}$ and define $P_x= (x=v_1,v_2,\dots,v_h=a_j,a_{j+1},\dots,a_t=v)$. Refer to Fig.~\ref{fi:inductive1A-pathToV}. Note that $u\notin V(P_x)$, given that $u\notin V(P'_x)$. Hence, it suffices to argue about the slopes of the edges of $P_x$ in $\Gamma$ (rather than in $\Gamma_\delta$). For $i=1,\dots,h-2$, the slope of the edge $v_iv_{i+1}$ is in $(-\alpha;\alpha)$ in $\Gamma$ since it is in $(-\frac{\alpha}{2};\frac{\alpha}{2})\subset (-\alpha;\alpha)$ in $\Gamma'$ and since neither $v_i$ nor $v_{i+1}$ moves when transforming $\Gamma'$ into $\Gamma$. Further, for $i=j,\dots,t-1$, the slope of the edge $a_ia_{i+1}$ in $\Gamma$ is $\arctan\left(\frac{\varepsilon}{d(\Gamma',a_2a_t)}\right)$, which is in the interval $(0;\alpha)\subset(-\alpha;\alpha)$, given that $\varepsilon,d(\Gamma',a_2a_t)>0$ and that $\varepsilon<\tan(\alpha) \cdot d(\Gamma',a_2a_t)$. Finally, let $s'$ and $s$ be the slopes of the edge $v_{h-1}v_h$ in $\Gamma'$ and $\Gamma$, respectively. Since $v_{h-1}v_h\in E(P'_x)$, we have $s'\in (-\frac{\alpha}{2};\frac{\alpha}{2})$; since $\alpha\leq \frac{\pi}{4}$, this implies that $x(v_{h-1})<x(v_{h})$ in $\Gamma'$ and $\Gamma$ (note that the $x$-coordinates of the vertices do not change when transforming $\Gamma'$ into $\Gamma$).  Further, by Properties~1--4 of $\Gamma'$, we have that $v_{h-1}$ lies below $\ell_u$, which contains $v_h$; hence, $y(v_{h-1})<y(v_{h})$ in $\Gamma'$. Since the vertex $v_{h}$ moves down (while $v_{h-1}$ stays put) when transforming $\Gamma'$ into $\Gamma$, and since $\varepsilon\leq \frac{Y}{2}<d_V(\Gamma',v_{h-1}v_h)$, it follows that $0<s<s'$; hence $s\in (0;\frac{\alpha}{2})\subset(-\alpha;\alpha)$.
\end{itemize}

We next deal with Property~5. Let $x\in V(G)$ and let $Q'_x=(x=w_1,w_2,\dots,w_q=u)$ be a path in $\Gamma'_\delta$ such that the slope of $w_iw_{i+1}$ is in $(\pi-\frac{\alpha}{2};\pi+\frac{\alpha}{2})$, for $i=1,2,\dots,q-1$. This path exists since $\Gamma'_\delta$ satisfies Property~5, by induction. Similarly to the proof that $\Gamma_\delta$ satisfies Property~4, we distinguish two cases. If no vertex of $Q'_x-\{u\}$ belongs to $\tau_{uv}(G')$, then let $Q_x=Q'_x$ and observe that $Q_x$ satisfies the required properties in $\Gamma_\delta$ since $Q'_x$ does in $\Gamma'_\delta$. Otherwise, let $h$ be the smallest index such that $w_h=a_j$, for some $a_j\in V(\tau_{uv}(G'))-\{u\}$ and define $Q_x=(x=w_1,w_2,\dots,w_h=a_j,a_{j-1},\dots,a_1=u)$. Refer to Fig.~\ref{fi:inductive1A-pathToU}. For $i=1,\dots,h-2$, the slope of the edge $w_iw_{i+1}$ is in $(\pi-\alpha;\pi+\alpha)$ in $\Gamma_{\delta}$ since it is in $(\pi-\frac{\alpha}{2};\pi+\frac{\alpha}{2})\subset (\pi-\alpha;\pi+\alpha)$ in $\Gamma'_{\delta}$. Further, similarly to the proof that the edge $v_{h-1}v_h$ has slope in $(-\alpha;\alpha)$ and hence $\Gamma_{\delta}$ satisfies Property~4, we have that the edge $w_{h-1}w_h$ has slope $s\in (\pi-\alpha;\pi+\alpha)$ in $\Gamma_{\delta}$. Indeed, the slope $s'$ of the edge $w_{h-1}w_h$ in $\Gamma'$ is in the interval $(\pi-\frac{\alpha}{2};\pi+\frac{\alpha}{2})\subset (\pi-\alpha;\pi+\alpha)$ in $\Gamma'_{\delta}$, given that $w_{h-1}w_h$ belongs to $Q'_x$. Further, since $x(w_{h-1})>x(w_h)$ and $y(w_{h-1})<y(w_h)$, we have that $s'\in (\pi-\alpha;\pi)$. Since the vertex $w_h$ moves down while $w_{h-1}$ stays put when transforming $\Gamma'$ into $\Gamma$, and since $\varepsilon\leq \frac{Y}{2}<d_V(\Gamma',w_{h-1}w_h)$, we have that $s'<s<\pi$, hence $s\in (\pi-\alpha;\pi)\subset (\pi-\alpha;\pi+\alpha)$. For $i=j,j-1\dots,3$, the edge $a_ia_{i-1}$ has slope $\pi+\arctan\left(\frac{\varepsilon}{d(\Gamma',a_2a_t)}\right)$, which is larger than $\pi$, given that $\varepsilon,d(\Gamma',a_2a_t)>0$, and smaller than $\pi+\alpha$, given that $\varepsilon<\tan(\alpha) \cdot d(\Gamma',a_2a_t)$. Finally, the edge $a_2a_1$ has slope $\pi - \arctan\left(\frac{\varepsilon}{\delta+d(\Gamma',a_1a_2)}\right)$, which is smaller than $\pi$, given that $\varepsilon,d(\Gamma',a_1a_2)>0$ and $\delta\geq 0$, and larger than $\pi-\alpha$, given that $\frac{\varepsilon}{\delta+d(\Gamma',a_1a_2)}\leq \frac{\varepsilon}{d(\Gamma',a_1a_2)}$ and that $\varepsilon<\tan(\alpha) \cdot d(\Gamma',a_1a_2)$.

\begin{figure}[htb]
	\centering
	\includegraphics[scale=0.6]{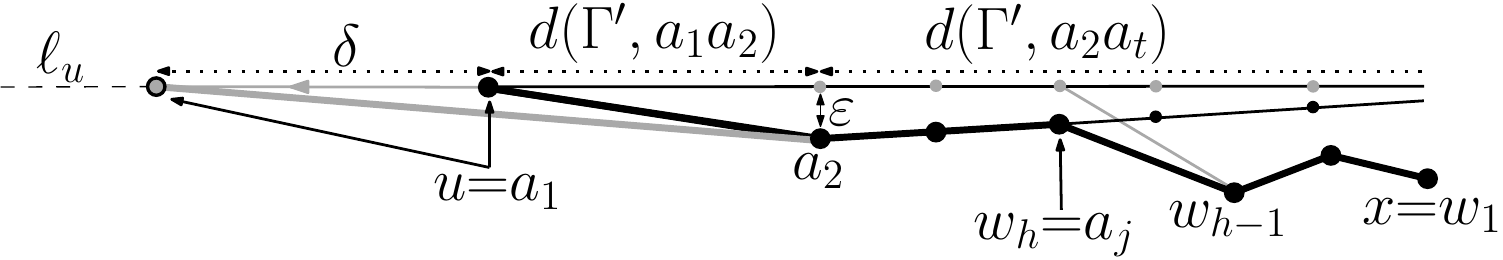}
	\caption{Illustration for the proof that the slope in $\Gamma_{\delta}$ of every edge in the path $Q_x= (x=w_1,w_2,\dots,w_h=a_j,a_{j-1},\dots,a_1=u)$ is in $(\pi-\alpha;\pi+\alpha)$. The path $Q_x$ is thick.}
	\label{fi:inductive1A-pathToU}
\end{figure}

Finally, we deal with Property~6. Consider any two vertices $x,y\in V(G)$. 

\begin{itemize}
\item First, assume that $x,y\neq u$. By induction, there exists a path $P_{xy}$ from $x$ to $y$ in $G'$ that is distance-decreasing in $\Gamma'$ with $u\notin V(P_{xy})$. By Lemma~\ref{le:perturbation-preserves-greedy} and since, for every vertex $z\in V(G)$, the Euclidean distance between the positions of $z$ in $\Gamma'$ and $\Gamma$ is at most $\varepsilon<\varepsilon^*_{\Gamma'}$, we have that $P_{xy}$ is also distance-decreasing in $\Gamma$. Further, since all the vertices other than $u$ have the same position in $\Gamma$ and $\Gamma_{\delta}$, it follows that $P_{xy}$ is a distance-decreasing path from $x$ to $y$ not passing through $u$ in $\Gamma_{\delta}$. 
\item Second, suppose that $y=u$. Consider the path $Q_x$ in $G$ from Property~5, whose every edge has slope in $(\pi-\alpha;\pi+\alpha)$ in $\Gamma_{\delta}$. Since $\alpha\leq \frac{\pi}{4}$, it follows that $Q_x$ is a $\pi$-path (according to the definition in~\cite{dfg-icgps-15}) or is $\pi$-monotone (according to the definition in~\cite{bbcklv-gt-16}), where for some angle $\beta$ a path $(q_1,q_2,\dots,q_r)$ is a {\em $\beta$-path} or equivalently is {\em $\beta$-monotone} if every edge $q_iq_{i+1}$ has slope in the interval $(\beta-\frac{\pi}{4};\beta+\frac{\pi}{4})$. In~\cite[Lemma 3]{dfg-icgps-15} it is proved that a $\beta$-path is distance-decreasing (in fact, it satisfies a much stronger property, namely it is {\em increasing-chord}); hence, $Q_x$ is distance-decreasing in $\Gamma_{\delta}$. 
\item Finally, suppose that $x=u$ and consider a path $P_{xy}$ from $x$ to $y$ in $G'$ that is distance-decreasing in $\Gamma'$. We prove that $P_{xy}$ is distance-decreasing in $\Gamma_{\delta}$, as well. Let $xx'$ be the edge of $P_{xy}$ incident to $x$. Differently from the case in which $x,y\neq u$, we cannot directly apply Lemma~\ref{le:perturbation-preserves-greedy}, given that it is not guaranteed that $\varepsilon<\varepsilon^*_{\Gamma'_{\delta}}$. However, since for every vertex $z\in V(P_{xy})$ the Euclidean distance between the positions of $z$ in $\Gamma'$ and $\Gamma$ is at most $\varepsilon<\varepsilon^*_{\Gamma'}$, by Lemma~\ref{le:perturbation-preserves-greedy} we have that $P_{xy}$ is distance-decreasing in $\Gamma$. Further, the path obtained from $P_{xy}$ by removing the vertex $x=u$ and the edge $xx'$ has the same representation in $\Gamma_{\delta}$ and $\Gamma$, given that it does not contain $u$, hence it is distance-decreasing in $\Gamma_{\delta}$. Thus, it only remains to show that $d(\Gamma_{\delta},xy)>d(\Gamma_{\delta},x'y)$. First, since $x=u$, we have that $x',y\neq u$. Hence, $d(\Gamma_{\delta},x'y)=d(\Gamma,x'y)$. Second, denote by $u_{\Gamma}$ and $u_{\Gamma_{\delta}}$ the positions of $u$ in $\Gamma$ and $\Gamma_{\delta}$, respectively. By Properties~1--3, the entire drawing $\Gamma$, and in particular vertex $y$, lies in the closed wedge that is delimited by the half-lines starting at $u_\Gamma$ and with slopes $0$ and $-\alpha$. Then the angle incident to $u_{\Gamma}$ in the triangle $\Delta yu_{\Gamma} u_{\Gamma_{\delta}} $ is at least $\pi-\alpha>\frac{\pi}{2}$, hence the straight-line segment between $u_{\Gamma_{\delta}}$ and $y$ is the longest side of that triangle. It follows that $d(\Gamma_{\delta},xy)\geq d(\Gamma,xy)$. Thus, $d(\Gamma_{\delta},xy)\geq d(\Gamma,xy)>d(\Gamma,x'y)=d(\Gamma_{\delta},x'y)$, where the second inequality holds true since $P_{xy}$ is distance-decreasing in $\Gamma$. Hence, $P_{xy}$ is distance-decreasing in $\Gamma_{\delta}$. 
\end{itemize}
This concludes the proof of the lemma.
\end{proof}


We now discuss the case in which {$\bf k\geq 2$}. Refer to Fig.~\ref{fi:inductive1B-statement}. By Lemma~\ref{le:decomposition-A}, for $i=1,\dots,k$, the triple $(G_i,u_{i-1},u_i)$ is a strong circuit graph, where $u_0=u$, $u_k=v$, and $u_i$ is the only vertex shared by $G_i$ and $G_{i+1}$, for $i=1,\dots,k-1$.

\begin{figure}[htb]
	\centering
	\includegraphics[scale=0.6]{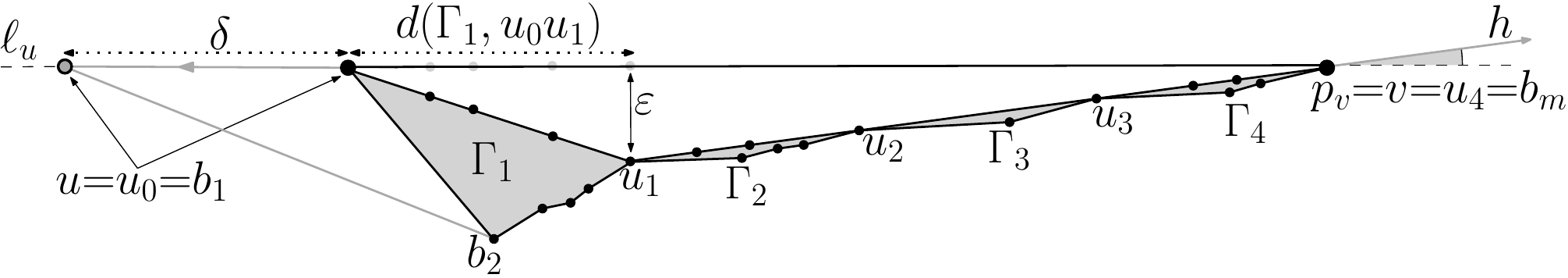}
	\caption{The straight-line drawing $\Gamma$ of $G$ in Case~A if $k\geq 2$. In this example $k=4$. The gray angle in the drawing is $\frac{\alpha}{2}$.}
	\label{fi:inductive1B-statement}
\end{figure}


If $G_1$ is a single edge, then apply induction in order to construct a straight-line drawing $\Gamma_1$ of $G_1$ and define $\varepsilon=\frac{1}{2}\min\{\varepsilon^*_{\Gamma_1},\tan(\alpha) \cdot d(\Gamma_1,u_0u_1)\}$.

If $G_1$ is not a single edge, then apply induction in order to construct a straight-line drawing $\Gamma_1$ of $G_1$ with $\frac{\alpha}{2}$ as a parameter. By Property~2 of $\Gamma_1$, the path $\tau_{u_0u_1}(G_1)$ lies on a horizontal line $\ell_u$. Let $Y>0$ be the minimum distance in $\Gamma_1$ of any vertex strictly below $\ell_u$ from $\ell_u$. Let $\varepsilon=\frac{1}{2}\min\{\varepsilon^*_{\Gamma_1},Y,\tan(\alpha) \cdot d(\Gamma_1,u_0u_1)\}$.

In both cases, decrease the $y$-coordinate of $u_1$ by $\varepsilon$. Further, decrease the $y$-coordinate of every internal vertex of the path $\tau_{u_0u_1}(G_1)$, if any, so that it ends up on the straight-line segment $\overline{u_0u_1}$.

Now consider a half-line $h$ with slope $s=\frac{\alpha}{2}$ starting at $u_1$. Denote by $p_v$ the point at which $h$ intersects the horizontal line $\ell_u$ through $u$. For $i=2,\dots,k$, apply induction in order to construct a straight-line drawing $\Gamma_i$ of $G_i$ with $\frac{\alpha}{3}$ as a parameter (if $G_i$ is a single edge, then the parameter does not matter). Uniformly scale the drawings $\Gamma_2,\dots,\Gamma_k$ so that the Euclidean distance between $u_{i-1}$ and $u_i$ is equal to $\frac{d(\Gamma_1,u_1 p_v)}{k-1}$. For $i=2,\dots,k$, rotate the scaled drawing $\Gamma_i$ around $u_{i-1}$ counter-clockwise by $s$ radians. Translate the scaled and rotated drawings $\Gamma_2,\dots,\Gamma_k$ so that the representations of $u_i$ in $\Gamma_i$ and $\Gamma_{i+1}$ coincide, for $i=1,\dots,k-1$. Finally, draw the edge $uv$ as a straight-line segment. This completes the construction of a drawing $\Gamma$ of $G$. We have the following. 

\begin{lemma} \label{le:case1B}
	For any $\delta\geq 0$, the drawing $\Gamma_{\delta}$ constructed in Case~A if $k\geq 2$ satisfies Properties~1--6 of Theorem~\ref{th:main-aux}.	
\end{lemma}

\begin{proof}
Throughout the proof, we denote by $\Gamma_{1,\delta}$ the drawing obtained from $\Gamma_1$ by moving the position of the vertex $u_0=u$ by $\delta$ units to the left (where $\Gamma_1$ is understood as the drawing of $G_1$ in which the vertices of $\tau_{u_0u_1}(G_1)$ all lie on $\ell_u$).
	
We first prove Property~2. Because of the uniform scaling which has been applied to $\Gamma_2,\dots,\Gamma_k$, we have that $d(\Gamma,u_{i-1}u_i)=\frac{d(\Gamma,u_1 p_v)}{k-1}$ for $i=2,\dots,k$. Since the vertices $u_1,\dots,u_k$ all lie on $h$, we have that $d_V(\Gamma,u_{i-1}u_i)=\frac{d_V(\Gamma,u_1 p_v)}{k-1}$ for $i=2,\dots,k$. Hence, the $y$-coordinate of $v=u_k$ is equal to $y(u_1)+d_V(\Gamma,u_1 p_v)=y(u)$, which implies that $u$ and $v$ lie on $\ell_u$ in $\Gamma$ and $\Gamma_\delta$. Further, $u$ is to the left of $v$ in $\Gamma$, given that $\Gamma_1$ satisfies Property~2 and that $0<s<\frac{\pi}{2}$. Since $\tau_{uv}(G)$ coincides with the edge $uv$, it follows that $\Gamma_\delta$ satisfies Property~2.

We next prove that $\Gamma_\delta$ satisfies Property~3. Observe that $\beta_{uv}(G)=\beta_{u_0u_1}(G_1)\cup \beta_{u_1u_2}(G_2)\cup \dots \beta_{u_{k-1}u_k}(G_k)$. We first argue about the slope of the edge $b_1b_2$. 
	
\begin{itemize}
	\item If $G_1$ is a single edge $u_0u_1$, then we have $u_0=b_1$ and $u_1=b_2$. Then the slope of the edge $b_1b_2$ in $\Gamma_{\delta}$ is $-\arctan\left(\frac{\varepsilon}{\delta+d(\Gamma_1,u_0u_1)}\right)$, which is smaller than $0$, given that $\varepsilon,d(\Gamma_1,u_0u_1)>0$ and $\delta\geq 0$, and larger than $-\alpha$, given that $\delta\geq 0$ and that $\varepsilon<\tan{\alpha}\cdot d(\Gamma_1,u_0u_1)$, hence it is in $(-\alpha;0)$. 
	\item If $G_1$ has more than two vertices, then $\beta_{u_0u_1}(G_1)$ is not a single edge $u_0u_1$, by Property~(c) of $(G_1,u_0,u_1)$; hence, $u_0=b_1$ and $b_2\neq u_1$. Since $\Gamma_{1,\delta}$ satisfies Property~3 and since $u_1$ is the only vertex of $\beta_{u_0u_1}(G_1)$ whose positions in $\Gamma_{1,\delta}$ and $\Gamma_{\delta}$ do not coincide, it follows that the slope of $b_1b_2$ in $\Gamma_{\delta}$ is in the interval $(-\alpha;0)$ since it is in the interval $(-\frac{\alpha}{2};0)\subset (-\alpha;0)$ in $\Gamma_{1,\delta}$. 
\end{itemize}

We now argue about the slope $s_j$ of the edge $b_jb_{j+1}$ in $\Gamma_{\delta}$, for any $j=2,\dots,m-1$.

\begin{itemize}
	\item If $b_jb_{j+1}$ coincides with a graph $G_i$, then $s_j=s=\frac{\alpha}{2}\in (0;\alpha)$.
	\item If $b_jb_{j+1}$ belongs to a graph $G_i$ with $|V(G_i)|\geq 3$, with $i\geq 2$, and with $b_j\neq u_{i-1}$, then $s_j$ is given by the slope $b_jb_{j+1}$ has in $\Gamma_i$, which is in $(0;\frac{\alpha}{3})$ by Property~3 of $\Gamma_i$, plus $s$, which results from the rotation of $\Gamma_i$. Hence $s_j\in (\frac{\alpha}{2};\frac{5\alpha}{6})\subset (0;\alpha)$.   
	\item If $b_jb_{j+1}$ belongs to a graph $G_i$ with $|V(G_i)|\geq 3$, with $i\geq 2$, and with $b_j= u_{i-1}$, then $s_j$ is given by the slope $b_jb_{j+1}$ has in $\Gamma_i$, which is in $(-\frac{\alpha}{3};0)$ by Property~3 of $\Gamma_i$, plus $s$, which results from the rotation of $\Gamma_i$. Hence $s_j\in (\frac{\alpha}{6};\frac{\alpha}{2})\subset (0;\alpha)$.  
	\item If $b_jb_{j+1}$ belongs to $G_1$, if $|V(G_1)|\geq 3$, and if $b_{j+1}\neq u_1$, then since $\Gamma_{1,\delta}$ satisfies Property~3 and since $u_1$ is the only vertex of $\beta_{u_0u_1}(G_1)$ whose positions in $\Gamma_{1,\delta}$ and $\Gamma_{\delta}$ do not coincide, it follows that $s_j\in (0; \alpha)$ since the slope of $b_jb_{j+1}$ in $\Gamma_{1,\delta}$ is in $(0;\alpha)$.
	\item Finally, assume that $b_jb_{j+1}$ belongs to $G_1$, that $|V(G_1)|\geq 3$, and that $b_{j+1}= u_1$. Note that $b_j\neq u$, given that $j\geq 2$, hence by Property~3 of $\Gamma_1$ we have that $x(b_j)<x(b_{j+1})$ and that $y(b_j)<y(b_{j+1})$ in $\Gamma_1$. Note that the positions of $b_{j}$ in $\Gamma_1$ and $\Gamma_{\delta}$ coincide, given that $b_j\notin V(\tau_{u_0u_1}(G_1))$; further, $b_{j+1}$ moves down by $\varepsilon$ when transforming $\Gamma_1$ into $\Gamma_{\delta}$, however its $x$-coordinate stays unchanged; this implies that $s_j$ is smaller than the slope of $b_jb_{j+1}$ in $\Gamma_1$, hence smaller than $\alpha$. Since $\varepsilon\leq \frac{Y}{2}<d_V(\Gamma_1,b_jb_{j+1})$ -- given that $u_1$ lies on $\ell_u$ in $\Gamma_1$ -- it follows that $y(b_j)<y(b_{j+1})$ holds true in $\Gamma_{\delta}$, hence $s_j>0$. Thus, $s_j\in (0; \alpha)$. 
\end{itemize}

We now prove Property~1. First, the edge $uv$ does not cross or overlap any other edge of $G$, since no vertex other than $u$ and $v$ lies on or above $\ell_u$ in $\Gamma$ and $\Gamma_{\delta}$. Hence, we only need to argue about crossings among edges in the graphs $G_1,\dots,G_k$. 

We first deal with $\Gamma$. For $i=1,\dots, k$, the inductively constructed drawing $\Gamma_i$ of $G_i$ is planar, by Property~1. Further, for $i=2,\dots, k$, the drawing of $G_i$ in $\Gamma$ is congruent to $\Gamma_i$, up to affine transformations (a uniform scaling, a rotation, and a translation), which preserve planarity. Moreover, since $\varepsilon< \varepsilon^*_{\Gamma_1}$, by Lemma~\ref{le:perturbation-preserves-planarity} we have that the drawing of $G_1$ in $\Gamma$ is planar, as well. It follows that no two edges in the same graph $G_i$ cross each other in $\Gamma$, for each $i=1,\dots, k$. Since $\Gamma$ satisfies Property~3, the path $\beta_{uv}(G)$ is represented in $\Gamma$ by a curve monotonically increasing in the $x$-direction from $u$ to $v$. Further, the path $\tau=\bigcup_{i=1}^{k} \tau_{u_{i-1}u_i}(G_i)$ is also represented in $\Gamma$ by a curve monotonically increasing in the $x$-direction from $u$ to $v$, since it is composed of the straight-line segment $\overline{u_0u_1}$, which has slope $-\arctan\left(\frac{\varepsilon}{d(\Gamma_1,u_0u_1)}\right)\in (-\frac{\pi}{2};0)$, and of the straight-line segment $\overline{u_1u_k}$, which has slope $s=\frac{\alpha}{2}\in (0;\frac{\pi}{2})$. Hence, for $i=1,2,\dots,k-1$, the vertical line through $u_i$ has the drawings of $G_1,\dots,G_i$ to its left and the drawings of $G_{i+1},\dots,G_k$ to its right in $\Gamma$. It follows that no two edges in distinct graphs $G_i$ and $G_j$ cross in $\Gamma$. This proves the planarity of $\Gamma$.

Since $\Gamma_{\delta}$ and $\Gamma$ coincide, except for the position of $u$, it remains to prove that no edge $uu'$ incident to $u$ with $u'\neq v$ crosses or overlaps any other edge in $\Gamma_{\delta}$. Since the vertical line through $u_1$ has the drawing of $G_1$ to its left and the drawings of $G_2,\dots,G_k$ to its right in $\Gamma_{\delta}$, such a crossing might only occur between $uu'$ and another edge $ww'$ of $G_1$. The proof that $uu'$ and $ww'$ do not cross or overlap is the same as in the proof of Lemma~\ref{le:case1A}, with $G_1$ playing the role of $G'$ and $\Gamma_1$ playing the role of $\Gamma'$. 
    
We now deal with Property~4. Let $x\in V(G)$. If $x=u$, let $P_x=(u,v)$; then the only edge of $P_x$ has slope $0\in (-\alpha;\alpha)$ in $\Gamma_{\delta}$. If $x=u_i$, for some $i\in \{1,\dots,k-1\}$, then let $P_x=\bigcup_{j=i+1}^{k} \tau_{u_{j-1}u_j}(G_j)$ and observe that all the edges of $P_x$ have slope $s=\frac{\alpha}{2}\in (-\alpha;\alpha)$; further $P_x$ does not pass through $u$. If $x\neq u_i$, for every $i\in \{0,1,\dots,k\}$, then $x$ belongs to a unique graph $G_i$, for some $i\in \{1,2,\dots,k\}$. We distinguish two cases.

\begin{itemize}
	\item Assume first that $i\geq 2$. Since $\Gamma_i$ satisfies Property~4, there exists a path $P^i_x$ from $x$ to $u_i$ in $G_i$ whose every edge has slope in $(-\frac{\alpha}{3};\frac{\alpha}{3})$ in $\Gamma_i$; then $P_x$ consists of $P^i_x$ and of the path $\bigcup_{j=i+1}^{k} \tau_{u_{j-1}u_j}(G_j)$. Since the drawing of $G_i$ in $\Gamma_{\delta}$ is congruent to $\Gamma_i$ up to a uniform scaling, a counter-clockwise rotation by $s=\frac{\alpha}{2}$ radians, and a translation, it follows that every edge of $P^i_x$ has slope in $(s-\frac{\alpha}{3};s+\frac{\alpha}{3})=(\frac{\alpha}{6};\frac{5\alpha}{6})\subset (-\alpha;\alpha)$ in $\Gamma_{\delta}$; further, as noted above, all the edges of $\bigcup_{j=i+1}^{k} \tau_{u_{j-1}u_j}(G_j)$ have slope $s=\frac{\alpha}{2}\in (-\alpha;\alpha)$. Hence, all the edges of $P_x$ have slope in $(-\alpha;\alpha)$; further, $P_x$ does not pass through $u$.
	
	\begin{figure}[htb]
		\centering
		\includegraphics[scale=0.6]{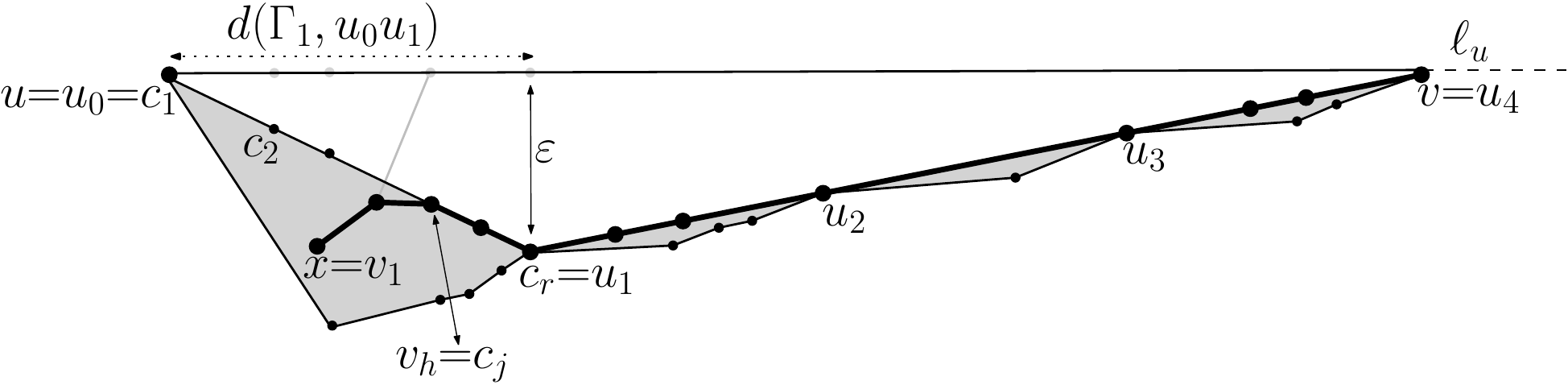}
		\caption{Illustration for the proof that the slope in $\Gamma_{\delta}$ of every edge in the path $P_x$ is in $(-\alpha;\alpha)$, in the case in which $x$ belongs to $G_1$. The path $P_x$ is thick.}
		\label{fi:inductive1B-pathToV}
	\end{figure}

	\item Assume next that $i=1$. Refer to Fig.~\ref{fi:inductive1B-pathToV}. Let $\tau_{u_0u_1}(G_1)=(u_0=c_1,c_2,\dots,c_r=u_1)$. Since $\Gamma_1$ satisfies Property~4, there exists a path $P^1_x=(x=v_1,v_2,\dots,v_p=u_1)$ from $x$ to $u_1$ in $G_1$, not passing through $u$, whose every edge has slope in $(-\frac{\alpha}{2};\frac{\alpha}{2})$ in $\Gamma_1$; let $h$ be the smallest index such that $v_h=c_j$, for some $j\in\{2,3,\dots,r\}$. Such an index $h$ exists (possibly $h=p$ and $j=r$). Then let $P_x$ consist of the paths $(x=v_1,v_2,\dots,v_h)$, $(v_h=c_j,c_{j+1},\dots,c_r)$, and $\bigcup_{j=2}^{k} \tau_{u_{j-1}u_j}(G_j)$. Since $u\notin V(P^1_x)$, we have that $u\notin V(P_x)$, hence it suffices to argue about the slopes of the edges of $P_x$ in $\Gamma$ rather than in $\Gamma_{\delta}$. 
	
	For $l=1,\dots,h-2$, the slope of $v_lv_{l+1}$ in $\Gamma$ is in the interval $(-\alpha;\alpha)$  since it is in the interval $(-\frac{\alpha}{2};\frac{\alpha}{2})\subset (-\alpha;\alpha)$ in $\Gamma_1$ and since neither $v_l$ nor $v_{l+1}$ moves when transforming $\Gamma_1$ into $\Gamma$. Further, for $l=j,\dots,r-1$, the slope of the edge $c_lc_{l+1}$ in $\Gamma$ is $-\arctan\left(\frac{\varepsilon}{d(\Gamma_1,u_0u_1)}\right)$, which is in the interval $(-\alpha;0)\subset(-\alpha;\alpha)$, given that $\varepsilon,d(\Gamma_1,u_0u_1)>0$ and that $\varepsilon<\tan(\alpha) \cdot d(\Gamma_1,u_0u_1)$. Moreover, as noted above, the edges of $\bigcup_{j=2}^{k} \tau_{u_{j-1}u_j}(G_j)$ have slope $s=\frac{\alpha}{2} \in (-\alpha;\alpha)$. Finally, let $\sigma_1$ and $\sigma$ be the slopes of the edge $v_{h-1}v_h$ in $\Gamma_1$ and $\Gamma$, respectively. Since $v_{h-1}v_h\in E(P^1_x)$, we have $\sigma_1 \in (-\frac{\alpha}{2};\frac{\alpha}{2})$; since $\alpha\leq \frac{\pi}{4}$, we have $x(v_{h-1})<x(v_{h})$ in $\Gamma_1$ and $\Gamma$ (note that the $x$-coordinates of the vertices do not change when transforming $\Gamma_1$ into $\Gamma$). Further, by Properties~1--4 of $\Gamma_1$, we have that $v_{h-1}$ lies below $\ell_u$, which contains $v_h$; hence, $y(v_{h-1})<y(v_{h})$ in $\Gamma_1$. Since the vertex $v_{h}$ moves down (while $v_{h-1}$ stays put) when transforming $\Gamma_1$ into $\Gamma$, and since $\varepsilon\leq \frac{Y}{2}<d_V(\Gamma_1,v_{h-1}v_h)$, it follows that $0<\sigma<\sigma_1$; hence $\sigma\in (0;\frac{\alpha}{2})\subset(-\alpha;\alpha)$. 
\end{itemize}

We now argue about Property~5. Let $x\in V(G)$. If $x=v$, let $Q_x=(v,u)$; then the only edge of $Q_x$ has slope $\pi\in (\pi-\alpha;\pi+\alpha)$ in $\Gamma_{\delta}$. If $x=u_i$, for some $i\in \{1,\dots,k-1\}$, then let $Q_x=\bigcup_{j=1}^{i} \beta_{u_j u_{j-1}}(G_j)$; recall that $\beta_{u_j u_{j-1}}(G_j)$ has the same vertices as $\tau_{u_{j-1} u_j}(G_j)$, however in the reverse linear order. Denote the vertices of $Q_x$ by $(x=w_1,w_2,\dots,w_q=u)$. Consider the edge $w_lw_{l+1}$, for any $1\leq l \leq q-1$. 

\begin{itemize}
	\item If $w_lw_{l+1}$ is in $G_j$, for some $j\geq 2$, then its slope in $\Gamma_{\delta}$ is $\pi+s=\pi+\frac{\alpha}{2}\in (\pi-\alpha;\pi+\alpha)$.
	\item If $w_lw_{l+1}$ is in $G_1$ and $l\leq q-2$, then its slope in $\Gamma_{\delta}$ is $\pi-\arctan\left(\frac{\varepsilon}{d(\Gamma_1,u_0u_1)}\right)$, which is in the interval $(\pi-\alpha;\pi)\subset (\pi-\alpha; \pi+\alpha)$, given that $\varepsilon,d(\Gamma_1,u_0u_1)>0$ and that $\varepsilon<\tan(\alpha) \cdot d(\Gamma_1,u_0u_1)$. 
	\item Finally, the slope of $w_{q-1}w_q$ in $\Gamma_{\delta}$ is $\pi-\arctan\left(\frac{\varepsilon}{\delta+d(\Gamma_1,u_0u_1)}\right)$, which is in the interval $(\pi-\alpha;\pi)\subset (\pi-\alpha; \pi+\alpha)$, given that $\varepsilon,d(\Gamma_1,u_0u_1)>0$, that $\delta\geq 0$, and that $\varepsilon<\tan(\alpha) \cdot d(\Gamma_1,u_0u_1)$. 
\end{itemize}

If $x\neq u_i$, for every $i\in \{0,1,\dots,k\}$, then $x$ belongs to a unique graph $G_i$, for some $i\in \{1,2,\dots,k\}$. We distinguish two cases. 

\begin{itemize}
	\item Assume first that $i\geq 2$. Since $\Gamma_i$ satisfies Property~5, there exists a path $Q^i_x$ from $x$ to $u_{i-1}$ in $G_i$ whose every edge has slope in $(\pi-\frac{\alpha}{3};\pi+\frac{\alpha}{3})$ in $\Gamma_i$; then $Q_x$ consists of $Q^i_x$ and of the path $\bigcup_{j=1}^{i-1} \beta_{u_ju_{j-1}}(G_j)$; denote the vertices of $Q_x$ by $(x=w_1,w_2,\dots,w_q=u)$. Consider the edge $w_lw_{l+1}$, for any $1\leq l \leq q-1$.  
	
	\begin{itemize}
		\item If $w_lw_{l+1}$ is in $G_i$, then it belongs to the path $Q^i_x$. Then $w_lw_{l+1}$ has slope in $(\pi-\frac{\alpha}{3};\pi+\frac{\alpha}{3})$ in $\Gamma_i$, hence it has slope $(\pi+s-\frac{\alpha}{3};\pi+s+\frac{\alpha}{3})=(\pi+\frac{\alpha}{6};\pi+\frac{5\alpha}{6})\subset (\pi-\alpha;\pi+\alpha)$ in $\Gamma_{\delta}$.
		\item If $w_lw_{l+1}$ is in $G_j$, for some $2\leq j\leq i-1$, then its slope in $\Gamma_{\delta}$ is $\pi+s=\pi+\frac{\alpha}{2}\in (\pi-\alpha;\pi+\alpha)$.
		\item If $w_lw_{l+1}$ is in $G_1$ and $l\leq q-2$, then its slope in $\Gamma_{\delta}$ is $\pi-\arctan\left(\frac{\varepsilon}{d(\Gamma_1,u_0u_1)}\right)$, which is in the interval $(\pi-\alpha; \pi+\alpha)$, as proved in the case $x=u_i$.
		 \item Finally, the slope of $w_{q-1}w_q$ in $\Gamma_{\delta}$ is $\pi-\arctan\left(\frac{\varepsilon}{\delta+d(\Gamma_1,u_0u_1)}\right)$, which is in the interval $(\pi-\alpha; \pi+\alpha)$, as proved in the case $x=u_i$. 
	\end{itemize}
	
	\item Assume next that $i=1$. Let $\tau_{u_0u_1}(G_1)=(u_0=c_1,c_2,\dots,c_r=u_1)$. Since $\Gamma_1$ satisfies Property~5, there exists a path $Q^1_x=(x=w_1,w_2,\dots,w_q=u)$ from $x$ to $u$ in $G_1$, whose every edge has slope in $(\pi-\frac{\alpha}{2};\pi+\frac{\alpha}{2})$ in $\Gamma_1$. Similarly to the proof that $\Gamma_{\delta}$ satisfies Property~5 in Lemma~\ref{le:case1A}, we distinguish two cases. If no vertex of $Q^1_x-\{u\}$ belongs to $\tau_{u_0u_1}(G_1)$, then let $Q_x=Q^1_x$ and observe that $Q_x$ satisfies the required properties in $\Gamma_\delta$ since $Q^1_x$ does in $\Gamma_{1,\delta}$. Otherwise, let $h$ be the smallest index such that $w_h=c_j$, for some $j\in \{2,3,\dots,r\}$ and define $Q_x=(x=w_1,w_2,\dots,w_h=c_j,c_{j-1},\dots,c_1=u)$. For $l=1,\dots,h-2$, the slope of the edge $w_lw_{l+1}$ is in $(\pi-\alpha;\pi+\alpha)$ in $\Gamma_{\delta}$ since it is in $(\pi-\alpha;\pi+\alpha)$ in $\Gamma_{1,\delta}$. Further, the edge $w_{h-1}w_h$ has slope in $(\pi-\alpha;\pi+\alpha)$ in $\Gamma_{\delta}$ since it has slope in that range in $\Gamma_1$ (given that it belongs to $Q^1_x$), since $x(w_{h-1})>x(w_h)$ and $y(w_{h-1})<y(w_h)$, since $u\neq w_{h-1},w_h$, and since $\varepsilon\leq \frac{Y}{2}<d_V(\Gamma_1,w_{h-1}w_h)$. Moreover, for $l=j,j-1\dots,3$, the slope of the edge $c_lc_{l-1}$ in $\Gamma_\delta$ is $\pi-\arctan\left(\frac{\varepsilon}{d(\Gamma_1,u_0u_1)}\right)$, which is in the interval $(\pi-\alpha; \pi+\alpha)$, as proved in the case $x=u_i$. Finally, the slope of the edge $c_2c_1$ in $\Gamma_\delta$ is $\pi - \arctan\left(\frac{\varepsilon}{\delta+d(\Gamma_1,u_0u_1)}\right)$, which is in the interval $(\pi-\alpha; \pi+\alpha)$, as proved in the case $x=u_i$.
\end{itemize}

We finally deal with Property~6. Consider any two distinct vertices $x,y\in V(G)$. 

If $x$ and $y$ belong to the same graph $G_i$, for some $i\in \{1,\dots,k\}$, then there exists a distance-decreasing path $P_{xy}$ from $x$ to $y$ in $\Gamma_i$, given that $\Gamma_i$ satisfies Property~6. If $i\in\{2,\dots, k\}$, the drawing of $G_i$ in $\Gamma_{\delta}$ is congruent to $\Gamma_i$, up to three affine transformations (a uniform scaling, a rotation, and a translation) that preserve the property of a path to be distance-decreasing; hence $P_{xy}$ is distance-decreasing in $\Gamma_{\delta}$ as well. If $i=1$, then the proof that $P_{xy}$ is distance-decreasing in $\Gamma_{\delta}$ is the same as the proof that $\Gamma_{\delta}$ satisfies Property~6 in Lemma~\ref{le:case1A}, with $\Gamma_1$ playing the role of $\Gamma'$. 

We can hence assume that $x$ and $y$ belong to two distinct graphs $G_i$ and $G_j$, respectively.

\begin{itemize}
	\item Suppose first that $2\leq i <j\leq k$. Then let $P_{xy}$ be the path composed of:
	\begin{itemize}
		\item a path $P^i_x$ in $G_i$ from $x$ to $u_i$ whose every edge has slope in $(-\frac{\alpha}{3}; \frac{\alpha}{3})$ in $\Gamma_i$;
		\item the path $\bigcup_{l=i+1}^{j-1} \tau_{u_{l-1}u_l}(G_l)$; and
		\item a path $P_{u_{j-1}y}$ in $G_j$ that is distance-decreasing in $\Gamma_j$. 
	\end{itemize}
	By induction, the paths $P^i_x$ and $P_{u_{j-1}y}$ exist since $\Gamma_i$ satisfies Property~4 and $\Gamma_j$ satisfies Property~6, respectively. We prove that $P_{xy}$ is distance-decreasing in $\Gamma_{\delta}$; note that $u\notin V(P_{xy})$. Let $P_{xy}=(z_1,z_2,\dots,z_s)$; then we need to prove that $d(\Gamma_{\delta},z_hz_s)>d(\Gamma_{\delta},z_{h+1}z_s)$, for $h=1,2,\dots,s-2$. We distinguish three cases. 
	\begin{itemize}
		\item If $z_hz_{h+1}$ is in $G_j$, then $(z_h,z_{h+1},\dots,z_s)$ is a sub-path of $P_{u_{j-1}y}$, hence it is distance-decreasing in $\Gamma_{\delta}$ since it is distance-decreasing in $\Gamma_j$ and since the drawing of $G_j$ in $\Gamma_{\delta}$ is congruent to $\Gamma_j$, up to three affine transformations (a uniform scaling, a rotation, and a translation) that preserve the property of a path to be distance-decreasing.
		
		\item If $z_hz_{h+1}$ is in $\tau_{u_{l-1}u_l}(G_l)$, for some $l\in\{i+1,i+2,\dots,j-1\}$, as in Fig.~\ref{fi:inductive1B-distancedecreasing}, then it has slope $s=\frac{\alpha}{2}$. Consider the line $\ell_{h}$ with slope $\frac{\pi+\alpha}{2}$ through $u_l$, oriented towards increasing $y$-coordinates. By Lemma~\ref{le:same-side}, this line has the drawings of $G_{l+1},G_{l+2},\dots,G_k$ to its right; this is because by Property~3 of $\Gamma_{\delta}$ every edge in $\beta_{u_l v}(G)$ has slope in the interval $(0;\alpha)$, where $\frac{-\pi+\alpha}{2}<0<\alpha<\frac{\pi+\alpha}{2}$, and because the path $\bigcup_{m=l+1}^{k} \tau_{u_{m-1}u_m}(G_m)$ has slope $s=\frac{\alpha}{2}$, where $\frac{-\pi+\alpha}{2}<\frac{\alpha}{2}<\frac{\pi+\alpha}{2}$. Further, by Lemma~\ref{le:same-side}, the line $\ell_{h}$ has the drawing of the path $\beta_{u_l u_{l-1}}(G_l)$ to its left; this is because every edge in $\beta_{u_l u_{l-1}}(G_l)$ has slope $s=\pi+\frac{\alpha}{2}$, where $\frac{\pi+\alpha}{2}<\pi+\frac{\alpha}{2}<\frac{3\pi+\alpha}{2}$. Then the line $\ell'_{h}$ parallel to $\ell_{h}$, passing through the midpoint of the edge $z_hz_{h+1}$, and oriented towards increasing $y$-coordinates has $\ell_{h}$ to its right, given that the path $\beta_{u_l u_{l-1}}(G_l)$ (and in particular the midpoint of the edge $z_hz_{h+1}$) is to the left of $\ell_{h}$, hence $\ell'_{h}$ has the drawings of $G_{l+1},G_{l+2},\dots,G_k$ (and in particular the vertex $z_s$) to its right. Since the half-plane to the right of $\ell'_{h}$ represents the locus of the points of the plane that are closer to $z_{h+1}$ than to $z_h$, it follows that $d(\Gamma_{\delta},z_hz_s)>d(\Gamma_{\delta},z_{h+1}z_s)$.  

		\begin{figure}[htb]
			\centering
			\includegraphics[scale=0.6]{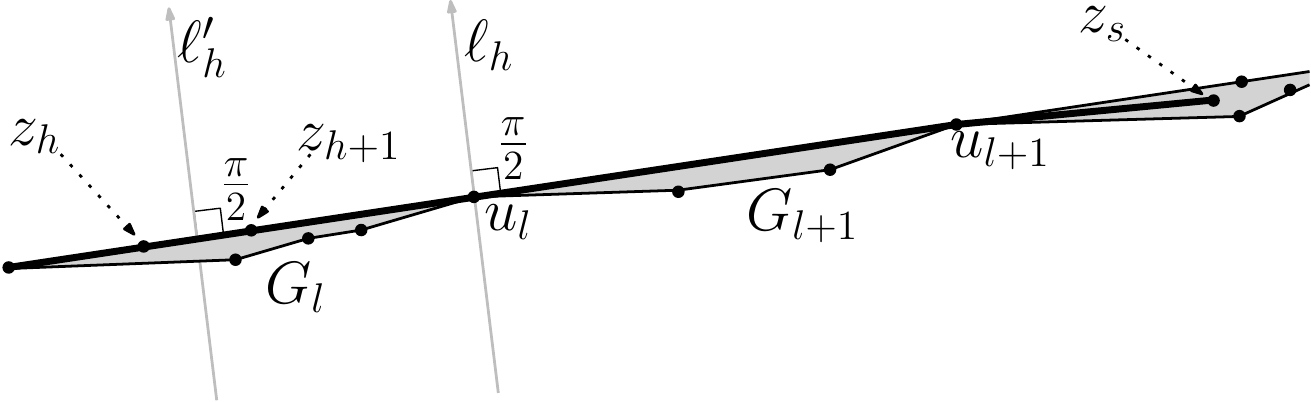}
			\caption{Illustration for the proof that $d(\Gamma_{\delta},z_hz_s)>d(\Gamma_{\delta},z_{h+1}z_s)$ if $z_hz_{h+1}$ is in $\tau_{u_{l-1}u_l}(G_l)$.}
			\label{fi:inductive1B-distancedecreasing}
		\end{figure}
		
		\item If $z_hz_{h+1}$ is in $P^i_x$, as in Fig.~\ref{fi:inductive1B-distancedecreasing2}, then by Property~4 it has slope in $(-\frac{\alpha}{3}; \frac{\alpha}{3})$ in $\Gamma_i$. Since $\Gamma_i$ is counter-clockwise rotated by $s$ radians in $\Gamma_{\delta}$, it follows that $z_hz_{h+1}$ has slope in $(s-\frac{\alpha}{3}; s+\frac{\alpha}{3})=(\frac{\alpha}{6};\frac{5\alpha}{6})$ in $\Gamma_{\delta}$. Consider the line $\ell_h$ that passes through $u_i$, that is directed towards increasing $y$-coordinates and that is orthogonal to the line through $z_h$ and $z_{h+1}$. Denote by $s_h$ the slope of $\ell_h$. Then $s_h\in (\frac{\pi}{2}+\frac{\alpha}{6};\frac{\pi}{2}+\frac{5\alpha}{6})$. By Lemma~\ref{le:same-side}, the line $\ell_h$ has the drawings of $G_{i+1},\dots,G_k$ to its right; this is because by Property~3 of $\Gamma_{\delta}$ every edge in $\beta_{u_i v}(G)$ has slope in $(0;\alpha)$ with $s_h-\pi<-\frac{\pi}{2}+\frac{5\alpha}{6}<0<\alpha<\frac{\pi}{2}+\frac{\alpha}{6}<s_h$ and because the path $\bigcup_{m=i+1}^{k} \tau_{u_{m-1}u_m}(G_m)$ has slope $s=\frac{\alpha}{2}$, where $s_h-\pi<-\frac{\pi}{2}+\frac{5\alpha}{6}<\frac{\alpha}{2}<\frac{\pi}{2}+\frac{\alpha}{6}<s_h$. Further, by Lemma~\ref{le:same-side}, the line $\ell_h$ has the drawings of $G_2,\dots,G_i$ to its left; this is because by Property~3 of $\Gamma_{\delta}$ every edge in $\tau_{u_i u_1}(G)$ has slope in $(\pi;\pi+\alpha)$ with $s_h<\frac{\pi}{2}+\frac{5\alpha}{6}<\pi<\pi+\alpha<\frac{3\pi}{2}+\frac{\alpha}{6}<\pi+s_h$ and because the path $\bigcup_{m=2}^{i} \beta_{u_{m}u_{m-1}}(G_m)$ has slope $s=\pi+\frac{\alpha}{2}$, where $s_h<\frac{\pi}{2}+\frac{5\alpha}{6}<\pi+\frac{\alpha}{2}<\frac{3\pi}{2}+\frac{\alpha}{6}<\pi+s_h$. Now consider the line $\ell'_{h}$ parallel to $\ell_{h}$, passing through the midpoint of the edge $z_hz_{h+1}$, and oriented towards increasing $y$-coordinates. This line has $\ell_{h}$ to its right, given that the drawing of $G_i$ (and in particular the midpoint of $z_hz_{h+1}$) is to the left of $\ell_{h}$ in $\Gamma_{\delta}$. Thus, $\ell'_{h}$ has the drawings of $G_{i+1},G_{i+2},\dots,G_k$ (and in particular the vertex $z_s$) to its right. Since the half-plane to the right of $\ell'_{h}$ represents the locus of the points of the plane that are closer to $z_{h+1}$ than to $z_h$, it follows that $d(\Gamma_{\delta},z_hz_s)>d(\Gamma_{\delta},z_{h+1}z_s)$.

		\begin{figure}[htb]
			\centering
			\includegraphics[scale=0.6]{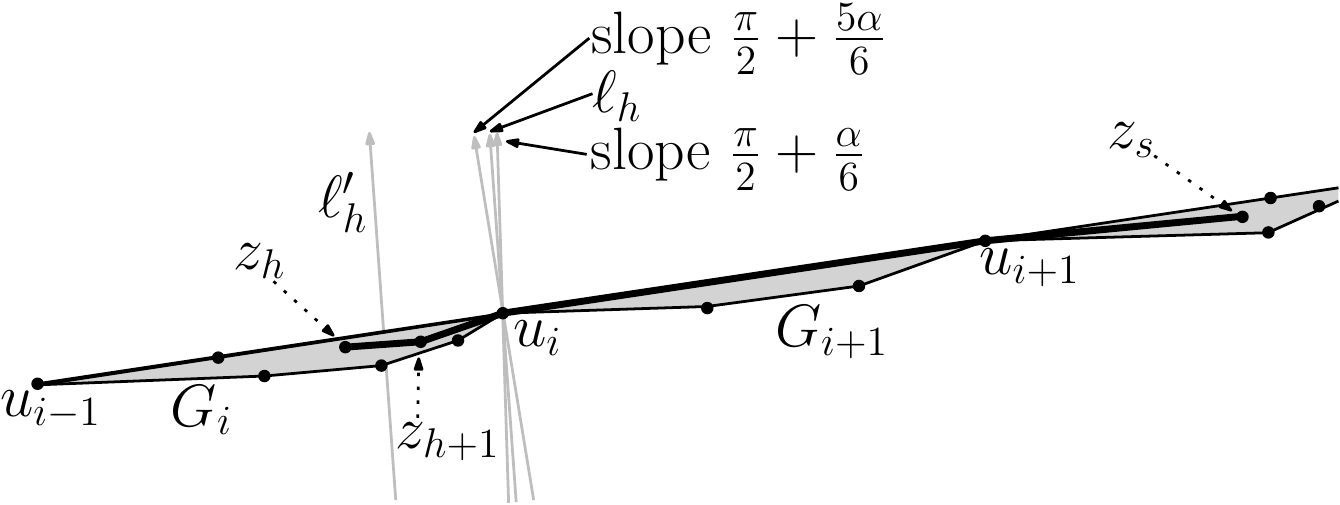}
			\caption{Illustration for the proof that $d(\Gamma_{\delta},z_hz_s)>d(\Gamma_{\delta},z_{h+1}z_s)$ if $z_hz_{h+1}$ is in $P^i_x$.}
			\label{fi:inductive1B-distancedecreasing2}
		\end{figure}
		
	\end{itemize}
	\item The case in which $2\leq j <i\leq k$ is symmetric to the case in which $2\leq i <j\leq k$. 
	\item Suppose next that $i=1$ and $j>1$. Then let $P_{xy}$ be the path composed of:
	\begin{itemize}
		\item a path $P^1_x$ in $G_1$ from $x$ to $u_1$ whose every edge has slope in $(-\alpha;\alpha)$ in $\Gamma_{\delta}$, where $P^1_x$ does not pass through $u$, unless $x=u$;
		\item the path $\bigcup_{l=2}^{j-1} \tau_{u_{l-1}u_l}(G_l)$; and
		\item a path $P_{u_{j-1}y}$ in $G_j$ that is distance-decreasing in $\Gamma_j$. 
	\end{itemize}
	The path $P_{u_{j-1}y}$ exists by induction since $\Gamma_j$ satisfies Property~6. 
	
	We prove that a path $P^1_x$ satisfying the above properties exists in $\Gamma_{\delta}$. Let $\tau_{u_0u_1}(G_1)=(u_0=c_1,c_2,\dots,c_r=u_1)$. If $x=u$, then let $P^1_x=\tau_{u_0u_1}(G_1)$. Every edge of $P^1_x$ other than $c_1c_2$ has slope $-\arctan\left(\frac{\varepsilon}{d(\Gamma_1,u_0u_1)}\right)$ in $\Gamma_{\delta}$, while $c_1c_2$ has slope $-\arctan\left(\frac{\varepsilon}{\delta+d(\Gamma_1,u_0u_1)}\right)$. These slopes are smaller than $0$, given that $\varepsilon,d(\Gamma_1,u_0u_1)>0$ and $\delta\geq 0$, and larger than $-\alpha$, given that $\delta\geq 0$ and $\varepsilon<\tan{\alpha}\cdot d(\Gamma_1,u_0u_1)$. Thus, every edge of $P^1_x$ has slope in $(-\alpha; \alpha)$ in $\Gamma_{\delta}$. If $x\neq u$, then $P^1_x$ can be shown to exist as in the proof that $\Gamma_{\delta}$ satisfies Property~4, by considering a path $(x=v_1,v_2,\dots,v_p=u_1)$ in $G_1$ that does not pass through $u$ and whose every edge has slope in $(-\frac{\alpha}{2};\frac{\alpha}{2})$ in $\Gamma_1$ and by defining $P^1_x=(v_1,v_2,\dots,v_h=c_j,c_{j+1},\dots,c_r)$, where $h$ is the smallest index such that $v_h\in V(\tau_{u_0u_1}(G_1))$. This concludes the proof that a path $P^1_x$ satisfying the required properties exists in $\Gamma_{\delta}$.
	
	Note that $u\notin V(P_{xy})$, unless $x=u$, given that $u\notin V(P^1_{x})$, unless $x=u$. Let $P_{xy}=(z_1,z_2,\dots,z_s)$; we prove that, for any $h=1,2,\dots,s-2$, it holds true that $d(\Gamma_{\delta},z_hz_s)>d(\Gamma_{\delta},z_{h+1}z_s)$. This can be proved exactly as in the case $2\leq i <j\leq k$ if $z_hz_{h+1}$ is in $P_{u_{j-1}y}$ or if $z_hz_{h+1}$ is in $\bigcup_{l=2}^{j-1} \tau_{u_{l-1}u_l}(G_l)$. Assume hence that $z_hz_{h+1}$ is in $P^1_x$ and recall that the slope of every edge of $P^1_x$ in $\Gamma_{\delta}$ is in $(-\alpha; \alpha)$. Similarly to the case $2\leq i <j\leq k$, consider the line $\ell_h$ that passes through $u_1$, that is directed towards increasing $y$-coordinates and that is orthogonal to the line through $z_h$ and $z_{h+1}$. Denote by $s_h$ the slope of $\ell_h$. Then $s_h\in (\frac{\pi}{2}-\alpha;\frac{\pi}{2}+\alpha)$. By Lemma~\ref{le:same-side}, the line $\ell_h$ has the drawings of $G_2,\dots,G_k$ to its right; this is because by Property~3 of $\Gamma_{\delta}$ every edge in $\beta_{u_1 v}(G)$ has slope in $(0;\alpha)$ with $s_h-\pi<-\frac{\pi}{2}+\alpha<0<\alpha<\frac{\pi}{2}-\alpha<s_h$ and because the path $\bigcup_{m=2}^{k} \tau_{u_{m-1}u_m}(G_m)$ has slope $s=\frac{\alpha}{2}$, where $s_h-\pi<-\frac{\pi}{2}+\alpha<\frac{\alpha}{2}<\frac{\pi}{2}-\alpha<s_h$. Further, by Lemma~\ref{le:same-side}, the line  $\ell_h$ has the drawing of $G_1$ to its left; this is because by Property~3 of $\Gamma_1$ every edge in $\tau_{u_1 u_0}(G)$ has slope in $(\pi-\frac{\alpha}{2};\pi+\frac{\alpha}{2})$, where $s_h<\frac{\pi}{2}+\alpha<\pi-\frac{\alpha}{2}<\pi+\frac{\alpha}{2}<\frac{3\pi}{2}-\alpha<\pi+s_h$ and because every edge of the path $\beta_{u_1u_0}(G_1)$ has slope either $\pi-\arctan\left(\frac{\varepsilon}{d(\Gamma_1,u_0u_1)}\right)$ or $\pi-\arctan\left(\frac{\varepsilon}{\delta+d(\Gamma_1,u_0u_1)}\right)$, where $s_h<\frac{\pi}{2}+\alpha<\pi-\alpha<\pi-\arctan\left(\frac{\varepsilon}{d(\Gamma_1,u_0u_1)}\right)\leq \pi-\arctan\left(\frac{\varepsilon}{\delta+d(\Gamma_1,u_0u_1)}\right)<\pi<\frac{3\pi}{2}-\alpha<\pi+s_h$ -- these inequalities exploit $\varepsilon,d(\Gamma_1,u_0u_1)>0$, $\delta\geq 0$, and $\varepsilon<\tan{\alpha}\cdot d(\Gamma_1,u_0u_1)$. Now consider the line $\ell'_{h}$ parallel to $\ell_{h}$, passing through the midpoint of the edge $z_hz_{h+1}$, and oriented towards increasing $y$-coordinates. This line has $\ell_{h}$ to its right, given that the drawing of $G_1$ (and in particular the midpoint of $z_hz_{h+1}$) is to the left of $\ell_{h}$ in $\Gamma_{\delta}$. Thus, $\ell'_{h}$ has the drawings of $G_{2},G_{3},\dots,G_k$ (and in particular the vertex $z_s$) to its right. Since the half-plane to the right of $\ell'_{h}$ represents the locus of the points of the plane that are closer to $z_{h+1}$ than to $z_h$, it follows that $d(\Gamma_{\delta},z_hz_s)>d(\Gamma_{\delta},z_{h+1}z_s)$.
			
	\item Suppose finally that $i>1$ and $j=1$. Then $P_{xy}$ consists of three paths, one contained in $G_i$, one coinciding with $\bigcup_{l=2}^{i-1} \beta_{u_lu_{l-1}}(G_l)$, and one contained in $G_1$.
	
	\begin{itemize}
		\item The first path in $P_{xy}$ is a path $Q^i_x$ in $G_i$ from $x$ to $u_{i-1}$ whose every slope in $\Gamma_i$ is in $(\pi-\frac{\alpha}{3}; \pi+\frac{\alpha}{3})$; this path exists since $\Gamma_i$ satisfies Property~5. Since $\Gamma_i$ is counter-clockwise rotated by $s$ radians in $\Gamma_{\delta}$, it follows that every edge of $Q^i_x$ has slope in $(\pi+s-\frac{\alpha}{3}; \pi+s+\frac{\alpha}{3})=(\pi+\frac{\alpha}{6};\pi+\frac{5\alpha}{6})$ in $\Gamma_{\delta}$. We prove that, for every edge $z_hz_{h+1}$ of $Q^i_x$, it holds true that $d(\Gamma_{\delta},z_hy)>d(\Gamma_{\delta},z_{h+1}y)$. Consider the line $\ell_h$ that passes through $u_{i-1}$, that is directed towards increasing $y$-coordinates and that is orthogonal to the line through $z_h$ and $z_{h+1}$. Denote by $s_h$ the slope of $\ell_h$. Then $s_h\in (\frac{\pi}{2}+\frac{\alpha}{6};\frac{\pi}{2}+\frac{5\alpha}{6})$. By Lemma~\ref{le:same-side}, the line $\ell_h$ has the drawing of $G_i$ to its right; this is because by Property~3 of $\Gamma_{\delta}$ every edge in $\beta_{u_{i-1} u_i}(G)$ has slope in $(-\alpha;\alpha)$ with $s_h-\pi<-\frac{\pi}{2}+\frac{5\alpha}{6}<-\alpha<\alpha<\frac{\pi}{2}+\frac{\alpha}{6}<s_h$ and because the path $\tau_{u_{i-1}u_i}(G_i)$ has slope $s=\frac{\alpha}{2}$, where $s_h-\pi<-\frac{\pi}{2}+\frac{5\alpha}{6}<\frac{\alpha}{2}<\frac{\pi}{2}+\frac{\alpha}{6}<s_h$. Further, by Lemma~\ref{le:same-side}, the line $\ell_h$ has the drawings of $G_1,\dots,G_{i-1}$ (and in particular $y$) to its left; this is because by Property~3 of $\Gamma_{\delta}$ every edge in $\tau_{u_{i-1} u_0}(G)$ has slope in $(\pi-\alpha;\pi+\alpha)$ where $s_h<\frac{\pi}{2}+\frac{5\alpha}{6}<\pi-\alpha<\pi+\alpha<\frac{3\pi}{2}+\frac{\alpha}{6}<\pi+s_h$ and because every edge of the path $\bigcup_{m=1}^{i-1} \beta_{u_{m}u_{m-1}}(G_m)$ has slope either $s=\pi+\frac{\alpha}{2}$, or $\pi-\arctan\left(\frac{\varepsilon}{d(\Gamma_1,u_0u_1)}\right)$, or $\pi-\arctan\left(\frac{\varepsilon}{\delta+d(\Gamma_1,u_0u_1)}\right)$, where $s_h<\frac{\pi}{2}+\frac{5\alpha}{6}<\pi-\alpha<\pi-\arctan\left(\frac{\varepsilon}{d(\Gamma_1,u_0u_1)}\right)\leq \pi-\arctan\left(\frac{\varepsilon}{\delta+d(\Gamma_1,u_0u_1)}\right)<\pi<\pi+\frac{\alpha}{2}< \frac{3\pi}{2}+\frac{\alpha}{6}<\pi+s_h$. Now consider the line $\ell'_{h}$ parallel to $\ell_{h}$, passing through the midpoint of the edge $z_hz_{h+1}$, and oriented towards increasing $y$-coordinates. This line has $\ell_{h}$ to its left, given that the drawing of $G_i$ (and in particular the midpoint of $z_hz_{h+1}$) is to the right of $\ell_{h}$ in $\Gamma_{\delta}$. Thus, $\ell'_{h}$ has the drawings of $G_{i-1},G_{i-2},\dots,G_1$ (and in particular the vertex $y$) to its left. Since the half-plane to the left of $\ell'_{h}$ represents the locus of the points of the plane that are closer to $z_{h+1}$ than to $z_h$, it follows that $d(\Gamma_{\delta},z_hy)>d(\Gamma_{\delta},z_{h+1}y)$.
		
		\item The second path in $P_{xy}$ is $\bigcup_{l=2}^{i-1} \beta_{u_lu_{l-1}}(G_l)$. Consider an edge $z_hz_{h+1}$ of this path in a graph $G_l$, for some $l\in\{2,\dots,i-1\}$. Then $z_hz_{h+1}$ has slope $\pi+s=\pi+\frac{\alpha}{2}$. Consider the line $\ell_{h}$ that has slope $s_h=\frac{\pi+\alpha}{2}$, that passes through $u_{l-1}$, and that is oriented towards increasing $y$-coordinates. By Lemma~\ref{le:same-side}, the line $\ell_h$ has the drawing of $G_l$ to its right; this is because by Property~3 of $\Gamma_{\delta}$ every edge in $\beta_{u_{l-1} u_l}(G)$ has slope in $(0;\alpha)$ with $s_h-\pi=\frac{-\pi+\alpha}{2}<0<\alpha<\frac{\pi+\alpha}{2}=s_h$ and because the path $\tau_{u_{l-1}u_l}(G_l)$ has slope $s=\frac{\alpha}{2}$, where $s_h-\pi=\frac{-\pi+\alpha}{2}<\frac{\alpha}{2}<\frac{\pi+\alpha}{2}=s_h$. Further, by Lemma~\ref{le:same-side}, the line $\ell_h$ has the drawings of $G_{l-1},G_{l-2},\dots,G_1$ to its left; this is because by Property~3 of $\Gamma_{\delta}$ every edge in $\tau_{u_{l-1} u_0}(G)$ has slope in $(\pi-\alpha;\pi+\alpha)$, where $s_h=\frac{\pi+\alpha}{2}<\pi-\alpha<\pi+\alpha<\frac{3\pi+\alpha}{2}=s_h+\pi$ and because every edge of the path $\bigcup_{m=1}^{i-1} \beta_{u_m u_{m-1}}(G_m)$ has slope either $\pi+\frac{\alpha}{2}$, or $\pi-\arctan\left(\frac{\varepsilon}{d(\Gamma_1,u_0u_1)}\right)$, or $\pi-\arctan\left(\frac{\varepsilon}{\delta+d(\Gamma_1,u_0u_1)}\right)$, where $s_h=\frac{\pi+\alpha}{2}<\pi-\alpha<\pi-\arctan\left(\frac{\varepsilon}{d(\Gamma_1,u_0u_1)}\right)\leq \pi-\arctan\left(\frac{\varepsilon}{\delta+d(\Gamma_1,u_0u_1)}\right)<\pi<\pi+\frac{\alpha}{2}<\frac{3\pi+\alpha}{2}=s_h+\pi$. Now consider the line $\ell'_{h}$ parallel to $\ell_{h}$, passing through the midpoint of the edge $z_hz_{h+1}$, and oriented towards increasing $y$-coordinates. This line has $\ell_{h}$ to its left, given that the drawing of $G_l$ (and in particular the midpoint of $z_hz_{h+1}$) is to the right of $\ell_{h}$ in $\Gamma_{\delta}$. Thus, $\ell'_{h}$ has the drawings of $G_{l-1},G_{l-2},\dots,G_1$ (and in particular vertex $y$) to its left. Since the half-plane to the left of $\ell'_{h}$ represents the locus of the points of the plane that are closer to $z_{h+1}$ than to $z_h$, it follows that $d(\Gamma_{\delta},z_hy)>d(\Gamma_{\delta},z_{h+1}y)$.  
		
		\item The third path $P_{u_1y}$ in $P_{xy}$ is defined as follows. If $y=u$, let $P_{u_1y}=\beta_{u_1u_0}(G_1)$. Then every edge of $P_{u_1y}$ has slope either $\pi-\arctan\left(\frac{\varepsilon}{d(\Gamma_1,u_0u_1)}\right)$ or $\pi-\arctan\left(\frac{\varepsilon}{\delta+d(\Gamma_1,u_0u_1)}\right)$. Both these slopes are smaller than $\pi$, given that $\varepsilon,d(\Gamma_1,u_0u_1)>0$ and $\delta\geq 0$, and larger than $\pi-\alpha$, given that $\delta\geq 0$ and $\varepsilon<\tan(\alpha) \cdot d(\Gamma_1,u_0u_1)$. Thus, $P_{u_1y}$ is a $\pi$-path, and hence it is distance-decreasing	(see~\cite{dfg-icgps-15} and the proof of Property~6 in Lemma~\ref{le:case1A}). If $y\neq u$, then let $P_{u_1y}$ be a distance-decreasing path in $\Gamma_1$ not passing through $u$. This path exists by induction, given that $\Gamma_1$ satisfies Property~6. Since $P_{u_1y}$ does not pass through $u$, it has the same representation in $\Gamma_\delta$ and $\Gamma$. Since the Euclidean distance between the positions of any vertex of $G_1$ in $\Gamma_1$ and $\Gamma$ is at most $\varepsilon<\varepsilon^*_{\Gamma_1}$, by Lemma~\ref{le:perturbation-preserves-greedy} we have that $P_{u_1y}$ is distance-decreasing in $\Gamma$ and hence in $\Gamma_\delta$.
\end{itemize}
\end{itemize}
Hence $\Gamma_{\delta}$ satisfies Property~6. This concludes the proof of the lemma.
\end{proof}


We now discuss {\bf Case~B}, in which $(G,u,v)$ is decomposed according to Lemma~\ref{le:decomposition-B}. Refer to Figs.~\ref{fi:inductiveB-statement} and~\ref{fi:inductiveB-statement2}. First, the triple $(H,u,y_1)$ is a strong circuit graph; further, $|V(H)|\geq 3$, hence $H$ is not a single edge. Apply induction in order to construct a straight-line drawing $\Gamma_H$ of $H$ with $\frac{\alpha}{2}$ as a parameter. 

	\begin{figure}[htb]
		\centering
		\includegraphics[scale=0.66]{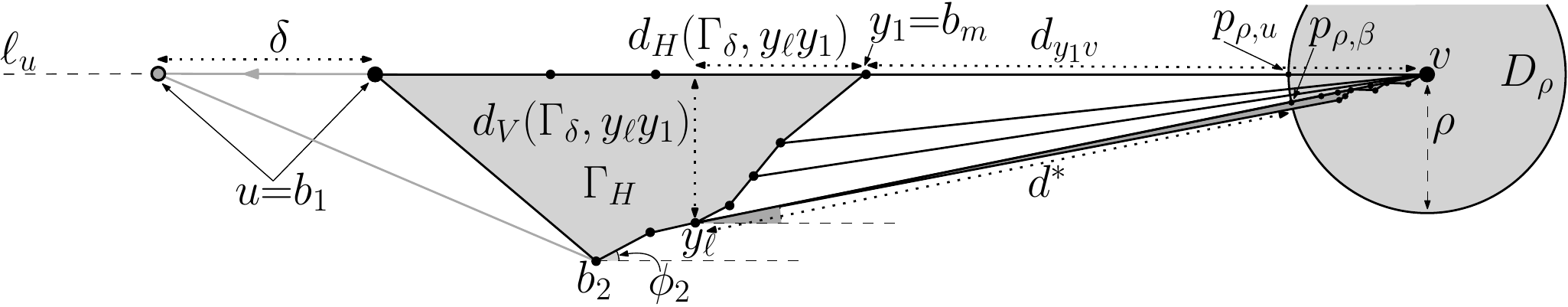}
		\caption{The straight-line drawing $\Gamma$ of $G$ in Case~B. For the sake of readability,  $\phi$ and $\rho$ are larger than they should be. The dark gray angle is equal to $\beta$.}
		\label{fi:inductiveB-statement}
	\end{figure}

Let $\beta_{uy_1}(H)=(u=b_1,b_2,\dots,b_m=y_1)$. Let $\phi_i$ be the slope of the edge $b_ib_{i+1}$ in $\Gamma_H$ and let $\phi=\min_{i=2,\dots,m-1} \{\phi_i\}$. By Property~(c) of $(H,u,y_1)$ if edge $uy_1$ belongs to $H$ then it coincides with the path $\tau_{uy_1}(H)$. Hence, $m\geq 3$ and $\phi$ is well-defined. Further, $\phi$ is in the interval $(0;\frac{\alpha}{2})$ by Property~3 of $\Gamma_H$. 

Let $\beta = \frac{1}{2}\min\left\{\phi, \arctan\left(\frac{d_V(\Gamma_H,y_{\ell}y_1)}{3d_V(\Gamma_H,y_{\ell}y_1)+3d_H(\Gamma_H,y_{\ell}y_1)}\right)\right\}$. Note that $\beta>0$, given that $\phi,d_V(\Gamma_H,y_{\ell}y_1) >0$ and $d_H(\Gamma_H,y_{\ell}y_1)\geq 0$. In particular, $d_V(\Gamma_H,y_{\ell}y_1) >0$ because $y_1$ is an internal vertex of $\tau_{uy_1}(H)$ and $y_{\ell}$ is an internal vertex of $\beta_{uy_1}(H)$ by Lemma~\ref{le:decomposition-B}, and because of Properties~1--3 of $\Gamma_H$. Also note that $\beta<\frac{\alpha}{4}$, given that $\phi<\frac{\alpha}{2}$.

Consider a half-line $h_\beta$ with slope $\beta$ starting at $y_{\ell}$. Place the vertex $v$ at the intersection point between $h_\beta$ and the horizontal line $\ell_u$ through $u$. Draw all the trivial $(H\cup \{v\})$-bridges of $G$ as straight-line segments. This concludes the construction if every $(H\cup \{v\})$-bridge of $G$ is trivial. Otherwise, $B_\ell$ is the only non-trivial $(H\cup \{v\})$-bridge of $G$. Then $B_\ell$ consists of $k$ strong circuit graphs $(G_i,u_{i-1},u_i)$, where $u_0=y_\ell$ and $u_k=v$. With a slight change of notation, in the remainder of the section we assume that, if the edge $y_\ell v$ exists, then it is an edge of $B_\ell$ (rather than an individual trivial $(H\cup \{v\})$-bridge $B_{\ell-1}$ of $G$); in this case $(B_\ell,u_0,u_k)$ is a strong circuit graph (this comes from the proof of Lemma~\ref{le:decomposition-B}, where the graph $B_\ell$ together with the edge $y_\ell v$ was denoted by $B'_\ell$). 

We claim that $v$ lies to the right of $y_1$. The polygonal line representing $\beta_{y_{\ell}y_1}(H)$ in $\Gamma_H$ and the straight-line segment $\overline{y_{\ell}v}$ are both incident to $y_{\ell}$. By definition of $\phi$ and since $\Gamma_H$ satisfies Property~3, $\beta_{y_{\ell}y_1}(H)$ is composed of straight-line segments with slopes in the range $[\phi;\frac{\alpha}{2})$, while $\overline{y_{\ell}v}$ has slope $\beta$. The claim then follows from $0<\beta<\phi<\frac{\pi}{2}$.

Denote by $d_{y_1v}$ the distance between $y_1$ and $v$. Let $Y>0$ be the minimum distance in $\Gamma_H$ of any vertex strictly below $\ell_u$ from $\ell_u$.

Let $\rho=\min\{\frac{d_{y_1v}}{3},\frac{Y}{2}\}$. Let $D_{\rho}$ be the disk with radius $\rho$ centered at $v$. Let $p_{\rho,\beta}$ ($p_{\rho,u}$) be the intersection point closer to $y_{\ell}$ (resp.\ to $y_1$) of the boundary of $D_{\rho}$ with $h_{\beta}$ (resp.\ with $\ell_u$). Let $d^*$ be the Euclidean distance between $y_{\ell}$ and $p_{\rho,\beta}$. 

\begin{figure}[htb]
	\centering
	\subfloat[]{
		\includegraphics[scale=0.6]{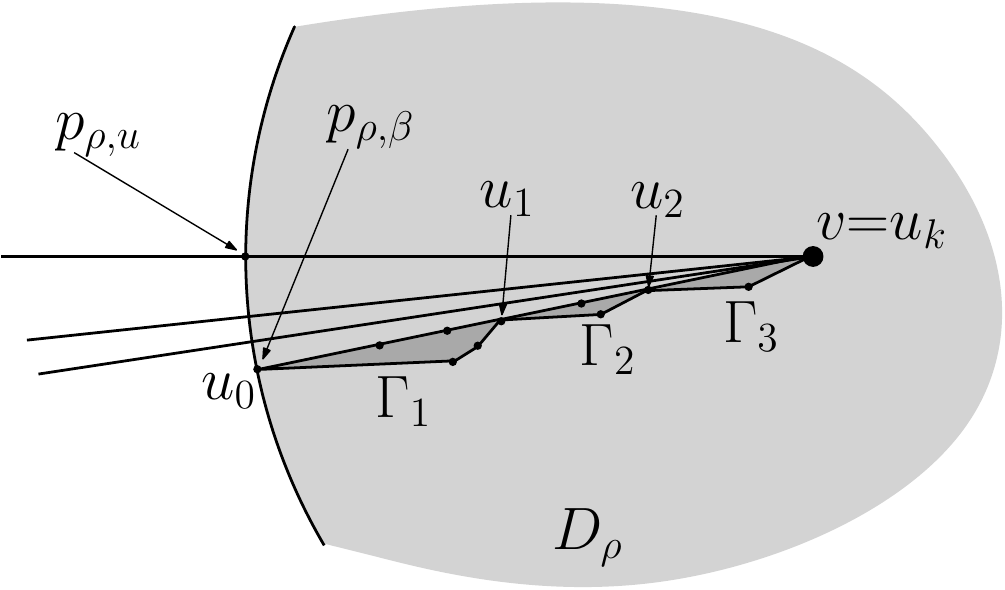}
		\label{fi:closer1}
	}\hfil
	\subfloat[]{
		\includegraphics[scale=0.6]{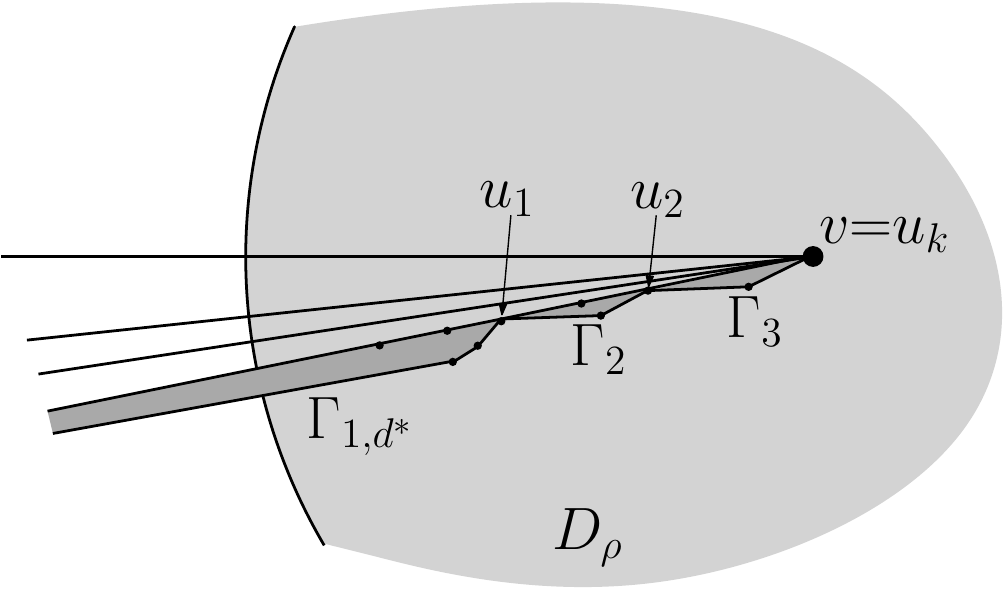}
		\label{fi:closer2}
	}
	\caption{A closer look at $D_{\rho}$. Figure (a) represents the drawings $\Gamma_1,\dots,\Gamma_k$, once they have been uniformly scaled, rotated, and translated, while (b) also has the vertex $u_0$ moved by $d^*$ units (this movement actually happens before the rotation and translation of $\Gamma_1$).}
	\label{fi:inductiveB-statement2}
\end{figure}

Let $\alpha'=\frac{\beta}{2}$. Since $\beta>0$, we have $\alpha'>0$; further, $\alpha'<\frac{\alpha}{8}$, given that $\beta<\frac{\alpha}{4}$. For $i=1,\dots,k$, apply induction in order to construct a straight-line drawing $\Gamma_i$ of $G_i$ with $\alpha'$ as a parameter (if $G_i$ is a single edge, then the parameter does not matter). Uniformly scale the drawings $\Gamma_1,\dots,\Gamma_k$ so that the Euclidean distance between $u_{i-1}$ and $u_i$ is equal to $\frac{\rho}{k}$. Move the vertex $u_0$ in $\Gamma_1$ by $d^*$ units to the left, obtaining a drawing $\Gamma_{1,d^*}$. Rotate the drawings $\Gamma_{1,d^*},\Gamma_2,\dots,\Gamma_k$ counter-clockwise by $\beta$ radians. Translate $\Gamma_{1,d^*},\Gamma_2,\dots,\Gamma_k$ so that, for $i=1,\dots,k-1$, the representations of $u_i$ in $\Gamma_i$ and $\Gamma_{i+1}$ (in $\Gamma_{1,d^*}$ and $\Gamma_2$ if $i=1$) coincide and so that the representation of $u_0$ in the scaled and rotated drawing $\Gamma_{1,d^*}$ coincides with the one of $y_{\ell}$ in $\Gamma_H$. This completes the construction of a straight-line drawing $\Gamma$ of $G$. We have the following.

\begin{lemma} \label{le:caseB}
For any $\delta\geq 0$, the drawing $\Gamma_{\delta}$ constructed in Case~B satisfies Properties~1--6 of Theorem~\ref{th:main-aux}.	
\end{lemma}

\begin{proof}
Let $\Gamma_{H,\delta}$ be the drawing obtained from $\Gamma_H$ by moving $u$ by $\delta$ units to the left. 

We prove Property~2. By Lemma~\ref{le:decomposition-B}, we have that $\tau_{uv}(G)=\tau_{uy_1}(H)\cup y_1v$. By Property~2 of $\Gamma_{H,\delta}$, we have that $\tau_{uy_1}(H)$ lies entirely on $\ell_u$ with $y_1$ to the right of $u$. By construction $v$ also lies on $\ell_u$. As proved before the lemma's statement, $v$ lies to the right of $y_1$. This implies Property~2 for $\Gamma_\delta$. 

We next prove that $\Gamma_\delta$ satisfies Property~3. By Lemma~\ref{le:decomposition-B}, we have that $\beta_{uv}(G)=\beta_{uy_\ell}(H)\cup \beta_{u_0u_1}(G_1)\cup \beta_{u_1u_2}(G_2)\cup \dots \cup \beta_{u_{k-1}u_k}(G_k)$. Denote $\beta_{uv}(G)=(u=b'_1,b'_2,\dots,b'_m=v)$. The slope of the edge $b'_1b'_2$ in $\Gamma_{\delta}$ is equal to its slope in $\Gamma_{H,\delta}$; this is because the drawing of $H$ in $\Gamma_{\delta}$ coincides with $\Gamma_{H,\delta}$ and because $b'_2$ is a vertex of $H$, given that $y_{\ell}\neq u$ since $y_{\ell}$ is an internal vertex of $\beta_{uv}(G)$ . Hence the slope of $b'_1b'_2$ is in $(-\frac{\alpha}{2};0)\subset (-\alpha;0)$ in $\Gamma_{\delta}$ since $\Gamma_{H,\delta}$  satisfies Property~3. We now argue about the slope $s_j$ of the edge $b'_jb'_{j+1}$ in $\Gamma_{\delta}$, for any $j=2,\dots,m-1$.
	
\begin{itemize}
	\item If $b'_jb'_{j+1}$ is an edge of $\beta_{uy_\ell}(H)$, then its slope in $\Gamma_{\delta}$ is equal to its slope in $\Gamma_{H,\delta}$, since the drawing of $H$ in $\Gamma_{\delta}$ coincides with $\Gamma_{H,\delta}$. Thus, $s_j\in (0;\frac{\alpha}{2})\subset (0;\alpha)$, since $\Gamma_{H,\delta}$  satisfies Property~3.  
	
	\item If $b'_jb'_{j+1}$ coincides with a graph $G_i$, then $s_j=\beta$. Since $0<\beta \leq \frac{\alpha}{4}$, we have $s_j \in (0;\alpha)$.
	
	\item If $b'_jb'_{j+1}$ belongs to a graph $G_i$, for some $i\in\{1,\dots,k\}$, with $|V(G_i)|\geq 3$, and with $b'_j\neq u_{i-1}$, then $s_j$ is given by the slope $b'_jb'_{j+1}$ has in $\Gamma_i$, which is in $(0;\alpha')$ by Property~3 of $\Gamma_i$, plus $\beta$, which results from the rotation of $\Gamma_i$. Hence $s_j\in (\beta;\beta+\alpha')$; since $\beta>0$, $\beta<\frac{\alpha}{4}$, and $\alpha'<\frac{\alpha}{8}$, we have that $s_j\in (0;\alpha)$.   
	
	\item If $b'_jb'_{j+1}$ belongs to a graph $G_i$, for some $i\in\{2,\dots,k\}$, with $|V(G_i)|\geq 3$, and with $b'_j= u_{i-1}$, then $s_j$ is given by the slope $b'_jb'_{j+1}$ has in $\Gamma_i$, which is in $(-\alpha';0)$ by Property~3 of $\Gamma_i$, plus $\beta$, which results from the rotation of $\Gamma_i$. Hence $s_j\in (\beta-\alpha';\beta)$. Since $\alpha'\leq \frac{\beta}{2}<\beta<\frac{\alpha}{4}<\alpha$, we have that $s_j\in (0;\alpha)$.

	\item Finally, assume that $b'_jb'_{j+1}$ belongs to $G_1$, that $|V(G_1)|\geq 3$, and that $b'_j= u_0$. Then $s_j$ is given by the slope $b'_jb'_{j+1}$ has in $\Gamma_{1,d^*}$, which is in $(-\alpha';0)$ by Property~3 of $\Gamma_{1,d^*}$, plus $\beta$, which results from the rotation of $\Gamma_{1,d^*}$. Hence $s_j\in (\beta-\alpha';\beta)\subset (0;\alpha)$.
\end{itemize}

We now prove Property~1. Before doing so, we prove the following useful statement: Every vertex $z\neq u_0$ that belongs to a graph $G_i$, for any $i\in\{1,\dots,k\}$, lies inside the disk $D_{\rho}$ in $\Gamma_{\delta}$. Note that this statement shows a sharp geometric separation between the vertices that are in $H$ and those that are not. 
\begin{figure}[htb]
	\centering
	\includegraphics[scale=0.6]{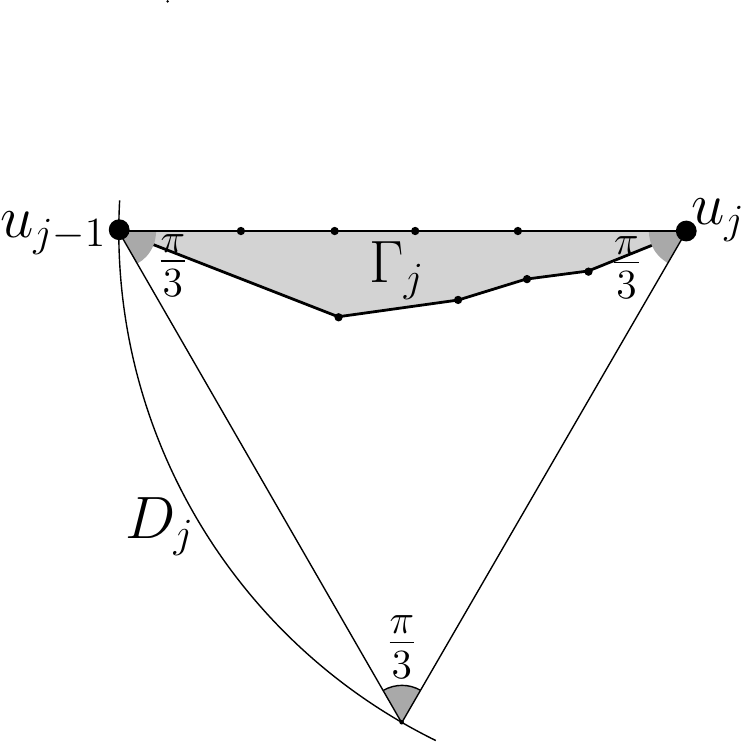}
	\caption{The drawing $\Gamma_j$ and the disk $D_j$ centered at $u_j$ with radius $d(\Gamma_j,u_{j-1}u_j)$.}
	\label{fi:disk}
\end{figure}
Refer to Fig.~\ref{fi:disk}. Consider the drawing $\Gamma_j$, for any $j\in\{1,\dots,k\}$ (note that $\Gamma_1$ is considered before moving $u_0$ by $d^*$ units to the left) and consider the disk $D_j$ centered at $u_j$ with radius $d(\Gamma_j,u_{j-1}u_j)$. By Properties~1 and~2 of $\Gamma_j$, the path $\tau_{u_{j-1}u_j}(G_j)$ lies on the straight-line segment $\overline{u_{j-1}u_j}$ in $\Gamma_j$, hence it lies inside $D_j$. Further, all the edges of $\beta_{u_{j-1}u_j}(G_j)$ have slope in $(-\alpha';\alpha')\subset (-\frac{\alpha}{8};\frac{\alpha}{8})\subset (-\frac{\pi}{32};\frac{\pi}{32})\subset (-\frac{\pi}{3};\frac{\pi}{3})$; hence $\beta_{u_{j-1}u_j}(G_j)$ also lies inside $D_j$. By Property~1 of $\Gamma_j$, the entire drawing $\Gamma_j$ lies inside $D_j$. Hence, $u_{j-1}$ is the farthest vertex of $G_j$ from $u_j$ in $\Gamma_j$. This property holds true also after the drawings $\Gamma_1, \dots, \Gamma_k$ are uniformly scaled; further, after the scaling, the distance between $u_{j-1}$ and $u_j$ is $\frac{\rho}{k}$, by construction. By the triangular inequality, we have that $d(\Gamma_{\delta},vz)\leq \sum_{j=i+1}^k d(\Gamma_{\delta},u_{j-1} u_j)+d(\Gamma_{\delta},u_iz)$.  Since $d(\Gamma_{\delta},u_{j-1} u_j)=\frac{\rho}{k}$ for any $j\in\{2,\dots,k\}$, and since $d(\Gamma_{\delta},u_iz)\leq \frac{\rho}{k}$ (this exploits $z\neq u_0$ and hence $d(\Gamma_{\delta},u_iz)=d(\Gamma_i,u_iz)$, where $\Gamma_i$ is understood as already scaled), we have that $d(\Gamma_{\delta},vz)\leq \frac{(k-i+1)\rho}{k}\leq \rho$. Thus $z$ lies inside $D_{\rho}$.

We now discuss the possible crossings that might occur in $\Gamma_{\delta}$. 

\begin{itemize}
	\item The drawing of $H$ in $\Gamma_{\delta}$ coincides with $\Gamma_{H,\delta}$, hence it is planar since $\Gamma_{H,\delta}$ satisfies Property~1 by induction. 
	\item Analogously, the drawings of $G_1,G_2,\dots,G_k$ in $\Gamma_{\delta}$ are planar since they coincide with $\Gamma_{1,d^*},\Gamma_2,\dots,\Gamma_k$, which satisfy Property~1 by induction. 
	\item Since $\Gamma_{\delta}$ satisfies Property~3, the path $\beta_{y_{\ell}v}(G)$ is represented in $\Gamma_{\delta}$ by a curve monotonically increasing in the $x$-direction from $y_{\ell}$ to $v$. Further, the path $\tau=\bigcup_{i=1}^{k} \tau_{u_{i-1}u_i}(G_i)$ is represented in $\Gamma_{\delta}$ by a straight-line segment with slope $\beta\in (0;\frac{\alpha}{4})\subset(0;\frac{\pi}{16})$. Hence, for $i=1,\dots,k-1$, the vertical line through $u_i$ has the drawings of $G_1,\dots,G_i$ to its left and those of $G_{i+1},\dots,G_k$ to its right in $\Gamma_{\delta}$. It follows that no two edges in distinct graphs $G_i$ and $G_j$ cross in $\Gamma_{\delta}$. 
	\item Recall that $\beta_{uy_1}(H)=(u=b_1,b_2,\dots,b_m=y_1)$. Let $y_{\ell}=b_j$, for some $j\in \{2,3,\dots,m-1\}$. We prove that the straight-line segments $\overline{b_jv},\overline{b_{j+1}v},\dots,\overline{b_m v}$ appear in this clockwise order around $v$ and have slopes in $[0;\beta]$ in $\Gamma_{\delta}$ (note that these straight-line segments do not necessarily correspond to edges of $G$). Refer to Fig.~\ref{fi:slopes}.
	\begin{figure}[htb]
		\centering
		\includegraphics[scale=0.6]{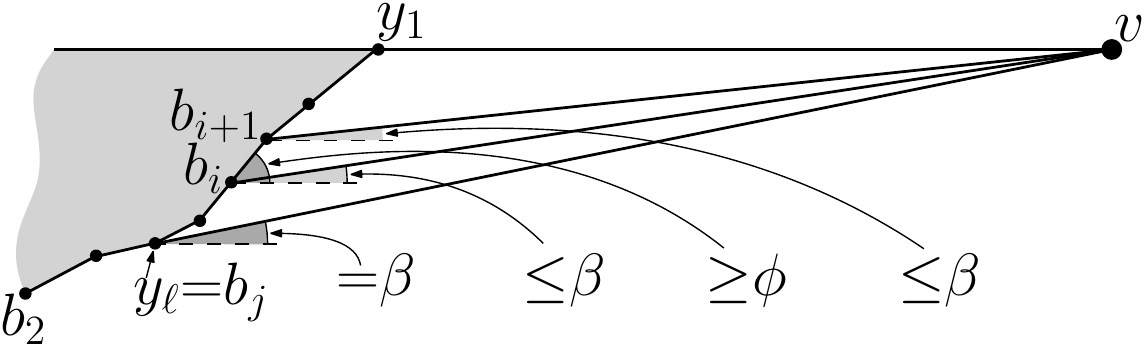}
		\caption{illustration for the proof that the straight-line segments $\overline{b_jv},\overline{b_{j+1}v},\dots,\overline{b_m v}$ appear in this clockwise order around $v$ and have slopes in $[0;\beta]$ in $\Gamma_{\delta}$.}
		\label{fi:slopes}
	\end{figure}
	Note that the slope of $\overline{b_jv}$ is $\beta$, by construction; now assume that $\overline{b_jv},\overline{b_{j+1}v},\dots,\overline{b_i v}$ appear in this clockwise order around $v$ and have slopes in $[0;\beta]$ in $\Gamma_{\delta}$, for some $i\in \{j,j+1,\dots,m-1\}$. The edge $b_i b_{i+1}$ has slope in $[\phi;\frac{\alpha}{2})$, by definition of $\phi$ and since $\Gamma_{H,\delta}$ satisfies Property~3. Since $\beta<\phi$, the edge $b_i b_{i+1}$ lies above the line through $b_i$ and $v$. Hence, $b_{i+1}v$ immediately follows $b_i v$ in the clockwise order of the edges incident to $v$ and it has slope smaller than the one of $b_{i}v$, hence smaller than $\beta$. The repetition of this argument concludes the proof that $\overline{b_jv},\overline{b_{j+1}v},\dots,\overline{b_m v}$ appear in this clockwise order around $v$ and have slopes in $[0;\beta]$ in $\Gamma_{\delta}$.

	Since the straight-line segments $\overline{b_jv},\overline{b_{j+1}v},\dots,\overline{b_m v}$ appear in this clockwise order around $v$, then no two $(H\cup\{v\})$-bridges of $G$ cross one another. Further, since the straight-line segments $\overline{b_jv},\overline{b_{j+1}v},\dots,\overline{b_m v}$ have slopes in $[0;\beta]$ and since $\beta<\phi$, they all lie to the right of the path $\beta_{b_2b_m}(H)$, whose edges have slopes in $[\phi;\frac{\alpha}{2})$. Then no trivial $(H\cup\{v\})$-bridge of $G$ crosses $H$ in $\Gamma_{\delta}$. 
	
	\item Consider the vertical line $\ell_1$ through $y_1$. By Properties~1--3 of $\Gamma_{H,\delta}$, the line $\ell_1$ has $\Gamma_{H,\delta}$ to its left. Further, since $\rho<d_{y_1v}$, the disk $D_{\rho}$ lies to the right of $\ell_1$. Since all the vertices different from $u_0$ of the graphs $G_1,\dots,G_k$ lie inside $D_{\rho}$, it follows that no edge in a graph $G_1,\dots,G_k$ crosses an edge of $H$, unless the former is incident to $u_0$. However, all the edges in $G_1,\dots,G_k$ that are incident to $u_0$ (in fact only $G_1$ contains such edges) have slope at most $\beta$, as they lie on or below $h_{\beta}$. Hence they all lie to the right of the path $\beta_{b_2b_m}(H)$ and do not cross edges of $H$ in $\Gamma_{\delta}$.
\end{itemize}

We now discuss Property~4. Let $x\in V(G)$. Assume first that $x\in V(H)$. Since the drawing of $H$ in $\Gamma_{\delta}$ coincides with $\Gamma_{H,\delta}$, there exists a path $P'_x$ in $H$ from $x$ to $y_1$, not passing through $u$ unless $x=u$, and whose every edge has slope in $(-\frac{\alpha}{2};\frac{\alpha}{2})\subset (-\alpha;\alpha)$ in $\Gamma_{\delta}$. Further, the edge $y_1v$ has slope $0$. Hence, the path $P_x=P'_x \cup y_1v$ satisfies the required properties. 
If $x\notin V(H)$, then $x\in V(G_i)$, for some $i\in \{1,\dots,k\}$. Assume that $i\geq 2$ (that $i=1$, resp.). Then the path $P_x$ consists of a path $P'_x$ from $x$ to $u_i$ in $G_i$ whose every edge has slope in $(-\alpha';\alpha')$ in $\Gamma_i$ (in $\Gamma_{1,d^*}$, resp.) -- this path exists since $\Gamma_i$ ($\Gamma_{1,d^*}$, resp.) satisfies Property~4 -- and of the path $\bigcup_{j=i+1}^k \tau_{u_{j-1}u_j}(G_j)$. Since $\Gamma_i$ ($\Gamma_{1,d^*}$, resp.) is rotated by $\beta$ radians in $\Gamma_{\delta}$, its edges have slope in the range  $(\beta-\alpha';\beta+\alpha')$. Since $\alpha'=\frac{\beta}{2}$ and $0<\beta<\frac{\alpha}{4}$, we have that $(\beta-\alpha';\beta+\alpha')\subset (0;\frac{3\alpha}{8})\subset (-\alpha;\alpha)$. Further, every edge in $\tau_{u_{j-1}u_j}(G_j)$ has slope $0$ in $\Gamma_j$ and hence $\beta$ in $\Gamma_{\delta}$. Since $0<\beta<\frac{\alpha}{4}$, we have that every edge in $\bigcup_{j=i+1}^k \tau_{u_{j-1}u_j}(G_j)$ has slope in $(-\alpha;\alpha)$. Note that $P_x$ does not pass through $u$, since $u$ does not belong to any graph among $G_1,\dots,G_k$. Thus, $P_x$ satisfies the required properties.

We now deal with Property~5. Let $x\in V(G)$. Assume first that $x\in V(H)$. Since the drawing of $H$ in $\Gamma_{\delta}$ coincides with $\Gamma_{H,\delta}$ and since $\Gamma_{H,\delta}$ satisfies Property~5, there exists a path $Q'_x$ from $x$ to $u$ whose every edge has slope in $(\pi-\frac{\alpha}{2};\pi+\frac{\alpha}{2})\subset (\pi-\alpha;\pi+\alpha)$. Thus, the path $Q_x=Q'_x$ satisfies the required properties. 
If $x\notin V(H)$, then $x\in V(G_i)$, for some $i\in \{1,\dots,k\}$. Then the path $Q_x$ consists of three paths. First, $Q_x$ contains a path $Q'_x$ from $x$ to $u_{i-1}$ in $G_i$ whose every edge has slope in $(\pi-\alpha';\pi+\alpha')$ in $\Gamma_i$ (in $\Gamma_{1,d^*}$, if $x\in V(G_1)$). This path exists since $\Gamma_i$ ($\Gamma_{1,d^*}$, resp.) satisfies Property~5. Since $\Gamma_i$ ($\Gamma_{1,d^*}$, resp.) is rotated by $\beta$ radians in $\Gamma_{\delta}$, its edges have slopes in the range  $(\pi+\beta-\alpha';\pi+\beta+\alpha')$. Since $\alpha'= \frac{\beta}{2}$ and $0<\beta<\frac{\alpha}{4}$, we have that $(\pi+\beta-\alpha';\pi+\beta+\alpha')\subset (\pi;\pi+\frac{3\alpha}{8})\subset (\pi-\alpha;\pi+\alpha)$. Second, $Q_x$ contains the path $\bigcup_{j=1}^{i-1} \beta_{u_ju_{j-1}}(G_j)$; by Properties~1 and~2, every edge in $\beta_{u_ju_{j-1}}(G_j)$ has slope $\pi$ in $\Gamma_i$ (in $\Gamma_{1,d^*}$ if $j=1$), hence it has slope $\pi+\beta$ in $\Gamma_{\delta}$. Since $0<\beta<\frac{\alpha}{4}$, we have that every edge in the path $\bigcup_{j=1}^{i-1} \beta_{u_ju_{j-1}}(G_j)$ has slope in $(\pi-\alpha;\pi+\alpha)$. Third, $Q_x$ contains a path $Q'_{y_\ell}$ from $y_\ell$ to $u$ in $H$ whose every edge has slope in $(\pi-\frac{\alpha}{2};\pi+\frac{\alpha}{2})\subset (\pi-\alpha;\pi+\alpha)$; this path exists since the drawing of $H$ in $\Gamma_{\delta}$ coincides with $\Gamma_{H,\delta}$ and since $\Gamma_{H,\delta}$ satisfies Property~5. Thus, the path $Q_x$ satisfies the required properties.

Finally, we deal with Property~6. Consider any two vertices $x,y\in V(G)$. We prove the existence of a path $P_{xy}$ from $x$ to $y$ in $G$ which does not pass through $u$, unless $x=u$ or $y=u$, and which is distance-decreasing in $\Gamma_{\delta}$. We distinguish several cases, based on which graphs among $H,G_1,\dots,G_k$ the vertices $x$ and $y$ belong to.

\begin{itemize}
\item Suppose first that $x$ and $y$ belong to $H$. Since $\Gamma_{H,\delta}$ satisfies Property~6, there exists a path $P_{xy}$ from $x$ to $y$ in $H$ which does not pass through $u$, unless $x=u$ or $y=u$, and which is distance-decreasing in $\Gamma_{H,\delta}$. Since the drawing of $H$ in $\Gamma_{\delta}$ coincides with $\Gamma_{H,\delta}$, it follows that $P_{xy}$ is distance-decreasing in $\Gamma_{\delta}$. 

\item Suppose next that $x$ and $y$ belong to the same graph $G_i$, for some $i\in\{1,2,\dots,k\}$. Since the drawing $\Gamma_i$ (or $\Gamma_{1,d^*}$ if $i=1$) satisfies Property~6, there exists a path $P_{xy}$ from $x$ to $y$ in $G_i$ that is distance-decreasing in $\Gamma_i$ (in $\Gamma_{1,d^*}$ if $i=1$). Since the drawing of $G_i$ in $\Gamma_{\delta}$ 
is congruent to $\Gamma_i$ (to $\Gamma_{1,d^*}$ if $i=1$) up to three affine transformations, namely a uniform scaling, a rotation, and a translation, that preserve the property of a path to be distance-decreasing, it follows that $P_{xy}$ is distance-decreasing in $\Gamma_{\delta}$. Note that $u\notin V(P_{xy})$.

\item Suppose now that $x$ belongs to a graph $G_i$ and $y$ belongs to a graph $G_j$ for some $1\leq i <j \leq k$. Then let $P_{xy}$ be the path composed of a path $P_x$ in $G_i$ from $x$ to $u_i$ whose every slope in $\Gamma_i$ is in $(-\alpha'; \alpha')$, of the path $\bigcup_{l=i+1}^{j-1} \tau_{u_{l-1}u_l}(G_l)$, and of a path $P_{u_{j-1}y}$ in $G_j$ that is distance-decreasing in $\Gamma_j$. The path $P_x$ exists since $\Gamma_i$ ($\Gamma_{1,d^*}$ if $i=1$) satisfies Property~4; the path $P_{u_{j-1}y}$ exists since $\Gamma_j$ satisfies Property~6. 

The proof that $P_{xy}$ is distance-decreasing in $\Gamma_{\delta}$ is the same as the one that $P_{xy}$ is distance-decreasing in $\Gamma_{\delta}$ when $x\in V(G_i)$, $y\in V(G_j)$, and $2\leq i <j \leq k$ in Lemma~\ref{le:case1B}, with $\beta$ in place of $s$ and $(-\alpha';\alpha')\subset (-\frac{\alpha}{8};\frac{\alpha}{8})$ in place of $(-\frac{\alpha}{3};\frac{\alpha}{3})$ as the interval of possible slopes for the edges of $P_x$. 

\item The case in which $1\leq j <i \leq k$ is symmetric to the previous one. 

\item Suppose now that $x$ belongs to $H$ and $y$ belongs to $G_i$, for some $i\in\{1,\dots,k\}$. If $i=1$ and $y=u_0$, then $y\in V(H)$ and $P_{xy}$ is defined as above. Assume hence that $y\neq u_0$. Then the path $P_{xy}$ consists of three sub-paths. 

\begin{itemize}
	\item The first sub-path of $P_{xy}$ is a path $P_x$ in $H$ from $x$ to $y_1$. 
	
	Suppose first that $x=u$. Then let $P_{x}=\tau_{uy_1}(H)$. Let $P_{x}=(x=z_1,z_2,\dots,z_s=y_1)$; we prove that $d(\Gamma_{\delta},z_hy)>d(\Gamma_{\delta},z_{h+1}y)$ holds true for any $h=1,\dots,s-1$. Consider the vertical line $\ell_1$ through $y_1$, oriented towards increasing $y$-coordinates; as argued above, the disk $D_{\rho}$ is to the right of $\ell_1$ and $y$ lies inside $D_{\rho}$. By Properties~1 and~2 of $\Gamma_{\delta}$, the edge $z_hz_{h+1}$ is horizontal, with $z_h$ to the left of $z_{h+1}$. Hence, the line $\ell'_{h}$ orthogonal to $z_hz_{h+1}$ and passing through its midpoint is also vertical and has $\ell_1$ to its right. It follows that $y$ is to the right of $\ell'_{h}$. Since the half-plane to the right of $\ell'_{h}$ represents the locus of the points of the plane that are closer to $z_{h+1}$ than to $z_h$, we have $d(\Gamma_{\delta},z_hy)>d(\Gamma_{\delta},z_{h+1}y)$.
	
	Suppose next that $x\neq u$. By Property~4 of $\Gamma_{H,\delta}$, there exists a path $P_{x}=(x=z_1,z_2,\dots,z_s=y_1)$ in $H$ that connects $x$ to $y_1$, that does not pass through $u$, and whose every edge has slope in $(-\frac{\alpha}{2};\frac{\alpha}{2})$ in $\Gamma_{H,\delta}$. We prove that, for any $h=1,2,\dots,s-1$, $d(\Gamma_{\delta},z_hy)>d(\Gamma_{\delta},z_{h+1}y)$; refer to Fig.~\ref{fi:inductiveB-distancedecreasing}. Since the drawing of $H$ in $\Gamma_{\delta}$ coincides with $\Gamma_{H,\delta}$, the edge $z_hz_{h+1}$ has slope in $(-\frac{\alpha}{2};\frac{\alpha}{2})$ in $\Gamma_{\delta}$. Consider the line $\ell_h$ that passes through $y_1$, that is directed towards increasing $y$-coordinates and that is orthogonal to the line through $z_h$ and $z_{h+1}$. Denote by $s_h$ the slope of $\ell_h$. Then $s_h\in (\frac{\pi-\alpha}{2};\frac{\pi+\alpha}{2})$. 
	
	\begin{figure}[htb]
		\centering
		\includegraphics[scale=0.6]{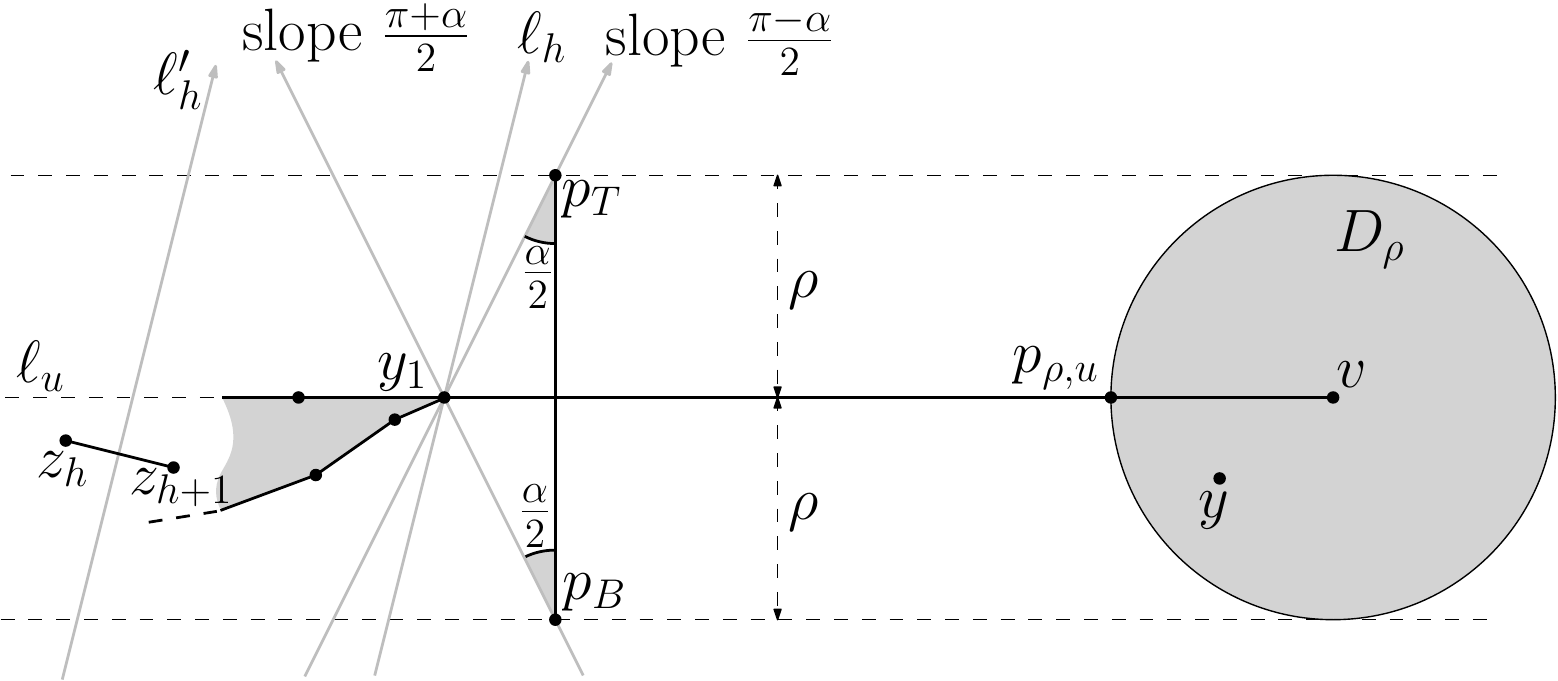}
		\caption{Illustration for the proof that $d(\Gamma_{\delta},z_hy)>d(\Gamma_{\delta},z_{h+1}y)$ if $z_hz_{h+1}$ is in $P_x$. For the sake of readability, $D_{\rho}$ is larger than it should be.}
		\label{fi:inductiveB-distancedecreasing}
	\end{figure}
			
	We prove that $\ell_h$ has the disk $D_{\rho}$ to its right. In order to do that, consider the point $p_T$ on the half-line with slope $\frac{\pi-\alpha}{2}$ starting at $y_1$ and such that $d_V(\Gamma_{\delta},y_1p_T)=\rho$. Further, consider the point $p_B$ on the half-line with slope $\frac{-\pi+\alpha}{2}$ starting at $y_1$ and such that $d_V(\Gamma_{\delta},y_1p_B)=\rho$. Note that $\overline{p_T p_B}$ is a vertical straight-line segment with length $2\rho$. Consider the infinite closed strip $S$ with height $2\rho$ that is delimited by the horizontal lines through $p_T$ and $p_B$. Since $D_{\rho}$ has its center on $\ell_u$ and has radius $\rho$, it lies inside $S$. The part of $\ell_h$ inside $S$ is to the left of $\overline{p_T p_B}$, given that $s_h\in (\frac{\pi-\alpha}{2};\frac{\pi+\alpha}{2})$. Hence, we only need to show that $p_{\rho,u}$, which is the point of $D_{\rho}$ with smallest $x$-coordinate, lies to the right of $\overline{p_T p_B}$. We have that $d(\Gamma_{\delta},y_1p_{\rho,u})=d_{y_1v}-\rho$. Further, $d_H(\Gamma_{\delta},y_1p_T)=\rho \cdot \tan(\frac{\alpha}{2})$. Hence, it suffices to prove $\rho \cdot \tan(\frac{\alpha}{2})<d_{y_1v}-\rho$, that is $\rho<\frac{d_{y_1v}}{1+ \tan(\frac{\alpha}{2})}$; this holds true since $\rho<\frac{d_{y_1v}}{3}$ and 
	$\tan(\frac{\alpha}{2})<1$, given that $0<\alpha<\frac{\pi}{4}$. 
	
	By Lemma~\ref{le:same-side}, the line $\ell_h$ has the drawing of $H$ (and in particular the midpoint of the edge $z_hz_{h+1}$) to its left; this is because by Property~2 of $\Gamma_{H,\delta}$ every edge in $\beta_{y_1 u}(H)$ has slope $\pi$, where $s_h<\frac{\pi+\alpha}{2}<\pi<\frac{3\pi-\alpha}{2}<\pi+s_h$, and because by Property~3 of $\Gamma_{H,\delta}$ every edge in $\tau_{y_1 u}(H)$ has slope in $(\pi-\frac{\alpha}{2};\pi+\frac{\alpha}{2})$, where $s_h<\frac{\pi+\alpha}{2}<\pi-\frac{\alpha}{2}<\pi+\frac{\alpha}{2}<\frac{3\pi-\alpha}{2}<\pi+s_h$. Now consider the line $\ell'_{h}$ parallel to $\ell_{h}$, passing through the midpoint of the edge $z_hz_{h+1}$, and oriented towards increasing $y$-coordinates; $\ell'_{h}$ has $\ell_{h}$ to its right, given that the midpoint of $z_hz_{h+1}$ is to the left of $\ell_{h}$ in $\Gamma_{\delta}$. Thus, $\ell'_{h}$ has $D_{\rho}$, and in particular $y$, to its right. Since the half-plane to the right of $\ell'_{h}$ represents the locus of the points of the plane that are closer to $z_{h+1}$ than to $z_h$, it follows that $d(\Gamma_{\delta},z_hy)>d(\Gamma_{\delta},z_{h+1}y)$.
			
	\item The second sub-path is the edge $y_1v$. Since $y$ lies in $D_{\rho}$, we have that $d(\Gamma_{\delta},vy)\leq \rho\leq \frac{d_{y_1v}}{3}$. By the triangular inequality, we have that $d(\Gamma_{\delta}, y_1y)>d(\Gamma_{\delta},y_1v)-d(\Gamma_{\delta},vy)\geq d_{y_1v}-\rho\geq \frac{2d_{y_1v}}{3}$. Hence, $d(\Gamma_{\delta}, y_1y)>d(\Gamma_{\delta},vy)$. 

	\item The third sub-path is a path $P_{vy}$ that connects $v$ to $y$, that belongs to $\bigcup_{l=i}^k G_l$, and that is distance-decreasing in $\Gamma_{\delta}$. This path exists, as from the case in which $x$ and $y$ belong to the same graph $G_i$ or from the case in which $x$ belongs to a graph $G_i$ and $y$ belongs to a graph $G_j$ for some $1\leq j <i \leq k$. 
\end{itemize}

\item Suppose finally that $x$ belongs to $G_i$, for some $i\in\{1,\dots,k\}$, and $y$ belongs to $H$. If $i=1$ and $x=u_0$, then $x\in V(H)$ and $P_{xy}$ is defined as above. Assume hence that $x\neq u_0$. We now describe the path $P_{xy}$, which consists of three sub-paths. 

\begin{itemize}
	\item The first sub-path of $P_{xy}$ is a path $Q_x$ in $G_i$ from $x$ to $u_{i-1}$ whose every edge has slope in $(\pi-\alpha'; \pi+\alpha')$ in $\Gamma_i$ (in $\Gamma_{1,d^*}$ if $i=1$). This path exists since $\Gamma_i$ ($\Gamma_{1,d^*}$ if $i=1$) satisfies Property~5. The second sub-path of $P_{xy}$ is $\bigcup_{j=1}^{i-1} \beta_{u_{j}u_{j-1}}(G_j)$. Since $\Gamma_j$ ($\Gamma_{1,d^*}$ when $j=1$) satisfies Properties~1 and~2, every edge in $\bigcup_{j=1}^{i-1} \beta_{u_{j}u_{j-1}}(G_j)$ has slope $\pi$ in $\Gamma_j$ (in $\Gamma_{1,d^*}$ when $j=1$). Let $(x=z_1,z_2,\dots,z_{s-1},z_s=y_\ell)$ be the union of these two sub-paths of $P_{xy}$. We prove that $d(\Gamma_{\delta},z_hy)>d(\Gamma_{\delta},z_{h+1}y)$, for any $h\in\{1,2,\dots,s-1\}$. Since the drawings $\Gamma_{1,d^*},\Gamma_2,\dots,\Gamma_k$ are counter-clockwise rotated by $\beta$ radians in $\Gamma_{\delta}$, it follows that $z_hz_{h+1}$ has slope in the interval $(\pi+\beta-\alpha'; \pi+\beta+\alpha')$ in $\Gamma_{\delta}$. 
	
	We first present a proof that $d(\Gamma_{\delta},z_hy)>d(\Gamma_{\delta},z_{h+1}y)$ for any $h\in\{1,2,\dots,s-2\}$; we will later argue that $d(\Gamma_{\delta},z_{s-1}y)>d(\Gamma_{\delta},z_s y)$. Recall that $z_h$ and $z_{h+1}$ lie in $D_{\rho}$ in $\Gamma_{\delta}$, given that $z_h,z_{h+1}\neq y_{\ell}$. Consider the line $\ell_h$ that passes through $y_1$, that is directed towards increasing $y$-coordinates and that is orthogonal to the line through $z_h$ and $z_{h+1}$. Denote by $s_h$ the slope of $\ell_h$. Then $s_h\in (\frac{\pi}{2}+\beta-\alpha'; \frac{\pi}{2}+\beta+\alpha')$. 
	
	Similarly to the case in which $x\in V(H)$ and $y\in V(G_i)$, we have that $\ell_h$ has the disk $D_{\rho}$ to its right and the drawing of $H$ to its left (the main difference is that the gray angles in Fig.~\ref{fi:inductiveB-distancedecreasing} are now $\beta+\alpha'$ rather than $\frac{\alpha}{2}$). We now present proofs for these statements.
	
	\begin{itemize}
		\item We prove that $\ell_h$ has $D_{\rho}$ to its right. Let $p_T$ ($p_B$) be the point on the half-line with slope $\frac{\pi}{2}-\beta-\alpha'$ (resp. $-\frac{\pi}{2}+\beta+\alpha'$) starting at $y_1$ and such that $d_V(\Gamma_{\delta},y_1p_T)=\rho$ (resp. $d_V(\Gamma_{\delta},y_1p_B)=\rho$). Then $\overline{p_T p_B}$ is a vertical straight-line segment with length $2\rho$ and $D_{\rho}$ lies inside the infinite closed strip $S$ with height $2\rho$ that is delimited by the horizontal lines through $p_T$ and $p_B$. The part of $\ell_h$ inside $S$ is to the left of $\overline{p_T p_B}$, since $s_h\in (\frac{\pi}{2}+\beta-\alpha'; \frac{\pi}{2}+\beta+\alpha')$. Hence, we only need to show that $p_{\rho,u}$ lies to the right of $\overline{p_T p_B}$.  We have $d(\Gamma_{\delta},y_1p_{\rho,u})=d_{y_1v}-\rho$, while $d_H(\Gamma_{\delta},y_1p_T)=\rho \cdot \tan(\beta+\alpha')$. Hence, it suffices to prove that $\rho<\frac{d_{y_1v}}{1+ \tan(\beta+\alpha')}$; this holds true since $\rho<\frac{d_{y_1v}}{3}$ and $\tan(\beta+\alpha')<1$, given that $0<\beta<\frac{\alpha}{4}<\frac{\pi}{16}$ and  $0<\alpha'<\frac{\alpha}{8}<\frac{\pi}{32}$. 
	
		\item By Lemma~\ref{le:same-side}, the line $\ell_h$ has $\Gamma_{H,\delta}$ (and in particular $y$) to its left; this is because by Property~2 of $\Gamma_{H,\delta}$ every edge in $\beta_{y_1 u}(H)$ has slope $\pi$, where $s_h<\frac{\pi}{2}+\beta+\alpha'<\pi<\frac{3\pi}{2}+\beta-\alpha'<\pi+s_h$, and because by Property~3 of $\Gamma_{H,\delta}$ every edge in $\tau_{y_1 u}(H)$ has slope in $(\pi-\frac{\alpha}{2};\pi+\frac{\alpha}{2})$, where $s_h<\frac{\pi}{2}+\beta+\alpha'<\pi-\frac{\alpha}{2}<\pi+\frac{\alpha}{2}<\frac{3\pi}{2}+\beta-\alpha'<\pi+s_h$. 
	\end{itemize}
	
	Now consider the line $\ell'_{h}$ parallel to $\ell_{h}$, passing through the midpoint of the edge $z_hz_{h+1}$, and oriented towards increasing $y$-coordinates; $\ell'_{h}$ has $\ell_{h}$ to its left, given that the midpoint of $z_hz_{h+1}$ is in $D_{\rho}$, hence to the right of $\ell_{h}$ in $\Gamma_{\delta}$. Thus, $\ell'_{h}$ has $\Gamma_{H,\delta}$ (and in particular $y$) to its left. Since the half-plane to the left of $\ell'_{h}$ represents the locus of the points of the plane that are closer to $z_{h+1}$ than to $z_h$, it follows that $d(\Gamma_{\delta},z_hy)>d(\Gamma_{\delta},z_{h+1}y)$.

	We now show that $d(\Gamma_{\delta},z_hy)>d(\Gamma_{\delta},z_{h+1} y)$ if $h=s-1$. Recall that $z_{h+1}=z_s=u_0=y_{\ell}$ and refer to Fig.~\ref{fi:inductiveB-distancedecreasing2}. We exploit again the fact that the line $\ell_h$ through $y_1$ orthogonal to the line through $z_{h}$ and $z_{h+1}$  has $\Gamma_{H,\delta}$ (and in particular $y$) to its left. Consider the line $\ell'_{h}$ parallel to $\ell_{h}$, oriented towards increasing $y$-coordinates, and passing through the midpoint $m_h$ of the edge $z_hz_{h+1}$. Differently from the case in which $h\in\{1,2,\dots,s-2\}$, the midpoint $m_h$ of $z_hz_{h+1}$ is not guaranteed to be in $D_{\rho}$ (in fact it is not in $D_{\rho}$, although we do not prove this statement formally as we do not need it in the remainder), given that $z_{h+1}=y_{\ell}$ is in $H$ and hence not in $D_{\rho}$. Since the half-plane to the left of $\ell'_{h}$ represents the locus of the points of the plane that are closer to $z_{h+1}$ than to $z_h$, we only need to show that the intersection point $p_h$ of the lines $\ell'_{h}$ and $\ell_u$ lies to the right of $y_1$ on $\ell_u$; in fact, this implies that $\ell'_{h}$ has $\ell_h$ (and hence $y$) to its left. 	
	
	\begin{figure}[htb]
		\centering
		\includegraphics[scale=0.6]{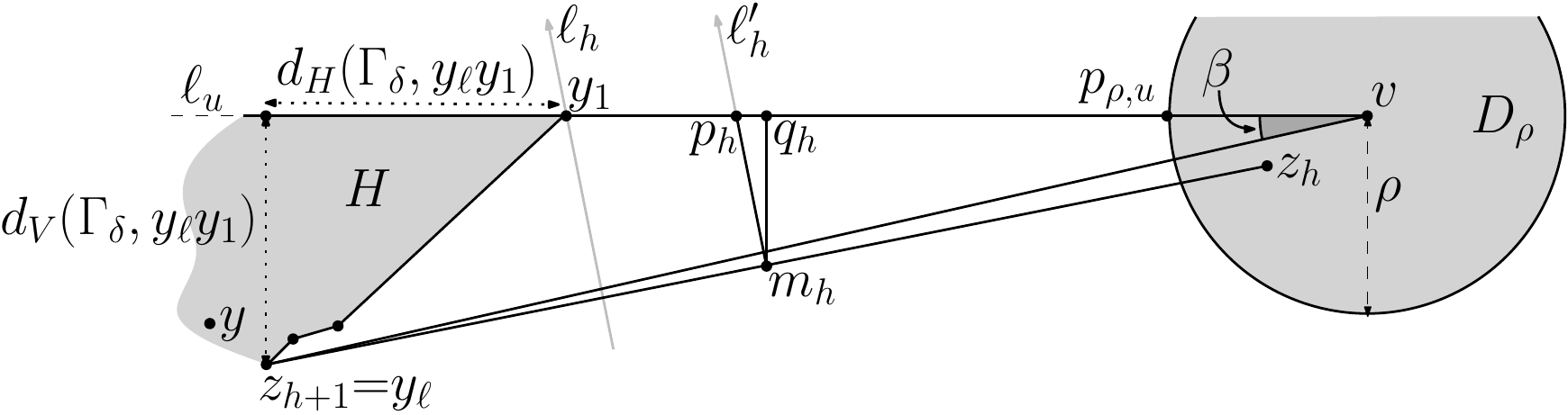}
		\caption{Illustration for the proof that $d(\Gamma_{\delta},z_hy)>d(\Gamma_{\delta},z_{h+1}y)$ if $h=s-1$.}
		\label{fi:inductiveB-distancedecreasing2}
	\end{figure}
	
	Since $z_h$ lies inside $D_{\rho}$, we have that $x(z_h)\geq x(p_{\rho,u})$. Further, $x(p_{\rho,u})=x(y_1)+d_{y_1v}-\rho$. Moreover, by Property~3 of $\Gamma_{H,\delta}$, we have that $x(y_1)=x(y_{\ell})+d_H(\Gamma_{\delta},y_\ell y_1)$. Thus, we have that $x(m_h)=\frac{x(y_\ell)+x(z_h)}{2}\geq \frac{x(y_\ell) + x(y_{\ell})+ d_H(\Gamma_{\delta},y_\ell y_1) + d_{y_1v}-\rho}{2}= x(y_\ell)+ \frac{d_H(\Gamma_{\delta},y_\ell y_1) + d_{y_1v}-\rho}{2}$. For the sake of the simplicity of the description, translate the Cartesian axes so that $x(y_{\ell})=0$. Thus, $x(m_h)\geq \frac{d_H(\Gamma_{\delta},y_\ell y_1) + d_{y_1v}-\rho}{2}$. 
	
	By Lemma~\ref{le:decomposition-B}, $y_{\ell}$ is an internal vertex of $\beta_{uv}(G)$, hence $y_{\ell}$ lies below $\ell_u$. Since $\rho\leq \frac{Y}{2}$ and $z_h$ lies in $D_{\rho}$, the $y$-coordinate of $y_{\ell}$ is smaller than the one of $z_h$. It follows that the slope of $z_hz_{h+1}$ is greater than $\pi$. Further, $z_h$ and hence $m_h$ lie on or below the line $h_{\beta}$ with slope $\beta$ through $y_{\ell}$. This implies that the slope of $z_hz_{h+1}$ is at most $\pi+\beta$. Thus, the slope $s'_h$ of $\ell'_h$ is in the interval $(\frac{\pi}{2};\frac{\pi}{2}+\beta)$. 
	
	We now derive a lower bound for the $x$-coordinate of $p_h$. Let $q_h$ be the point such that $x(q_h)=x(m_h)$ and $y(q_h)=y(p_h)$. Consider the triangle $\Delta m_h p_h q_h$. Since the $y$-coordinate of $y_\ell$ is smaller than the one of $z_h$, it is also smaller than the one of $m_h$. Thus, $d(\Gamma_{\delta},m_h q_h)\leq d_V(\Gamma_{\delta},y_\ell y_1)$. Since $s'_h\in (\frac{\pi}{2};\frac{\pi}{2}+\beta)$, the angle $\measuredangle p_h m_h q_h$ is at most $\beta$. Hence, $d(\Gamma_{\delta},p_h q_h)\leq d_V(\Gamma_{\delta},y_\ell y_1) \cdot \tan(\beta)$. It follows that $x(p_h)= x(m_h)-d(\Gamma_{\delta},p_h q_h)\geq \frac{d_H(\Gamma_{\delta},y_\ell y_1) + d_{y_1v}-\rho}{2} - d_V(\Gamma_{\delta},y_\ell y_1) \cdot \tan(\beta)$. It remains to prove that this quantity is larger than $d_H(\Gamma_{\delta},y_\ell y_1)$, which is the $x$-coordinate of $y_1$. 
	
	Since $\beta<\frac{\alpha}{4}<\frac{\pi}{16}$, we have that $\tan(\beta)\leq 1$. It follows that $\frac{d_H(\Gamma_{\delta},y_\ell y_1) + d_{y_1v}-\rho}{2} - d_V(\Gamma_{\delta},y_\ell y_1) \cdot \tan(\beta)\geq \frac{d_H(\Gamma_{\delta},y_\ell y_1) + d_{y_1v}-\rho}{2} - d_V(\Gamma_{\delta},y_\ell y_1)$. Hence, we want to establish that $\frac{d_H(\Gamma_{\delta},y_\ell y_1) + d_{y_1v}-\rho}{2} - d_V(\Gamma_{\delta},y_\ell y_1)>d_H(\Gamma_{\delta},y_\ell y_1)$, that is, $d_{y_1v}>2d_V(\Gamma_{\delta},y_\ell y_1)+d_H(\Gamma_{\delta},y_\ell y_1)+\rho$. Since $\rho\leq \frac{d_{y_1v}}{3}$, we need to prove that $d_{y_1v}>\frac{6d_V(\Gamma_{\delta},y_\ell y_1)+3d_H(\Gamma_{\delta},y_\ell y_1)}{2}$. 
	
	We now express $d_{y_1v}$ as a function of $\beta$. This is done by looking at the triangle whose vertices are $y_{\ell}$, $v$, and the point on $\ell_u$ with the same $x$-coordinate as $y_{\ell}$. Since the angle of this triangle at $v$ is $\beta$, we get that $d_{y_1v}=\frac{d_V(\Gamma_{\delta},y_\ell y_1)}{\tan(\beta)}-d_H(\Gamma_{\delta},y_\ell y_1)$. Substituting this into the previous inequality, we need to have $\frac{d_V(\Gamma_{\delta},y_\ell y_1)}{\tan(\beta)}-d_H(\Gamma_{\delta},y_\ell y_1)>\frac{6d_V(\Gamma_{\delta},y_\ell y_1)+3d_H(\Gamma_{\delta},y_\ell y_1)}{2}$, hence $\tan(\beta)<\frac{2d_V(\Gamma_{\delta},y_\ell y_1)}{6d_V(\Gamma_{\delta},y_\ell y_1)+5d_H(\Gamma_{\delta},y_\ell y_1)}$. This inequality holds true since $\beta<\arctan\left(\frac{d_V(\Gamma_H,y_{\ell}y_1)}{3d_V(\Gamma_H,y_{\ell}y_1)+3d_H(\Gamma_H,y_{\ell}y_1)}\right)$. This concludes the proof that $d(\Gamma_{\delta},z_hy)>d(\Gamma_{\delta},z_{h+1} y)$ if $h=s-1$.
	 
	\item The third sub-path of $P_{xy}$ is a path $P_{y_{\ell}y}$ that connects $y_{\ell}$ to $y$, that belongs to $H$, and that is distance-decreasing in $\Gamma_{H,\delta}$. This path exists since $\Gamma_{H,\delta}$ satisfies Property~6. Since the drawing of $H$ in $\Gamma_{\delta}$ coincides with $\Gamma_{H,\delta}$, the path $P_{y_{\ell}y}$ is also distance-decreasing in $\Gamma_{\delta}$.
\end{itemize}

\end{itemize}
This concludes the proof of the lemma.
\end{proof}

Given a strong circuit graph $(G,u,v)$ such that $G$ is not a single edge or a simple cycle, we are in Case~A or Case~B depending on whether the edge $uv$ exists or not, respectively. Thus, Lemmata~\ref{le:base-case}--\ref{le:caseB} prove Theorem~\ref{th:main-aux}. It remains to show how to use Theorem~\ref{th:main-aux} in order to prove Theorem~\ref{th:main}. This is easily done as follows. Consider any $3$-connected planar graph $G$ and associate any plane embedding to it; let $u$ and $v$ be two consecutive vertices in the clockwise order of the vertices along the outer face of $G$. We have that $(G,u,v)$ is a strong circuit graph. Indeed: (a) by assumption $G$ is $2$-connected -- in fact $3$-connected -- and associated with a plane embedding; (b) by construction $u$ and $v$ are two distinct external vertices of $G$; (c) edge $uv$ exists and coincides with $\tau_{uv}(G)$, given that $v$ immediately follows $u$ in the clockwise order of the vertices along the outer face of $G$; and (d) $G$ does not have any $2$-cut, given that it is $3$-connected. Thus, Theorem~\ref{th:main-aux} can be applied in order to construct a planar greedy drawing of $G$. This concludes the proof of Theorem~\ref{th:main}.

\section{Conclusions} 

In this paper we have shown how to construct planar greedy drawings of $3$-connected planar graphs. It is tempting to try to use the graph decomposition we employed in this paper for proving that $3$-connected planar graphs admit {\em convex} greedy drawings. However, despite some efforts in this direction, we have not been able to modify the statement of Theorem~\ref{th:main-aux} in order to guarantee the desired convexities of the angles in the drawings. Thus, proving or disproving the convex greedy embedding conjecture remains an elusive goal.

\bibliography{bibliography}

\end{document}